\documentclass[11pt]{article}

\usepackage[
  hmargin=\dimexpr1in-1pt\relax,
  vmargin=\dimexpr1in-3pt\relax
]{geometry}

\usepackage{etoolbox}

\usepackage{titlesec}
\titlespacing*{\section}      {0pt}{2.25ex plus 0.55ex minus 0.30ex}{1.40ex plus 0.25ex}
\titlespacing*{\subsection}   {0pt}{1.80ex plus 0.45ex minus 0.30ex}{0.50ex plus 0.12ex}
\titlespacing*{\subsubsection}{0pt}{1.40ex plus 0.35ex minus 0.20ex}{0.20ex plus 0.10ex}
\titleformat  {\paragraph}   [runin]{\normalfont\normalsize\bfseries}{\theparagraph}   {1em}{}
\titlespacing*{\paragraph}    {0pt}{0.70ex plus 0.15ex minus 0.10ex}{1em}
\titleformat  {\subparagraph}[runin]{\normalfont\normalsize\itshape} {\thesubparagraph}{1em}{}
\titlespacing*{\subparagraph} {0pt}{0.35ex plus 0.08ex minus 0.05ex}{1em}

\usepackage{setspace}
\usepackage[T1]{fontenc}
\usepackage{mathpazo}
\usepackage{microtype}

\usepackage{amsmath, amssymb, amsthm}
\usepackage{bm}
\usepackage{xfrac}
\AtBeginDocument{
  \setlength{\abovedisplayskip}{5pt plus 1pt minus 1pt}
  \setlength{\belowdisplayskip}{5pt plus 1pt minus 1pt}
  \setlength{\abovedisplayshortskip}{5pt plus 1pt minus 1pt}
  \setlength{\belowdisplayshortskip}{5pt plus 1pt minus 1pt}
}

\usepackage{booktabs}
\usepackage{threeparttable}
\AtBeginEnvironment{table}  {\setlength{\belowcaptionskip}{3pt}}
\AtBeginEnvironment{table*} {\setlength{\belowcaptionskip}{3pt}}
\AtBeginEnvironment{figure} {\setlength{\abovecaptionskip}{6pt}}
\AtBeginEnvironment{figure*}{\setlength{\abovecaptionskip}{6pt}}
\usepackage{placeins}
\usepackage{graphicx}
\usepackage{siunitx}
\graphicspath{{figures/}}

\usepackage{algorithm}
\usepackage{algpseudocode}
\usepackage{fvextra}

\usepackage{xcolor}
\usepackage{tikz}
\usetikzlibrary{arrows.meta,positioning,calc}
\definecolor{steelblue}{HTML}{4682B4}
\definecolor{negred}{HTML}{B55D60}

\definecolor{DeepGreen}{HTML}{228B22}   
\definecolor{LinkMaroon}{HTML}{950000}
\definecolor{ChadBlue}{rgb}{0.10, 0.10, 0.50}
\newcommand{\YES}{\textsc{Yes}}
\newcommand{\NO} {\textsc{No}}
\newcommand{\NOtoYES}{$m$-\NO-to-$(n\!-\!m)$-\YES}
\newcommand{\YEStoNO}{$m$-\YES-to-$(n\!-\!m)$-\NO}
\newcommand{\bid}{b}
\newcommand{\ask}{a}

\newcommand{\auditbox}[1]{\par\nopagebreak
  {\noindent\setlength{\fboxsep}{5pt}%
   \fbox{\parbox{\dimexpr\textwidth-2\fboxsep-2\fboxrule\relax}{\small #1}}\par}}
\newcommand{\medsum}{%
  \mathop{\scalebox{0.85}{$\displaystyle\sum$}}\displaylimits
}
\newcommand{\medstar}{\scalebox{1.25}{$\scriptstyle\star$}}
\newcommand{\asterisk}{\scalebox{1.25}{$\scriptstyle *$}}

\usepackage{enumitem}

\usepackage[round,authoryear]{natbib}
\usepackage{xurl}
\usepackage{hyperref}
\hypersetup{
    hidelinks,
    colorlinks=true,
    linkcolor=LinkMaroon,
    urlcolor=LinkMaroon,
    citecolor=LinkMaroon,
    anchorcolor=LinkMaroon,
    linktocpage=true
}

\theoremstyle{plain}

\newtheorem{proposition}{Proposition}

\theoremstyle{definition}
\newtheorem{definition}{Definition}

\newtheorem{remark}{Remark}

\title{
  Combinatorial Arbitrage in Prediction Markets\thanks{We thank the Yale Center for Research Computing
  and the TUM Chair of Mathematical Statistics
  for providing computational resources and support.}
}
\author{
  Hongjian Shi\thanks{Department of Mathematics, Technical University of Munich.
  Email: \texttt{hongjian.shi@tum.de}.}
  \and
  Daojing Zhai\thanks{Baruch College, City University of New York.
  Email: \texttt{daojing.zhai@baruch.cuny.edu}.}
}
\date{
  First version: September 15, 2026\\This version: {October 25, 2026}
}

\begin{document}
{
\setlength{\parskip}{0em}
\maketitle
\vspace{-1em}
}

\begin{abstract}
\noindent We develop a framework to identify and quantify arbitrage across logically related
prediction-market contracts.  We prove that prices outside the convex hull of admissible payoff
vectors admit arbitrage and identify optimal arbitrage portfolios through integer programming.
Using venue metadata and semantic analysis of resolution criteria, we recover
1.9 million logical relations and compile them into constraint systems covering
more than 255{,}000 contracts.
We replay full-depth Polymarket order books to simulate an arbitrageur executing these
portfolios against resting orders.  Among settled portfolios, an ideal arbitrageur that
takes the full displayed depth of every opportunity and reuses its capital at resolution
earns a risk-free return of 3.5\% a week after spreads and fees.
We find that arbitrage opportunities persist longer and attract less observed participation where
relation discovery and execution are more costly.
\end{abstract}

\vspace{1em}

\clearpage

\section{Introduction}
\label{sec:intro}

Prediction markets are emerging as a new asset class.
Like conventional financial assets, prediction contracts support price discovery, hedging,
and speculation.  Their market structure, however, creates a particular challenge for consistent
pricing.  Exchanges such as Polymarket and Kalshi list many contingent claims that resolve on the
same underlying events, so logical relations tie their payoffs together.
Each contract trades in its own order book, and the exchanges do not generally enforce the price
constraints implied by these relations.  This fragmentation can create 
\textit{combinatorial~arbitrage}:\linebreak a portfolio of related contracts that guarantees 
a profit under every admissible joint resolution.

This paper develops a unified framework to identify these opportunities and to measure their
executable size and profit.  The framework has three parts.
First, we mathematically characterize when the prices of logically related contracts admit arbitrage
and find the portfolio with the largest guaranteed profit by integer programming.
Second, we recover the logical relations among contracts from interpreting venue metadata and 
from analyzing resolution text using structured parsers and large language models (LLMs).
Third, we collect Polymarket's full-depth order books and simulate an arbitrageur who trades these
portfolios against displayed depth.

More specifically, we formalize the arbitrage opportunity by representing each admissible joint
resolution as a binary payoff vector.  Quoted prices are arbitrage-free if and only if they lie in
the coherent set, the convex hull of these payoff vectors.
When prices lie outside this set, the largest guaranteed profit per unit position equals their
$\ell_1$ distance from it.  With many related contracts, the payoff vectors are too numerous to
list.  We therefore encode logical relations as a linear constraint system and solve for the
optimal portfolio by integer programming.
Quoted prices give the direction of the trade but not its size.
We then scan the displayed depth at each price level of the order book to determine how much the
arbitrageur can trade after spreads and taker fees.

The empirical analysis covers Polymarket from 23 June to 3 August 2026 and combines three data
sources: roughly 745{,}000 contracts with their metadata and resolution texts retrieved from Polymarket's
own market catalog, 73.22 billion order-book messages, and 229.9 million on-chain fill \linebreak events.
The catalog tells us how contracts relate, the order books measure what an arbitrageur can trade,
and the fills show what traders did.  We recover the logical relations in three layers.
\emph{Venue-asserted} relations come from the groups of mutually exclusive contracts that Polymarket
designates itself and from the resol\-ution rules of these groups.
\emph{Venue-derived} relations follow from the match, contract type,
and line that the venue reports for sports contracts.
\emph{Semantic-inferred} relations are
read from contract text, by structured parsers or by large language models.  In these three layers,
we recover 1.9~million relations
and reduce them to about 570{,}000 relations covering more than 255{,}000 contracts; 
we compile each layer into a separate constraint system and evaluate
interacting relations jointly within it.  An invalid relation can create a false profit guarantee,
whereas omitting a valid relation can only hide an opportunity.
We therefore screen relations using checks appropriate to each source before admitting them to the
corresponding~system.

We report three findings.  First, combinatorial arbitrage survives spreads and fees.
We detect nearly 492{,}000 arbitrage windows, the intervals over which a portfolio has positive guaranteed
profit at the best quotes, and the arbitrageur can trade in about 125{,}000 of them.
In the windows whose contracts have all resolved, an ideal arbitrageur that
takes the full displayed depth of every opportunity \textls[-1]{and reuses its capital at resolution
accumulates \$2.50 million of guaranteed profit,} a risk-free return of 3.5\% a week 
on the \$11.1 million of capital the strategy requires at its peak.
The first execution in each window alone yields \$1.10 million in guaranteed profit.
Venue-asserted and venue-derived relations contribute 34\% and 59\% of the guaranteed profit,
and semantic-inferred relations the remaining 7\%.
These profits are robust to the execution assumptions.
With a 1-second order delay, the filled orders retain 79\% of the profit at resolution.
Returns on capital also change little when re-posted depth does not count as new liquidity and when the
arbitrageur fills only part of the displayed depth. \looseness=-1

Second, opportunities persist longer where the logical relation is harder to find.
Among the windows that last at least 0.3 seconds, the median duration is 2.9 seconds for venue-asserted
relations, 6.3 seconds for venue-derived relations, and about 31 seconds for semantic-inferred
relations. \linebreak The ordering follows the work each layer requires:
the venue publishes the first layer, the second layer must be reconstructed from match fields,
and the third layer requires reading and comparing resolution text.
The differences remain after controlling for portfolio characteristics,
the number of arbitrage windows the same component opened that day,
and day and hour fixed effects:
relative to venue-asserted windows, venue-derived windows last about 2 times as long and
semantic-inferred windows about 10 times as long.
Each additional leg in a portfolio is also associated with a longer duration.
The layers differ in their contracts and execution mechanisms as well as in the information they
require, so the evidence is consistent with both discovery and execution costs limiting arbitrage.

Third, actual trading follows the arbitrage direction more often when the venue makes the relation
explicit.  In 44\% of the executed windows based on venue-asserted relations,
at least one trader is observed to trade every leg of the portfolio in the arbitrage direction.
For executed windows based on venue-derived and semantic-inferred relations,
the percentage declines to 15\% and 9\%. \linebreak
Across all three layers, these trades arrive quickly:
the median lag from the opening of a window is about 2.6 seconds;
trading is also concentrated among a small number of wallets:
854 addresses make these trades, and the 10 most active ones account for 78\% of the windows
with such a trade.
The semantic-inferred layer---and in particular its LLM component,
whose relations require the most work to identify---is indeed
the one in which opportunities last longest and the fewest are taken. \looseness=-1

\paragraph{Related literature.}
First, the paper contributes to research on prediction markets.
This literature studies how prices aggregate information \citep{wolfers2004prediction},
how price discovery is shared across venues \citep{ng2026price}, how informed traders contribute to
forecast accuracy \citep{gomezcram2026accuracy}, and how adverse selection affects liquidity
provision \citep{bartlett2026adverse}.  A smaller set of papers examines prices across related
contracts.  \citet{sonnemann2013framing} document framing effects across equivalent partitions,
and \citet{rothschild2014misalignment} document misalignment among identical and logically related
contracts and test liquidity through field trades. \linebreak
\citet{saguillo2025arbitrage} use language models to identify dependencies between Polymarket
markets and measure arbitrage within markets and between pairs of markets with equivalent payoffs.
To our knowledge, this paper is the first to evaluate logical relations of different types jointly
in one constraint system and to measure the executable size and profit of the resulting portfolios
on full-depth order books.

Second, the paper relates to the finance literature on violations of the law of one price
\citep{lamont2003anomalies}.  Familiar cases include spreads between Treasury cash and derivative
markets \citep{fleckenstein2020renting,zhai2026treasuries} and deviations from covered interest
parity \citep{du2018deviations}.  These studies attribute price differences to market segmentation, \linebreak
funding costs, and balance-sheet constraints, and models of financially constrained arbitrageurs
explain why the differences persist \citep{gromb2002equilibrium}.
They take the payoff relation as known. \linebreak Prediction markets differ in two ways.
First, the relations are unknown and numerous: they must be discovered among hundreds of thousands
of contracts.  Second, the relations are richer than equivalence:
they include mutual exclusivity, partition, implication, and conjunction.
Both features make it challenging to identify arbitrage opportunities.

Third, the paper connects to the literature on combinatorial market design.
This literature studies how a market maker can keep prices coherent across related claims
\citep{hanson2003combinatorial,abernethy2013efficient}.
Computing coherent prices is hard when the claims are many \citep{chen2008complexity},
and \citet{dudik2012tractable} and \citet{kroer2016arbitrage} make it tractable with constraint
generation and integer programming.  Recent work recovers the relations between contracts from text:
\citet{gebele2026semantic} identify equivalent or subset-related contract pairs across venues,
and \citet{capponi2025agentic} discover semantic relations for statistical trading. \linebreak
We differ in two ways.  First, we take the trader's side: prices have already formed in separate
order books, and no market maker keeps them coherent.
Second, we connect relation recovery to arbitrage: \linebreak the recovered relations become the constraints of
the arbitrage problem, and we evaluate the result-\linebreak{}ing portfolios against full-depth order books.

Finally, the paper contributes to the literature on limits to arbitrage.
Mispricing can persist because arbitrageurs face noise-trader risk \citep{delong1990noise},
constrained capital \citep{shleifer1997limits}, and holding and transaction costs
\citep{pontiff2006costly}.  \citet{moskowitz2021sports} uses sports betting markets to test asset
pricing theories and finds momentum returns that fail to overcome transaction costs,
which prevents arbitrage from eliminating them.
The opportunities we measure remain after spreads and fees, so transaction costs alone do not
explain them.  Our evidence points to the cost of discovering the relation,
an information cost in the sense of \citet{grossman1980impossibility}:
prices adjust more slowly when information requires more complicated analysis
\citep{cohen2012complicated}. \linebreak Our findings are consistent with this evidence:
opportunities persist longest, and the fewest are taken, where the relation must be inferred from
contract text.  Execution adds a second cost.
Because the legs of a portfolio are traded in separate order books,
the arbitrageur bears the risk that one leg fills while another does not.

\paragraph{Road map.}
Section~\ref{sec:market} describes the prediction market structure and the data.
Section~\ref{sec:theory} develops the framework, from quoted prices to finite-depth order books.
Section~\ref{sec:relations} recovers the logical relation system in three layers.
Section~\ref{sec:empirical} measures the opportunities, their persistence,
and the trading around them.  Section~\ref{sec:conclusion} concludes the paper.
Appendix~\ref{app:theory} collects the proofs and derivations.
Appendix~\ref{app:data} documents data collection, relation system construction, and opportunity measurement. 

\section{Prediction Market Structure and Data}
\label{sec:market}

\enlargethispage*{.75\baselineskip}

\subsection{Market Structure and the Source of Combinatorial Arbitrage}
\label{sec:institution}
Polymarket and Kalshi saw their combined monthly notional trading volume rise from under \$5 billion
in mid-2025 to above \$50 billion in July 2026.\footnote{See
\href{https://www.pewresearch.org/short-reads/2026/09/23/prediction-markets-trading-volume-doubled-%
between-may-and-july-largely-driven-by-sports/}{Pew Research Center, September 23, 2026}.}
This growth has drawn retail investors alongside Wall Street hedge funds and professional
market makers, while prediction contracts are increasingly functioning as standardized financial
instruments for hedging event risk, discovering prices, and aggregating dispersed beliefs.

A prediction-market contract is a fully collateralized binary claim through which investors express
a view on an event governed by a prespecified resolution rule.
Each contract creates a \YES{} and a \NO{} outcome token backed jointly by one dollar:
one dollar \emph{splits} into one \YES{} token and one~\NO{} token, the pair \emph{merges} back into the dollar at
any time, and at resolution the winning token pays one dollar while the other pays zero.
Before resolution, the contract trades between zero and one in a central limit order book,
and the price of the \YES{} token is commonly interpreted as the market-implied probability of a
\YES{} resolution.

The two tokens are complementary representations of the same claim.
Although the venue exposes an order book under each token identifier, the two views describe the
same economic order book: a \NO{} bid at \(b^N\) is a \YES{} ask at \(1-b^N\),
while a \NO{} ask at \(a^N\) is a \YES{} bid at \(1-a^N\).
The matching engine crosses complementary orders by splitting a dollar into the token pair or
merging the pair,
so for a taker, selling \NO{} is equivalent to buying \YES{}, and buying \NO{} is equivalent to selling
\YES{}.\footnote{Throughout the paper, we therefore express all positions and quotes in
\YES{}-equivalent units and work with a single canonical \YES{}-side order book.} 
The fee structure of Polymarket also respects the same complementarity:
whether a taker buys or sells, \linebreak the same fee formula applies---the taker fee per share 
is proportional to the product of the share price and its complement.

One distinctive feature of prediction markets is that a single real-world event can support multi\-ple
contracts.  A football match, for example, may generate contracts on the winner, the exact score,
the total number of goals, and combinations of these outcomes.
Because these contracts depend on \linebreak the same underlying realized match outcome, their terminal payoffs must
satisfy the constraints imposed by their logical relations.
This common resolution is not reflected in the trading architecture:
each contract has its own order book, and the matching engine neither enforces coherence across order books 
nor supports atomic execution of multi-leg orders across them.
Prices can therefore \linebreak appear reasonable contract by contract
yet cannot be rationalized jointly by any single belief,
creating combinatorial arbitrage opportunities.
Venues provide little direct enforcement of these constraints, so restoring coherence across order
books is left to arbitrageurs who trade against the mispricing.\footnote{\label{fn:one-exception}The one exception is
Polymarket's
\href{https://docs.polymarket.com/concepts/negative-risk}{negative-risk}
mechanism (Kalshi's
\href{https://help.kalshi.com/en/articles/13823816-collateral-return}{collateral-return}
facility plays a similar role): for venue-designated groups of mutually exclusive contracts,
one dollar of collateral \emph{mints} one \YES{} token of every contract in the group
(a split followed by a \emph{conversion} of the \NO{} token),
so selling the whole group is cheap to execute; buying it still takes one order per contract.
Section~\ref{sec:limits} shows that arbitrage opportunities inside these groups close most quickly.} \looseness=-1

\subsection{Data: Polymarket at Full Depth}
\label{sec:data}

To examine arbitrage opportunities in the prediction market, we assemble a full-depth record of
Polymarket along three dimensions: which contracts are listed (market catalog),
how each contract is quoted (full-depth order books), and what trades take place (real
transactions).\footnote{We focus on Polymarket because all three dimensions are publicly observable: 
quotes stream at full depth, the listing-\linebreak{}time resolution question is recoverable on chain,
and every trade settles on Polygon.
On centrally cleared venues such as Kalshi, settlement occurs inside the exchange's clearinghouse,
so published trades cannot be linked to individual accounts. \looseness=-1}
The record covers 23 June to 3 August 2026, and over these 42 days we observe roughly 745{,}000
binary contracts, 73.22 billion order-book messages, and 229.9 million on-chain fill events resolving
to 145.16 million observed trades that carry \$5.23 billion of volume.
Table~\ref{tab:data-sources} summarizes the three parts of the dataset. 

\paragraph{Market catalog.}
We retrieve the venue's listing catalog and reconstruct its hierarchy from event pages down to
binary conditions and outcome tokens.  For each contract we observe its identifiers, fee schedule,
negative-risk designation, and the displayed resolution text preserved in our catalog snapshot.
The hierarchy tells us which contracts can share an underlying state, and the criteria text is the
raw material from which Section~\ref{sec:relations} recovers their precise logical relations.
Relations are audited against this preserved catalog text rather than against the listing-time
\texttt{QuestionInitialized} record on chain. 
The catalog covers roughly 745{,}000 market records on 178{,}000 event pages and identifies 20{,}866
negative-risk groups (Table~\ref{tab:data-sources}).
Appendix~\ref{app:metadata} details its construction.

\paragraph{Full-depth order books.}
Polymarket provides no retrospective depth feed, so we operated a live collector throughout the
sample, recording venue-wide book snapshots and price-level updates with exchange and receipt
timestamps.  Full depth makes execution measurable.
Across valid two-sided quote states, the pooled median spread is 8.0 cents,
and median depth is 40 shares at the best bid and 30 shares at the best ask.
Orders exceeding the quantity at the best quote execute at additional price levels.
Best quotes alone therefore do not determine their execution costs. 
The full-depth order-book \linebreak data contain 73.22 billion messages, reduced to about 3.6 billion valid two-sided
quote states for venue-wide detection; executable size is always measured on the full order books
reconstructed for each candidate window.  Appendices~\ref{app:collector}
and~\ref{app:reconstruction} detail the collection and the~reconstruction.

\begin{table}[!t]
\centering
\caption{Summary statistics of the Polymarket record}
\label{tab:data-sources}
\begin{threeparttable}
\footnotesize
\setlength{\tabcolsep}{3pt}
\renewcommand{\arraystretch}{0.975}
\begin{tabular*}{\textwidth}{@{\extracolsep{\fill}}l!{\hspace{15pt}}rrrrrrr@{}}
\toprule
 & Sports & Crypto & Politics & Finance & Weather & Other & All \\
\midrule
\multicolumn{8}{@{}l}{\textit{Panel A: Market catalog (daily)}} \\
\addlinespace[2.5pt]
Active contracts & 50{,}603 & 8{,}366 & 4{,}077 & 2{,}727 & 2{,}061 & 1{,}865 & 69{,}698 \\
Active event pages & 5{,}033 & 6{,}042 & 1{,}238 & 497 & 230 & 314 & 13{,}355 \\
\midrule
\multicolumn{8}{@{}l}{\textit{Panel B: Full-depth order books (daily)}} \\
\addlinespace[2.5pt]
Quote states & 30{,}471{,}793 & 12{,}805{,}233 & 11{,}245{,}638 & 9{,}180{,}315 & 13{,}542{,}637 & 8{,}163{,}812 & 85{,}409{,}429 \\
Active order books & 90{,}388 & 15{,}524 & 7{,}449 & 5{,}149 & 3{,}388 & 3{,}335 & 125{,}233 \\
Bid--ask spread (cents) & 31.54 & 4.74 & 8.20 & 14.15 & 3.97 & 12.65 & 9.42 \\
\midrule
\multicolumn{8}{@{}l}{\textit{Panel C: Real transactions (daily)}} \\
\addlinespace[2.5pt]
Trading volume (USD) & 88{,}886{,}882 & 18{,}580{,}218 & 8{,}401{,}972 & 1{,}043{,}160 & 2{,}009{,}231 & 578{,}857 & 124{,}558{,}668 \\
Observed trades & 786{,}286 & 2{,}276{,}899 & 119{,}357 & 27{,}774 & 193{,}715 & 22{,}118 & 3{,}456{,}169 \\
Trader addresses & 50{,}143 & 26{,}848 & 19{,}164 & 5{,}435 & 7{,}435 & 4{,}332 & 94{,}983 \\
\bottomrule
\end{tabular*}
\begin{tablenotes}[flushleft]
\footnotesize
\item\!\textit{Notes:} Daily means over 23 June--3 August 2026 (42 days). Other pools technology, culture and untagged
contracts. Active contracts and event pages are those with a quoted book that day; the spread is the
within-day median over valid two-sided quote states; volume counts each observed trade once. The class columns cover the contracts in the catalog snapshot. \looseness=-1
\end{tablenotes}
\end{threeparttable}
\vspace{-5pt}
\end{table}

\paragraph{Real transactions.}
We decode every \texttt{OrderFilled} settlement event emitted by 
Polymarket's two (regular and negative-risk) exchange smart contracts
on Polygon over our sample period. We refer to these records as \emph{On-Chain Fills}. 
Each~fill reports the traded token, side, price, quantity, fee, 
maker and taker identifiers, transaction hash and log index, and block timestamp;
the token identifiers \linebreak link each~fill back to the catalog and the reconstructed order
books.  These records show whether actual trades follow the detected portfolios.
Section~\ref{sec:takers} matches fills to their contracts and trading directions,
while Sections~\ref{sec:implementation}--\ref{sec:limits} measure opportunities at displayed quotes.
The transaction data contain 229.9 million \texttt{OrderFilled} events, which resolve to 145.16
million observed trades carrying \$5.23 billion of volume among 679{,}000 distinct trader addresses
(Table~\ref{tab:data-sources}).  The class columns of the table cover the trades on contracts in
the catalog snapshot, 99.1 percent of the trades and 95.9~percent of the volume;  the remainder is
on contracts outside the snapshot, chiefly short-lived crypto and in-play sports markets that had
ended before it opens. Appendix~\ref{app:onchain-fills} details the retrieval. 

\FloatBarrier
\section{A Framework for Combinatorial Arbitrage}
\label{sec:theory}

Building on the institutional primitives in Section~\ref{sec:market}, this section develops a
framework for combinatorial arbitrage in three steps.
Section~\ref{sec:coherent} characterizes arbitrage as a failure of price coherence and measures its
profit as a distance.  Section~\ref{sec:ip} compiles the logical relations into a constraint system
and solves for the optimal portfolio by integer programming.
Section~\ref{sec:executable} then replaces quoted prices with order books, where 
executable size and guaranteed profit depend on spreads, depth, and fees. \looseness=-1

\subsection{Characterizing Arbitrage through Price Coherence}
\label{sec:coherent}
We begin with a single quoted price per contract, isolating the payoff constraints imposed by
logical relations.  The no-arbitrage characterization follows the Dutch-book argument
of~\citet{definetti1937prevision}.

\paragraph{No-arbitrage characterization.}
Let $\Omega$ be a finite state space, with $\omega\in\Omega$ denoting a state of the world.
A finite collection of binary \YES{} contracts, indexed by $i \in I$, trades at quoted prices
$\bm{p} = (p_i)_{i \in I} \in [0,1]^{I}$, and the payoff map $\bm{\phi} : \Omega \to \{0,1\}^{I}$
records the amount $\phi_i(\omega)$ that one share of contract $i$ pays at resolution in state
$\omega$.  We maintain binary settlement in the modeled state space.
Section~\ref{sec:institution} establishes the \YES{}--\NO{} payoff equivalence,
so we express all positions and quotes in \YES{}-equivalent units.
A \emph{portfolio} $\bm{y} \in \mathbb{R}^{I}$ has initial cost $\bm{p}^\top\bm{y}$ and terminal
payoff $\bm{y}^\top\bm{\phi}(\omega)$; negative components of $\bm{y}$ are short positions,
and a negative initial cost means cash is received at inception.
The portfolio is held until resolution, at which point its terminal payoff is realized.\footnote{The
empirical guarantees are conditional on binary settlement, so Section~\ref{sec:implementation}
evaluates only windows whose filled legs resolve to zero or one.
In the quoted-price benchmark, the portfolio can be unwound before resolution once quotes return to
the coherent set.  With finite-depth order books, however, future liquidation spreads, depth,
and fees are unknown when the portfolio is formed, so early-exit profit cannot be guaranteed ex
ante.  Appendix~\ref{app:early-liquidation} considers this possibility; \linebreak
throughout the empirical analysis, we therefore hold portfolios to resolution.}
Definition~\ref{def:arb} and Proposition~\ref{prop:coherence} formalize arbitrage and price
coherence.  Proposition~\ref{prop:l1} characterizes the largest guaranteed profit and an optimal
portfolio.

\begin{definition}[Arbitrage]
\label{def:arb}
A portfolio $\bm{y} \in \mathbb{R}^{I}$ is an \emph{arbitrage} at prices $\bm{p}$ if its
\emph{guaranteed profit}, the minimum profit across states, is strictly positive:
\begin{equation}
\label{eq:arb}
\min_{\omega \in \Omega}\bm{y}^\top \bm{\phi}(\omega)\,-\,\bm{p}^\top\bm{y} > 0,
\end{equation}
that is, if its \emph{worst-case terminal payoff} $\min_{\omega}\bm{y}^\top\bm{\phi}(\omega)$
exceeds its initial cost.
\end{definition}

\begin{proposition}[No arbitrage $\iff$ coherent prices]
\label{prop:coherence}
Let $M := \operatorname{conv}\{\bm{\phi}(\omega) : \omega \in \Omega\}$ denote the convex hull of
the payoff vectors across all states.
No arbitrage exists at prices $\bm{p}$ if and only if $\bm{p} \in M$.
\end{proposition}

\begin{proposition}[Guaranteed profit as $\ell_1$ distance]
\label{prop:l1}
The worst-case terminal payoff of every portfolio $\bm{y}$ is equal to $\min_{\bm{\mu}\in M}\bm{y}^\top\bm{\mu}$.
Under the unit-position normalization $\|\bm{y}\|_\infty\le1$, the largest profit any portfolio can
guarantee equals the $\ell_1$ distance from the quote vector to the coherent set,
\begin{align}
\label{eq:l1-distance}
  \max_{\|\bm{y}\|_\infty\le1}\;\min_{\omega\in\Omega}\big\{\bm{y}^\top\bm{\phi}(\omega)-%
  \bm{p}^\top\bm{y}\big\}
=\min_{\bm{\mu}\in M}\|\bm{\mu}-\bm{p}\|_1,
\end{align}
and an optimal portfolio buys every contract priced below the minimizing coherent price vector
$\bm{\mu}^{\medstar}$ and sells every contract priced above it.
\end{proposition}

\vspace{-.5em}

Proofs of both propositions are in Appendix~\ref{app:proofs}.
The two propositions carry the framework's economic content.
$M$ collects the expected payoff vectors generated by all beliefs over states:
each point of $M$ is $\mathbb{E}_q[\bm{\phi}]$ for some belief $q$, so $\bm{p}\in M$ means that one
probability distribution prices every contract at its expected payoff,
and the quotes reflect a single coherent belief.
If $\bm{p}$ falls outside $M$, no belief rationalizes the quotes and an arbitrage exists;
the best portfolio guarantees exactly the $\ell_1$ distance from the quotes to the coherent set,
and the separating direction $\bm{y}$ is the trade itself.

\begin{figure}[!t]
\centering
\includegraphics[width=\textwidth]{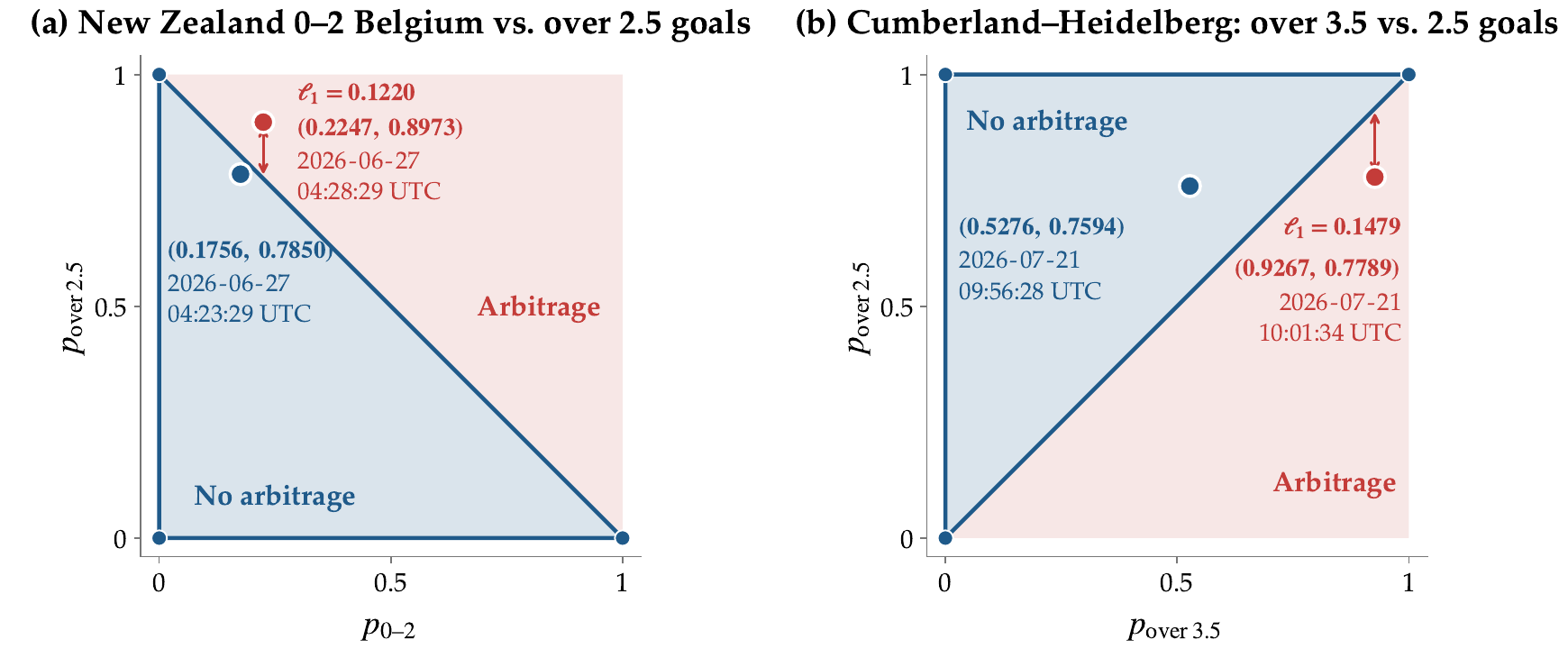}%
\vspace{-3pt}
\caption{Price coherence in two football contract pairs, June and July 2026.
Panels (a) and (b) show the New Zealand--Belgium exact score 0--2 versus over 2.5 goals and the
Cumberland--Heidelberg over 3.5 goals versus over 2.5 goals;
arrows mark the $\ell_1$ distance to coherence.}
\label{fig:sell-buy-asymmetry}
\vspace{-5pt}
\end{figure}

Figure~\ref{fig:sell-buy-asymmetry}
illustrates two contract pairs of different types drawn from the data.
Panel (a) shows two mutually exclusive (not collectively exhaustive) contracts
``\href{https://polymarket.com/sports/world-cup/fifwc-nzl-bel-2026-06-26}%
{Exact Score: New Zealand 0--2 Belgium?}'' \linebreak
and
``\href{https://polymarket.com/sports/world-cup/fifwc-nzl-bel-2026-06-26}%
{New Zealand vs. Belgium: Over 2.5 goals}'':
a 0--2 final score has two goals and cannot coincide with more than 2.5 goals in total,
so the only admissible payoff vectors are $(0,0)$, $(1,0)$, and~$(0,1)$, \linebreak 
whose convex hull is the blue triangle.
The fee-adjusted bid vector $(0.1756,0.7850)$ lies inside it;
five minutes later, the fee-adjusted bid vector $(0.2247,0.8973)$ lies outside the no-arbitrage
region at an $\ell_1$ distance of $0.1220$ from it, 
so selling one share of each contract guarantees \$0.1220 after fees. \linebreak
Panel (b) shows two nested~contracts from an Australia Cup match,
``\href{https://polymarket.com/event/auc-cum-hei-2026-07-21}%
{Cumberland United vs. Heidelberg United: Over 3.5 goals}''
and
``\href{https://polymarket.com/event/auc-cum-hei-2026-07-21}%
{Cumberland United vs. Heidelberg United: Over 2.5 goals}'':
more than 3.5 goals are also more than 2.5 goals, so the admissible payoff vectors $(0,0)$, $(0,1)$,
and $(1,1)$ \linebreak place the coherent set above the diagonal.
The fee-adjusted bid for the over-3.5 contract and ask for the over-2.5 contract form the execution
quote vector:
$(0.5276,0.7594)$ is coherent, whereas five minutes later $(0.9267,0.7789)$ lies $0.1479$
below the diagonal, so selling the over-3.5 contract and buying the over-2.5 contract guarantees
\$0.1479 after fees.  Figure~\ref{fig:additional-relations} in Appendix~\ref{app:logical-examples}
draws the same geometry for partition and conjunction relations.

\subsection{Detecting Combinatorial Arbitrage via Integer Programming}
\label{sec:ip}
When only two or three contracts are involved, the corresponding arbitrage can be easily identified
geometrically, as illustrated in Figure~\ref{fig:sell-buy-asymmetry} for two contracts and
Figure~\ref{fig:additional-relations} for three.  As the number of related contracts grows, however,
this geometric characterization becomes less tractable. 
We therefore formulate the problem as an integer program
\citep{schrijver1986theory,nemhauser1988integer}. \looseness=-1

\enlargethispage*{.5\baselineskip}

The program must first determine a given portfolio's lowest payoff across all payoff vectors
admitted by the logical relation system.  Let $\bm{z}\in\{0,1\}^{I}$ denote a joint payoff vector,
where $z_i=1$ exactly when contract $i$ resolves \YES{}; there are up to $2^{|I|}$ possible binary
payoff vectors.  Each logical relation excludes some vectors and compiles into one or more linear constraint
rows, as in Table~\ref{tab:relation-templates}. \linebreak Stacking these constraint rows in matrix form as $F\bm{z}\le\bm{g}$
defines the \emph{constraint system} $(I,F,\bm{g})$\footnote{%
When no confusion arises, we write $(F,\bm{g})$ or simply $F$ for this constraint system.}
and yields the \emph{admissible payoff set} $\Phi:=\{\bm{z}\in\{0,1\}^{I}:F\bm{z}\leq\bm{g}\}$,
which coincides with $\{\bm{\phi}(\omega):\omega\in\Omega\}$ when the compiled constraint system
exactly characterizes the payoff set induced by the state space.
In that case, the coherent set $M$ of Section~\ref{sec:coherent} is given by $\operatorname{conv}(\Phi)$,
the integer hull (the convex hull of all integer points) of the LP-relaxation polytope
$\Phi_{\mathrm{frac}}\!:=\!\{\bm{z}\in[0,1]^{I}:F\bm{z}\leq\bm{g}\}$.
\looseness=-1

The system in fact changes over time: the admissible payoff set is
$\Phi_t:=\{\bm{z}\in\{0,1\}^{I_t}: F_t\bm{z}\leq\bm{g}_t\}$ \linebreak at instant $t$,
where $I_t$ denotes the set of contracts
listed at $t$ and $F_t\bm{z}\le\bm{g}_t$ stacks the rows compiled from the relations in effect at
$t$, that is, the relations whose contracts are all listed at $t$.
We take the time-aggregated system $(I_*,F_*,\bm{g}_*)$ to be a reduced system such that,
for every $t$, $I_t\subseteq I_*$ and the rows of $F_*\bm{z}\le\bm{g}_*$ in effect at $t$ form a
subset of the rows of $F_t\bm{z}\le\bm{g}_t$ that implies the remaining rows for every binary
$\bm{z}$, so that these rows alone define $\Phi_t$.
Every proposition is stated for a general~$(I,F,\bm{g})$ and therefore holds at each $t$.
Detection at $t$ uses $(I_t,F_t,\bm{g}_t)$, and Section~\ref{sec:relations}
describes the aggregate $F_*$. \looseness=-1

\begin{table}[!t]
\centering
\caption{Logical relations, their constraint rows, and the licensed trades}
\label{tab:relation-templates}
\begin{threeparttable}
\footnotesize
\setlength{\tabcolsep}{3pt}
\renewcommand{\arraystretch}{1.2}
\begin{tabular*}{\textwidth}{@{\extracolsep{\fill}}
                             >{\raggedright\arraybackslash}p{0.12\textwidth}
                             >{\raggedright\arraybackslash}p{0.16\textwidth}
                             >{\raggedright\arraybackslash}p{0.17\textwidth}
                             >{\raggedright\arraybackslash}p{0.17\textwidth}
                             >{\raggedright\arraybackslash}p{0.28\textwidth}
                             @{}}
\toprule
Relation & Meaning & Constraints & Violation at quotes & Licensed trade \\
\midrule
\textbf{Mutual exclusivity} \newline of $A$ and $B$
  & At most one of the two contracts pays
  & $z_A+z_B\le 1$\footnotemark\xdef\pairwiseFn{\thefootnote}
  & $\tilde b_A+\tilde b_B>1$
  & Sell one share of $A$ and one of $B$ \\
\addlinespace[10pt]
\textbf{Partition} \newline of $A_1,\dots,A_n$
  & Exactly one of the \newline $n$ contracts pays
  & $z_{A_1}+\dots+z_{A_n}\le 1$ \newline $z_{A_1}+\dots+z_{A_n}\ge 1$
  & $\sum_i\tilde b_{A_i}>1$ \newline $\sum_i\tilde a_{A_i}<1$
  & Sell one share of every contract \newline Buy one share of every contract \\
\addlinespace[10pt]
\textbf{Implication} \newline $A$ implies $B$
  & Whenever $A$ pays, $B$ pays
  & $z_A\le z_B$
  & $\tilde b_A>\tilde a_B$
  & Sell one share of $A$, buy one of $B$ \\
\addlinespace[10pt]
\textbf{Conjunction} \newline $C=A\wedge B$
  & $C$ pays exactly when both \newline $A$ and $B$ pay
  & $z_C\le z_A$ \newline $z_C\le z_B$ \newline $z_C\ge z_A+z_B-1$
  & $\tilde b_C>\tilde a_A$ (or $\tilde a_B$) \newline $\tilde a_C<\tilde b_A+\tilde b_B-1$
  & Sell $C$, buy $A$ (or $B$) \newline Buy $C$, sell $A$ and $B$ \\
\bottomrule
\end{tabular*}
\begin{tablenotes}[flushleft]
\footnotesize
\item\!\textit{Notes:} $z_i\in\{0,1\}$ is the payoff of contract $i$, and $\tilde b_i$ and $\tilde a_i$ are its
fee-adjusted best bid and best ask.
\end{tablenotes}
\end{threeparttable}
\end{table}
\footnotetext[\pairwiseFn]{\hypertarget{Hfootnote.\pairwiseFn}{}\phantomsection\label{fn:pairwise}%
For binary payoffs $z_i$, $i\in I$,
mutual exclusivity of every pair (pairwise exclusivity) is equivalent to
mutual exclusivity of the whole set (\;\!$\sum_{i\in I} z_i \leq 1$);
see the proof in Appendix~\ref{app:proofs}.}

Given the compiled constraint system $(I,F,\bm{g})$, 
because payoffs depend on the state only through its payoff vector, 
the minimum over the finite set $\Phi=\{\bm{z}\in\{0,1\}^{I}:F\bm{z}\leq\bm{g}\}$
is attained and 
recovers exactly the minimum over $\Omega$ in Definition~\ref{def:arb}. 
Thus the largest guaranteed profit can be written as \linebreak
\(
\label{eq:largest-guaranteed-profit}
\max_{\|\bm{y}\|_\infty\le1}\{(\min_{\bm{z}\in\Phi}\bm{y}^\top\bm{z})-\bm{p}^\top\bm{y}\}.
\)
Arbitrage detection then proceeds in two nested steps. 

\paragraph{Inner problem.}
Fix a proposed portfolio $\bm{y}$. Finding an admissible payoff vector
under which this portfolio pays least is the following integer program:
\begin{equation}
\label{eq:worst-case-payoff}
V(\bm{y})
:=
\min_{\bm{z}}\ \bm{y}^\top\bm{z}
\quad
\text{subject to }
F\bm{z}\leq\bm{g},\qquad
\bm{z}\in\{0,1\}^{I}.
\end{equation}
The portfolio $\bm{y}$ is an arbitrage if and only if its worst-case terminal payoff
$V(\bm{y})$ is larger than $\bm{p}^\top\bm{y}$. 

\paragraph{Outer problem.}
The portfolio $\bm{y}$ is not fixed in advance.
We now optimize over the portfolio itself, searching over normalized portfolios for the one with the
largest guaranteed profit:
\begin{equation}
\label{eq:detector-maxmin}
\pi^{\medstar}
=
\max_{\|\bm{y}\|_\infty\leq1}
\big\{V(\bm{y})-\bm{p}^\top\bm{y}\big\}.
\end{equation}
A positive $\pi^{\medstar}$ identifies an arbitrage, the optimizer $\bm{y}^{\medstar}$ specifies the trade;
the value $\pi^{\medstar}$ equals the $\ell_1$~distance of Proposition~\ref{prop:l1}.

\paragraph{Solving the maximin problem by cutting planes.}
To see how the two problems fit together, let $\eta$ be a portfolio's claimed worst-case
terminal payoff.  Each admissible payoff vector $\bm{z}$ imposes the inequality
$\bm{z}^\top\bm{y}\ge \eta$; imposing this inequality for every $\bm{z}\in\Phi$ and maximizing
$\eta-\bm{p}^\top\bm{y}$ reproduces the outer problem in~\eqref{eq:detector-maxmin}.
Direct evaluation would therefore require enumerating all admissible payoff vectors in $\Phi\subseteq\{0,1\}^{I}$,
whose number can grow exponentially with the number of related contracts.\footnote{Most detected
opportunities involve only a few contracts, but the detector is evaluated at every relevant
best-quote update.  A six-contract connected component has $2^6=64$ possible binary payoff vectors;
at the scale of about 3.59 billion best-quote states in our sample, repeating even this small
enumeration would entail about 230 billion payoff-vector evaluations.
Even small payoff sets can therefore be expensive to enumerate repeatedly across the full sample.} \looseness=-1

Algorithm~\ref{alg:detector} may help avoid exhaustive enumeration:
start with a small working set $\Phi_{\mathrm{work}}\subseteq \Phi$, \linebreak
initialized here as a singleton $\{\bm{z}_0\}$.  First, the \textsc{Master} step chooses a portfolio $\bm{y}$ and its claimed
payoff~$\eta$ that maximize~$\eta-\bm{p}^\top\bm{y}$, subject to the condition that
\(\bm{z}^\top\bm{y}\ge \eta\) holds for every \(\bm{z}\in\Phi_{\mathrm{work}}\). \linebreak
The \textsc{Oracle} step then solves the inner problem in \eqref{eq:worst-case-payoff} and
identifies an admissible payoff vector~$\bm{z}_{\mathrm{new}}$ under which the current portfolio pays least.
If $\bm{y}^\top\bm{z}_{\mathrm{new}}$ is less than $\eta$, 
the \textsc{Master}'s claim fails and $\bm{z}_{\mathrm{new}}$ enters the working set;
otherwise the claimed payoff holds over all of $\Phi$, and the current port\-folio solves the full
problem.  Each failed claim adds a new vector from the finite set $\Phi$,
ensuring finite termination; the worst case may still iterate over every vector.
Algorithm~\ref{alg:detector} is thus the cutting-plane \linebreak method of \citet{kelley1960cutting} adapted
to the concave piecewise-linear function $V$ in \eqref{eq:worst-case-payoff}, with the integer program in
\eqref{eq:worst-case-payoff} acting as the separation oracle in the sense of
\citet{grotschel1988geometric}. 

\vspace{-2.5pt}
\begin{algorithm}[!htbp]
\caption{Cutting-plane arbitrage detector}
\label{alg:detector}
\begin{algorithmic}[1]
\Require quote vector $\bm{p}$; constraints $(F,\bm{g})$; any admissible payoff vector $\bm{z}_0$
\State $\Phi_{\mathrm{work}}\gets\{\bm{z}_0\}$
\Loop
  \State $(\bm{y},\eta)\gets\Call{Master}{\bm{p},\Phi_{\mathrm{work}}}$
  \Comment{best claimed payoff under known payoff vectors}
  \State $\bm{z}_{\mathrm{new}}\gets\Call{Oracle}{\bm{y},F,\bm{g}}$
  \Comment{worst-case payoff vector}
  \If{$\bm{y}^\top\bm{z}_{\mathrm{new}}\ge \eta$}
    \State \Return $\bm{y}^{\medstar}=\bm{y}$ and
    $\pi^{\medstar}=\eta-\bm{p}^\top\bm{y}$
    \Comment{claimed payoff certified}
  \Else
    \State $\Phi_{\mathrm{work}}\gets
    \Phi_{\mathrm{work}}\cup\{\bm{z}_{\mathrm{new}}\}$
    \Comment{add violated constraint}
  \EndIf
\EndLoop
\end{algorithmic}
\end{algorithm}
\vspace{-7.5pt}

\enlargethispage*{.75\baselineskip}
As long as the constraint system remains unchanged, 
reusing previously discovered payoff vectors as \textsc{Master} constraints
across quote states can reduce the number of subsequent \textsc{Oracle} calls. \linebreak
In Appendix~\ref{app:portfolio}, we detail the linear-programming reformulation behind Algorithm~\ref{alg:detector}.

\subsection{Measuring Executable Arbitrage in Finite-Depth Order Books}
\label{sec:executable}

The quoted-price benchmark abstracts from two features of real execution.
First, execution costs depend on bid--ask spreads, depth at each price level, and taker fees.
Second, the unit-position normalization identifies the optimal portfolio direction but not its
executable size.  To address both limitations, we replace the linear quoted-price cost
$\bm{p}^\top\bm{y}$ with the actual execution cost, and we keep the payoff side of
Sections~\ref{sec:coherent} and~\ref{sec:ip} unchanged.
We then answer three questions in turn: when an executable arbitrage exists, how large it is,
and how Algorithm~\ref{alg:detector} computes it. 

\paragraph{The execution cost of one contract.}
For contract~\(i\), we arrange its bids and asks with their quantities into a price ladder
with $B_i$ bid levels and $A_i$ ask levels,
such that the fee-adjusted bid prices~$\tilde b_{i\ell}$ and \linebreak ask prices $\tilde a_{i\ell}$
are ordered from lowest to highest
\(\tilde b_{iB_i} \le \cdots \le \tilde b_{i2} \le \tilde b_{i1} \le
\tilde a_{i1} \le \tilde a_{i2} \le \cdots \le \tilde a_{iA_i}\):
\begin{equation}
\label{eq:price-ladder}
\begin{array}{c@{\qquad}c}
(\tilde b_{iB_i},q_{iB_i}^{\bid}),\enspace \ldots,\enspace
(\tilde b_{i2},q_{i2}^{\bid}),\enspace (\tilde b_{i1},q_{i1}^{\bid}),
&
(\tilde a_{i1},q_{i1}^{\ask}),\enspace (\tilde a_{i2},q_{i2}^{\ask}),\enspace
\ldots,\enspace (\tilde a_{iA_i},q_{iA_i}^{\ask}).
\\[-0.5ex]
\xleftarrow[\text{\qquad\qquad\quad selling more walks down the bids}]{}
&
\xrightarrow[\text{buying more walks up the asks \qquad\qquad\qquad}]{}
\end{array}
\end{equation}
Let $K_i(s)$ denote the total cost of changing the position in contract~$i$ by $s$ shares, with $s>0$ for buying and $s<0$ for selling.
We obtain $K_i(s)$ by moving stepwise up or down the ladder: 
as an order reaches deeper price levels, additional shares are bought at higher asks or sold at lower bids.
For example, a purchase of $s$ shares with
$q^{\ask}_{i1}+q^{\ask}_{i2}< s\le q^{\ask}_{i1}+q^{\ask}_{i2}+q^{\ask}_{i3}$ costs
\(K_i(s)=\tilde a_{i1}q^{\ask}_{i1}+\tilde a_{i2}q^{\ask}_{i2}
+\tilde a_{i3}\bigl(s-q^{\ask}_{i1}-q^{\ask}_{i2}\bigr).\) \linebreak
For $s<0$, a sale of $|s|$ shares walks down the bids in the same way,
and its cost equals the negative of the proceeds, so $K_i(s)<0$.
Trading nothing costs nothing, so $K_i(0)=0$.
A trade beyond the displayed depth cannot be executed, so $K_i(s)=+\infty$
for $s$ outside $[-Q^{\bid}_i,\,Q^{\ask}_i]$,
where \(Q^{\bid}_i:=\medsum\nolimits_{\ell=1}^{B_i} q^{\bid}_{i\ell}\) \linebreak
and \(Q^{\ask}_i:=\medsum\nolimits_{\ell=1}^{A_i} q^{\ask}_{i\ell}\)
are the displayed depths on the bid and ask sides.

\paragraph{Executable portfolios and their profit.}
The execution cost of a portfolio $\bm{y}$ is the sum of the costs of its legs,
\(K(\bm{y}):=\sum_{i\in I}K_i(y_i)\), so $K(\bm{0})=0$.
A portfolio can be executed exactly when its cost is finite, \linebreak so the \emph{executable set} is
\begin{equation}
\label{eq:executable-set}
D:=\operatorname{dom}K=\{\bm{y}\in\mathbb{R}^{I}:K(\bm{y})<+\infty\}
=\prod_{i\in I}\bigl[-Q^{\bid}_i,\,Q^{\ask}_i\bigr],
\end{equation}
a box whose sides are the displayed depths.  For a portfolio $\bm{y}\in D$, the executable profit under payoff vector $\bm{z}\in\Phi$ is
\(\bm{y}^\top\bm{z}-K(\bm{y})\), and its guaranteed executable profit is
\begin{equation}
\label{eq:guaranteed-executable-profit}
  \min_{\bm{z}\in\Phi}\big\{\bm{y}^\top\bm{z}-K(\bm{y})\big\}
  =
  \min_{\bm{z}\in\Phi}\bm{y}^\top\bm{z}-K(\bm{y}).
\end{equation}
The payoff side is unchanged, whereas $K(\bm{y})$ records the prices, depth, and fees consumed along
the order book.  In parallel with Definition~\ref{def:arb}, a portfolio $\bm{y}\in D$ is an
\emph{executable arbitrage} if its guaranteed executable profit is strictly positive.

Proposition~\ref{prop:executable} extends the single-price benchmark of 
Proposition~\ref{prop:coherence} by incorporating trading costs. \linebreak
No executable arbitrage exists if and only if a single belief 
can value every contract within its fee-\linebreak{}adjusted bid--ask spread.
Spreads and fees therefore allow some incoherence among quoted prices without creating
a profitable trade: arbitrage requires discrepancies that persist 
at actual execution prices \citep{jouini1995martingales,pennanen2011arbitrage}.
Appendix~\ref{app:proofs} provides the proof and 
recovers Propositions~\ref{prop:coherence} and~\ref{prop:l1} as the special case of the linear quoted-price cost.

\begin{proposition}[Executable arbitrage]
\label{prop:executable}
\:\!\!Over the executable set $D$, the maximum guaranteed executable profit \looseness=-1
\begin{equation}
  \Pi^{\medstar}
  =
  \max_{\bm{y}\in D}
  \min_{\bm{z}\in\Phi}
  \big\{\bm{y}^\top\bm{z}-K(\bm{y})\big\}
\label{eq:exec-duality}
\end{equation}
is nonnegative and is attained at an optimal executable portfolio $\bm{y}^{\medstar}$.
An executable arbitrage exists, that is, $\Pi^{\medstar}>0$, if and only if no coherent price
vector lies within every fee-adjusted spread:
\begin{equation}
\label{eq:exec-no-arb}
M\cap\prod_{i\in I}\bigl[\tilde b_{i1},\,\tilde a_{i1}\bigr]=\emptyset .
\end{equation}
\end{proposition}

While the best quotes $\tilde b_{i1}$ and
$\tilde a_{i1}$ determine whether an arbitrage exists, the full order book determines
its optimal scale and total guaranteed profit.  As a portfolio expands, higher purchase prices and
lower sale proceeds erode the gain from additional trades.
A large price discrepancy supported by little depth can therefore offer less profit than a smaller
discrepancy in a deeper market. \looseness=-1

Algorithm~\ref{alg:detector} carries over with the same payoff constraints and worst-case
\textsc{Oracle} step; 
the \linebreak \textsc{Master} step now uses the order-book cost $K$ and the executable set $D$;
because $K$ is piecewise linear, \linebreak it remains a linear program.  The cap on the number of legs in
Section\,\ref{sec:implementation} is the exception: the portfolios that satisfy it do not form a
convex set.
We use auxiliary indicator variables,
so the \textsc{Master} \linebreak step becomes a mixed-integer linear program.
Appendix~\ref{app:portfolio} provides computational details.

Section~\ref{sec:implementation} measures executable arbitrage in two steps before settlement:
the detector uses the best quotes to identify arbitrage windows, and the trade simulation
reoptimizes quantities on the detected legs within each window using the full-depth order book.
This separates the detection of inconsistent prices from the measurement of the profit available
at displayed quantities.

\section{Recovering the Logical Relation System at Scale}
\label{sec:relations}

Section~\ref{sec:theory} takes the constraint system $(F,\bm{g})$ as given.
Prediction-market venues, however, list contracts individually and do not publish the complete
logical structure.  Section~\ref{sec:classification-errors} first defines the Type I and Type II
errors that can arise when recovering these relations and shows why their effects on arbitrage
detection are asymmetric.  Section~\ref{sec:identify} then explains how we recover the logical
structure in three layers. 

\subsection{Scale and Relation Errors}
\label{sec:classification-errors}

In Section~\ref{sec:theory}, we implicitly assumed a complete and valid logical relation system.
We now examine two possible errors in recovering this system.
First, $N$ contracts generate $N(N-1)/2$ possible pairs and up to $2^N$ subsets for multivariate
search.  Exhaustive search is therefore infeasible, so valid relations may remain unidentified.
Second, a proposed relation may itself be incorrect.
We call the admission of an invalid relation a Type I error and the omission of a valid relation a
Type II error.

A Type I error directly changes the objects defined in Section~\ref{sec:theory}.
Adding a false constraint to $F\bm{z}\leq\bm{g}$ can exclude a payoff vector that is feasible under
the true relation system, producing an incorrectly small payoff set $\Phi^{\dagger}\subsetneqq\Phi$
and coherent set $M^{\dagger}=\operatorname{conv}(\Phi^{\dagger})$.
The detector then minimizes over the wrong subset, so the calculated worst-case payoff can exceed
the true worst-case payoff.  The reported guaranteed profit and the portfolio weights
chosen to maximize it may both be wrong; in an excluded feasible state,
the same portfolio can lose money.  Repeated small apparent gains can therefore conceal a large loss
when that state occurs---the ``picking up pennies in front of a steamroller'' problem discussed by
\citet[p.~182]{stulz2007hedge}.

A Type II error instead leaves out a valid constraint and therefore expands the payoff set.
Proposition~\ref{prop:underidentification} shows that this error has a more limited consequence.
Appendix~\ref{app:proofs} gives the proof.

\begin{proposition}[False omission is conservative]
\label{prop:underidentification}
Let $\Phi$ be the payoff set generated by the full set of valid logical relations,
and let $\Phi^{-}\supseteq\Phi$ be the payoff set obtained by omitting valid relations without
adding invalid ones.  For any execution-cost function $K$, executable set $D=\operatorname{dom}K$,
and portfolio $\bm{y}\in D$,
\vspace{-1pt}\begin{equation}\vspace{-1pt}
\label{eq:omitted-relation-payoff}
\min_{\bm{z}\in\Phi^{-}}\{\bm{y}^{\top}\bm{z}-K(\bm{y})\}
\;\leq\;
\min_{\bm{z}\in\Phi}\{\bm{y}^{\top}\bm{z}-K(\bm{y})\}.
\end{equation}
The maximal guaranteed profit computed under $\Phi^{-}$ therefore cannot exceed that computed under
$\Phi$.  Any portfolio classified as an arbitrage under $\Phi^{-}$ remains an arbitrage under the
true payoff set $\Phi$.
\end{proposition}

\vspace{-.5em}

Type II errors are not costless, though:
they reduce detection power and the associated profit opportunities.
The proposition shows only that relations left unidentified because of scale cannot manufacture a
guaranteed profit that does not exist.  We therefore seek broad coverage while requiring stronger
evidence before admitting a relation.  The next subsection describes how three complementary layers
pursue both goals.

\vspace{-.25em}

\subsection{Recovering Relations in Three Layers}
\label{sec:identify}

We build the relation system in three layers, as shown in Figure~\ref{fig:relation-pipeline-v2}.
First, we obtain \emph{venue-asserted} relations from the venue's collateral mechanism and 
the resolution rules of venue-designated groups.
Second, for sports contracts, we construct \emph{venue-derived} relations by mapping typed
fields
into a common score state.
Third, we extract \emph{semantic-inferred} relations from natural-language contract text. \linebreak
Each layer uses evidence suited to its source and 
checks relations before they enter the system. \looseness=-1

\begin{figure}[!t]
\centering
\includegraphics[width=\textwidth]{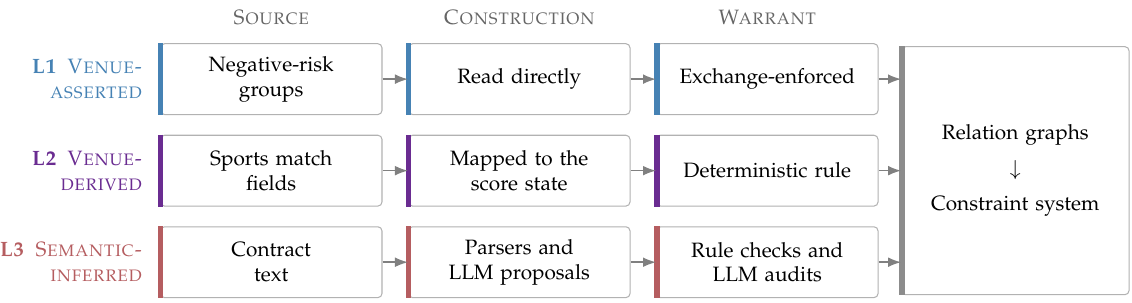}%
\vspace{-5pt}
\caption{Three relation layers and their checks.}
\label{fig:relation-pipeline-v2}
\vspace{-10pt}
\end{figure}

\enlargethispage*{1\baselineskip}

\paragraph{Layer 1: Venue-asserted relations.}
Polymarket places contracts that it treats as mutually exclusive \linebreak in the same negative-risk group,
identified by a common \texttt{negRiskMarketID}.
We check these mutual-\textls[-5]{exclusivity relations against the venue's collateral mechanism.
For example, group
\href{https://gamma-api.polymarket.com/events/slug/fifwc-alg-aut-2026-06-27}%
{\texttt{0x143fa1\ldots cc00}}}
contains three contracts for the Algeria--Austria World Cup match on 28 June 2026:
``\href{https://polymarket.com/sports/world-cup/fifwc-alg-aut-2026-06-27}{Algeria wins}'',
``\href{https://polymarket.com/sports/world-cup/fifwc-alg-aut-2026-06-27}{Draw}'',
and
``\href{https://polymarket.com/sports/world-cup/fifwc-alg-aut-2026-06-27}{Austria wins}''.
Exactly one of these three contracts resolves \YES{} (the 
``\href{https://polymarket.com/sports/world-cup/fifwc-alg-aut-2026-06-27}{Draw}''
contract would resolve to \YES{} if the match were cancelled without a make-up game),
so we encode the group as the partition
constraint $z_{\mathrm{ALG}}+z_{\mathrm{Draw}}+z_{\mathrm{AUT}}=1$.
The 20{,}866 negative-risk groups in our sample yield 20{,}118 partition and 748 mutual-exclusivity
constraints.\footnote{Negative-risk group membership alone is not sufficient to guarantee
collective exhaustiveness: the venue rules out more than one \YES{} outcome but not an all-\NO{}
outcome.  We therefore read the complete contract rules of every group. \linebreak
For 20{,}118 of the 20{,}866 groups the rules establish collective exhaustiveness,
and then we impose equality. 
The remaining \linebreak 748 groups, whose rules contradict exactly-one settlement or which we could not
assess, retain only the mutual-exclusivity inequality.
Buying or selling \YES{} baskets can yield arbitrage
in the former, whereas only selling can do so in the latter.\vspace{-2pt}}

\paragraph{Layer 2: Venue-derived relations.}
In some contract classes, constraints are not encoded directly by the venue but follow
deterministically from its typed metadata.  For sports contracts, Polymarket reports the match,
contract type, and line.  We map these fields into 
a representation of the common score state that determines each contract's payoff.
This mapping reveals the logical relations among contracts on the same match.
The venue-derived layer produces more relations than either of the other two layers
(Table~\ref{tab:relation-system-layers}). We retain only relations verified by the deterministic
score-state mapping. \looseness=-1

\paragraph{Layer 3: Semantic-inferred relations.}
For relations that are not encoded in Layers 1 or 2, we use the natural-language contract text.
Some contracts contain standardized numerical thresholds or deadlines that can be parsed
mechanically.  For example,
``\href{https://polymarket.com/event/bitcoin-above-on-july-28-2026-5am-et/bitcoin-above-65000-on-%
july-28-2026-5am-et}{Bitcoin above \$65{,}000 on July 28, 5AM ET?}'' \linebreak
implies
``\href{https://polymarket.com/event/bitcoin-above-on-july-28-2026-5am-et/bitcoin-above-64800-on-%
july-28-2026-5am-et}{Bitcoin above \$64{,}800 on July 28, 5AM ET?}''.
Our parsers first match the underlying event, measured quantity, unit, horizon, resolution source,
and operative deadline.  They then infer only the relations implied by those matched fields.

\subparagraph{LLM proposals and auditing.}
Our parsers cover standardized patterns, but other relations require interpreting free-form
resolution criteria.  We use large language models (LLMs) to search for these relations.
Because comparing every contract pair is infeasible, we first pre-screen the pairs.
We rank contract pairs by weighted word overlap across their event titles, questions,
and group titles, giving more weight to distinctive words and phrases.
This pre-screen cuts 277 billion possible pairs to 704{,}450 candidate pairs and 13{,}894 candidate
groups.  Appendix~\ref{app:candidate-retrieval} describes the retrieval algorithm and group
construction.

We first use the five leading open-weight LLMs available in September 2026 as
\emph{proposal models}: \linebreak Qwen3-32B \citep{yang2025qwen3}, Qwen3-235B-A22B \citep{qwen2025qwen235},
Llama-4-Maverick \citep{meta2025llama4}, DeepSeek-V4-Flash \citep{deepseek2026v4},
and GLM-5.2 \citep{zai2026glm52}. For each candidate, these proposal models receive the relevant
contract metadata and natural-language text.  Each model is prompted to return a relation type and direction,
exact textual evidence, and an explanation, or to abstain when the evidence is insufficient.

In addition, because proposal models are error-prone, every LLM proposal is reviewed by two
\emph{audit models}:
Claude Fable 5.1 \citep{anthropic2026fable51}
and 
GPT-5.6 Sol \citep{openai2026gpt56sol}---later supplemented by 
GPT-6 Astra \citep{openai2026gpt6astra}.
All are frontier LLMs available as of September~2026 and \linebreak are deployed as agents with max or ultra reasoning effort.
Each auditor independently checks the contract text and
dates, searches for a counterexample, and determines whether a relation exists and, if so,
what type of relation the evidence supports. We retain a proposal only when both auditors agree on
its type and direction; all other proposals are removed.
Appendices~\ref{app:prompt-templates} and~\ref{app:llm-examples} report the prompts and audit examples.

In Table~\ref{tab:relation-audit}, the percentage of proposals that both auditors confirm 
varies widely across the five proposal models and tracks model size only loosely.
GLM-5.2 has the highest confirmation rate: \linebreak it proposes 137{,}055 pairwise and 2{,}011
multivariate relations, of which 64.6\% and 21.7\%, respectively, are confirmed by both audit
models.

\begin{table}[!t]
\centering
\caption{Audit-model agreement by proposal model}
\label{tab:relation-audit}
\begin{threeparttable}
\footnotesize
\begin{tabular*}{\textwidth}{@{\extracolsep{\fill}}lrrrrrr@{}}
\toprule
& \multicolumn{2}{c}{Parameters (billions)} & \multicolumn{2}{c}{Pairwise relations} & \multicolumn{2}{c}{Multivariate units} \\
\cmidrule(lr){2-3}\cmidrule(lr){4-5}\cmidrule(lr){6-7}
Proposal source & Active & Total & Proposed & Confirmed & Proposed & Confirmed \\
\midrule
Llama-4-Maverick  & 17   & 400  & 99{,}361  & 18{,}651 (18.8\%) & 8{,}818 & 349 (4.0\%) \\
Qwen3-32B         & 32.8 & 32.8 & 90{,}398  & 23{,}084 (25.5\%) & 7{,}871 & 283 (3.6\%) \\
DeepSeek-V4-Flash & 13   & 284  & 128{,}084 & 49{,}729 (38.8\%) & 9{,}298 & 702 (7.6\%) \\
Qwen3-235B-A22B   & 22   & 235  & 135{,}202 & 58{,}963 (43.6\%) & 6{,}007 & 488 (8.1\%) \\
GLM-5.2           & 40   & 744  & 137{,}055 & 88{,}571 (64.6\%) & 2{,}011 & 437 (21.7\%) \\
\midrule
Union set         & --   & --   & 239{,}653 & 102{,}288 (42.7\%) & 23{,}519 & 1{,}153 (4.9\%) \\
\bottomrule
\end{tabular*}
\begin{tablenotes}[flushleft]\footnotesize
\item\!\emph{Notes:} Confirmed means that both auditors confirm the model's original relation.
\end{tablenotes}
\end{threeparttable}
\vspace{-5pt}
\end{table}

\paragraph{The relation system. }
Table~\ref{tab:relation-system-layers} reports the overall system. 
Panel A counts the 1{,}942{,}133 relations the three layers recover, including the ones implied by others;
Panel B counts the 573{,}024 relations of the reduced system, of which 5{,}268 appear
in two layers; these relations compile into 603{,}242 rows of $F_*$.
Each relation compiles into the rows of Table~\ref{tab:relation-templates},
and the last column of Table~\ref{tab:relation-system-layers} counts these rows by layer.
Layer 2 contributes the most because the score-state mapping is deterministic and applies to
every contract carrying the venue's typed sports fields, the largest contract class in our
sample.  Layer 3 contributes 118{,}792 relations after deduplication across 
the dual-audited LLM outputs of Table~\ref{tab:relation-audit}
and the outputs of the numeric and temporal parsers.

Coverage of semantic-inferred relations may not and need not be complete.
Candidate retrieval, model abstention, and limited computation can all omit valid relations,
creating Type II errors.  Such omissions reduce detection power, but they cannot manufacture an
opportunity that is not there.  More accurate models and greater computation could search a wider
candidate set and reduce these omissions further.  By Proposition~\ref{prop:underidentification},
provided the retained relations are valid, the guaranteed profit computed on an incomplete system is
a lower bound on the guaranteed profit under the complete one, so every portfolio we certify remains
an arbitrage once the missing relations are restored.

Figure~\ref{fig:relation-system-map1} shows how the three layers contribute differently across event
classes and logical relation types.  Layer~1 is mostly partitions:
20{,}118 of the 20{,}866 venue-designated negative-risk groups in our sample are compiled as exhaustive partitions.
Layer~2 compares the payoffs of contracts on the same match---exact scores, moneylines, spreads, totals,
both-teams-to-score, and team totals---across all possible scores.
This common underlying score state enables a single layer to recover all four logical
relation types.
Within Layer~2, implication is the most common type, with 246{,}848 relations. \linebreak
Layer~3 reaches the classes largely left uncovered by the other two layers and is dominated by
implication relations among threshold contracts in crypto and finance.

\begin{table}[!p]
\centering
\caption{Recovered and reduced relations by layer}
\label{tab:relation-system-layers}
\begin{threeparttable}
\footnotesize
\begin{tabular*}{\textwidth}{@{\extracolsep{\fill}}lrrrrrrr@{}}
\toprule
Relation layer & \!\!\!\!\!\!\!\!\!\!\!\!Mutual exclusivity & Partition & Implication & Conjunction & Other & Relations & Rows \\
\midrule
\multicolumn{8}{@{}l}{\textit{Panel A: Recovered relations in $F$}} \\
\addlinespace[2pt]
Layer 1: Venue-asserted    & 748       & 20{,}118  & 0         & 0         & 0         & 20{,}866  & 40{,}984  \\
Layer 2: Venue-derived     & 870{,}389 & 5{,}254   & 480{,}310 & 2{,}370   & 0         & 1{,}358{,}323 & 1{,}368{,}317 \\
Layer 3: Semantic-inferred & 25{,}448  & 25        & 537{,}381 & 0         & 90        & 562{,}944 & 563{,}050 \\
\midrule
Total                      & 896{,}585 & 25{,}397  & 1{,}017{,}691 & 2{,}370   & 90        & 1{,}942{,}133 & 1{,}972{,}351 \\
\midrule[0.08em]
\multicolumn{8}{@{}l}{\textit{Panel B: Reduced relations in $F_*$}} \\
\addlinespace[2pt]
Layer 1: Venue-asserted    & 748       & 20{,}118  & 0         & 0         & 0         & 20{,}866  & 40{,}984  \\
Layer 2: Venue-derived     & 178{,}894 & 5{,}254   & 246{,}848 & 2{,}370   & 0         & 433{,}366 & 443{,}360 \\
Layer 3: Semantic-inferred & 14{,}292  & 25        & 104{,}385 & 0         & 90        & 118{,}792 & 118{,}898 \\
\midrule
Total                      & 193{,}934 & 25{,}397  & 351{,}233 & 2{,}370   & 90        & 573{,}024 & 603{,}242 \\
\bottomrule
\end{tabular*}
\begin{tablenotes}[flushleft]\footnotesize
\item\!\emph{Notes:} Panel A counts recovered relations, with identical constraints merged within
each layer. Panel B reports counts for the reduced system. Relations shared by different layers are
counted in each layer. Other counts constraint rows not in the logical relation types of
Table~\ref{tab:relation-templates}.
\end{tablenotes}
\end{threeparttable}
\end{table}

\begin{figure}[!p]
\centering
\includegraphics[
  width=\textwidth,
  trim=0 0.35cm 0 0.15cm,
  clip
]{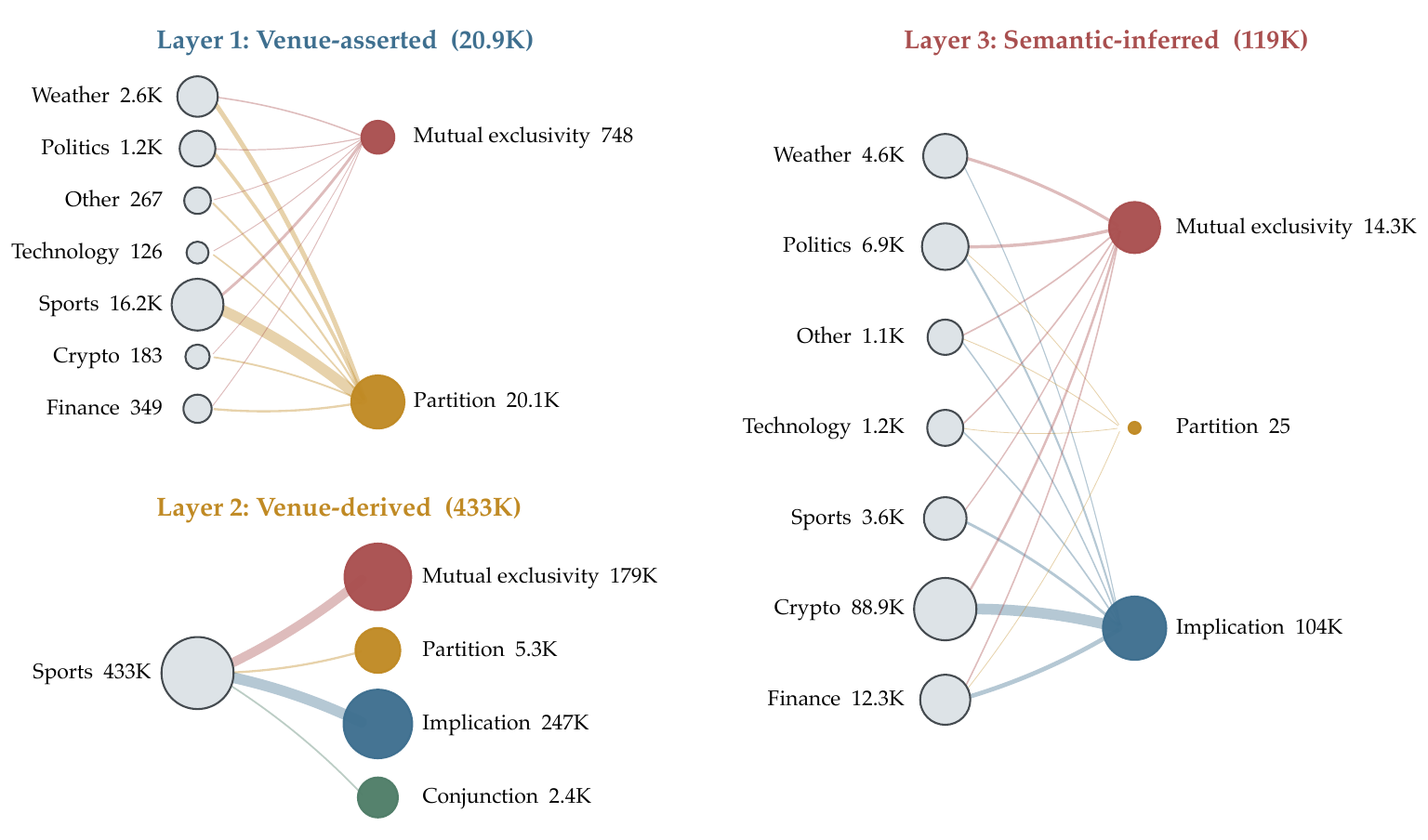}
\caption{Reduced relations in $F_*$ by layer, event class, and logical relation type.
Each panel maps one layer's reduced relations from event class to logical relation type;
node diameter is linear in the logarithm of the count, and the labels give the counts.}
\label{fig:relation-system-map1}
\vspace{-2pt}
\end{figure}

The overall relation system is highly concentrated in two contract classes:
sports accounts for 453{,}190 relations and crypto for 89{,}088, together representing 94.6\% of the
total, while no other class exceeds~3\%.  Two forces produce this imbalance.
Sports and crypto are the venue's largest classes, \linebreak accounting for 67.5\% and 22.4\% of
its listed contracts.  Additionally, their contracts are typically expressed as typed lines and numeric
thresholds, so their relations can be recovered by score mappings or structured parsers without
individual LLM evaluation.

\FloatBarrier
\begin{figure}[!t]
\centering
\includegraphics[width=\textwidth]{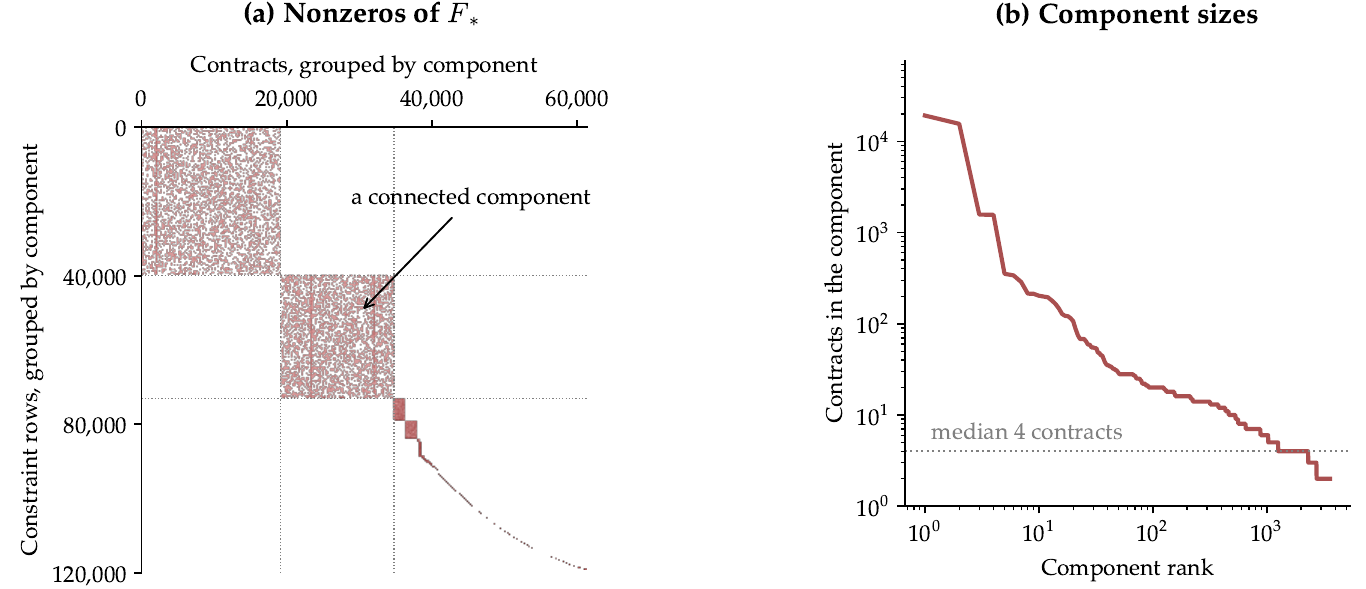}
\caption{Connected components of the semantic-inferred layer's constraint matrix \(F_*\).
Panel (a) shows the nonzeros of the semantic-inferred layer's \(F_*\)
on a $1{,}200\times1{,}200$ grid of cells, each covering about 100~rows and 50 contracts,
with rows and columns ordered by connected component from largest to smallest and each cell shaded
by its number of nonzeros on a log scale; panel (b) shows the number of contracts in each component
against its rank.}
\label{fig:relation-system-map2}
\end{figure}

\paragraph{Connected components.}
Contracts linked by a sequence of relations form a \emph{connected component}:
the smallest set that must be priced jointly, because no constraint reaches across its boundary.
We group the rows of each layer's system by the contracts they share.
For the semantic-inferred layer, whose $F_*$ has 118{,}898 rows over 61{,}349 contract columns,
the system separates into 3{,}591 connected components, with a median size of 4 contracts and a
maximum of 19{,}188.  Figure~\ref{fig:relation-system-map2} shows the matrix with rows and columns
permuted so that each component occupies a diagonal block.

The two largest consist of Bitcoin and Ethereum price-threshold contracts,
with 19{,}188 and 15{,}566 contracts.  These components are large because two types of relations
chain together almost all contracts on the same coin:
``\href{https://polymarket.com/event/what-price-will-bitcoin-hit-on-july-28/will-bitcoin-reach-65k-on-july-28}%
{Bitcoin hits \$65{,}000 on July 28}''
implies
``\href{https://polymarket.com/event/what-price-will-bitcoin-hit-on-july-28/will-bitcoin-reach-64k-on-july-28}%
{Bitcoin hits \$64{,}000 on July 28}'', 
and \linebreak the former contract, by temporal nesting, also implies
``\href{https://polymarket.com/event/what-price-will-bitcoin-hit-in-july-2026/will-bitcoin-reach-65k-in-july-2026}%
{Bitcoin hits \$65{,}000 in July}''
while being implied by
``\href{https://polymarket.com/event/bitcoin-above-on-july-28-2026-5am-et/bitcoin-above-65000-on-%
july-28-2026-5am-et}{Bitcoin above \$65{,}000 on July 28, 5AM ET?}''.
Thus, within one horizon every threshold is tied to every lower threshold for reaching
and to every higher threshold for dipping, while each contract is tied to the reach or
dip contracts of every containing horizon, causing
19{,}188 of the 19{,}843 Bitcoin 
and 15{,}566 of the 16{,}119 Ethereum price-threshold contracts
listed during the sample to fall into a single component each.
The blocks are shaded by the~number of nonzeros per cell; within a block the contracts are in no
particular order, so the nonzeros spread over the block.
The remaining 3{,}589 components form the thin diagonal;
87\% of them hold ten contracts or fewer. 

The venue-asserted and venue-derived layers do not enter this matrix: we compile each layer into its
own constraint system, and Section~\ref{sec:implementation} detects arbitrage separately within each,
so that profits supported by different layers stay distinct 
and an invalid relation in one layer cannot create a false guarantee in another.

\FloatBarrier
\section{The Anatomy of the Arbitrage Opportunity}
\label{sec:empirical}
We now combine the logical relation system of Section~\ref{sec:relations} with the full-depth order
books to measure combinatorial arbitrage.  Section~\ref{sec:implementation} explains how we detect
arbitrage windows, measure their executable size, and compare their guaranteed profit with their
profit at resolution.  Section~\ref{sec:incidence} reports the resulting opportunities and examines
their sensitivity to execution assumptions.  Section~\ref{sec:limits} examines the conditions under
which arbitrage opportunities persist longer.  In addition, Section~\ref{sec:takers} uses on-chain
fills to examine who takes these opportunities.

\subsection{Measurement: Detection, Execution, and Settlement}
\label{sec:implementation}

We measure arbitrage in three steps.  Best quotes identify arbitrage windows,
full-depth order books determine the portfolios that can be executed, and contract resolutions
determine their profit at resolution.  Appendix~\ref{app:measurement-details} details the data
construction and checks.

\paragraph{Step 1: detect arbitrage windows using only the best quotes.}
Figure~\ref{fig:relation-system-map2} shows how related contracts form connected components,
with no constraints across components.
We can therefore run the detector separately on each component of each layer's system.
A component can be large, but at any given instant only some of its contracts are quoted,
so only the relations among them that are in effect at that instant apply.
We call the quoted contracts and these constraints the component's \emph{active~block}. \looseness=-1

We detect when arbitrage appears using only the best quotes reconstructed from all messages timestamped at or before
the instant being evaluated.\footnote{We run data collectors that receive order-book messages from
Polymarket over multiple network connections.
In practice, delays differ across connections and can reach several minutes,
so we align messages by exchange timestamps rather than their receipt times.
To check reconstruction accuracy, we also compare each book's best bid and ask with the exchange's
reported quotes, which agree in 97.2\% of 13.3 billion comparisons.
\vspace{1pt}}
Whenever a best quote in the active block changes, we solve the arbitrage problem in
Section~\ref{sec:ip} using current bids and asks after taker
fees.\footnote{\label{fn:execution-caps}We cap the active block at 32 contracts in the
venue-asserted layer, 48 in the venue-derived layer, and 16 in the semantic-inferred layer.
We cap portfolios at 5 legs because coordinating trades across more order books can be difficult in
practice.  The 5-leg cap exempts sales in venue-asserted groups, but not purchases
(Appendix~\ref{app:partition-sides} explains why).
\vspace{1pt}} \linebreak
An \emph{arbitrage window} opens when the detector finds a portfolio with positive guaranteed profit
at the best quotes.  It remains open while successive evaluations return positive profit and closes
at the first evaluation that does not.
This step identifies 491{,}827 arbitrage windows.\footnote{Here and elsewhere,
counts are reported after excising
the windows whose profit guarantees relied on the six
invalid relations identified in the ex-post audit conducted on September 22, 2026
(see ``Corrective excision'' in Appendix~\ref{app:three-layers}).}
Appendix~\ref{app:three-layers} details the detection procedure for three layers.

\paragraph{Step 2: simulate arbitrage trades against full-depth order books.}

Full-depth simulation is needed only for the arbitrage windows identified in \textbf{Step 1}:
if even the best quotes are arbitrage-free, then deeper quotes,
which are no more favorable, cannot yield arbitrage
(Proposition~\ref{prop:executable}).
Within an \linebreak arbitrage window, best quotes are sufficient to identify a profitable
portfolio direction, but do not alone determine how much can be traded.
To determine executable size, we must incorporate displayed depth;
Figure~\ref{fig:step2-take-bracket}(a) illustrates an example with two contracts. 

We therefore replay the live order-book messages in exchange-time order to reconstruct the full-depth
books throughout each arbitrage window.
At every update, we reoptimize quantities on the detected
legs, retaining their trading directions, and simulate trades against displayed depth to determine
executable size and guaranteed profit. Because simulated executions do not alter the live
order-book messages, previously consumed liquidity can reappear in subsequent reconstructed books. 
To prevent double-counting, we maintain a shared liquidity ledger and deduct previously consumed
depth before each reoptimization.
A window is called an \emph{executed window} if 
at least one trade is executed within it in our simulation.
Of the 491{,}827 arbitrage windows previously detected,
we find 125{,}095 executed windows.

\begin{figure}[!t]
\centering
\includegraphics[width=\textwidth]{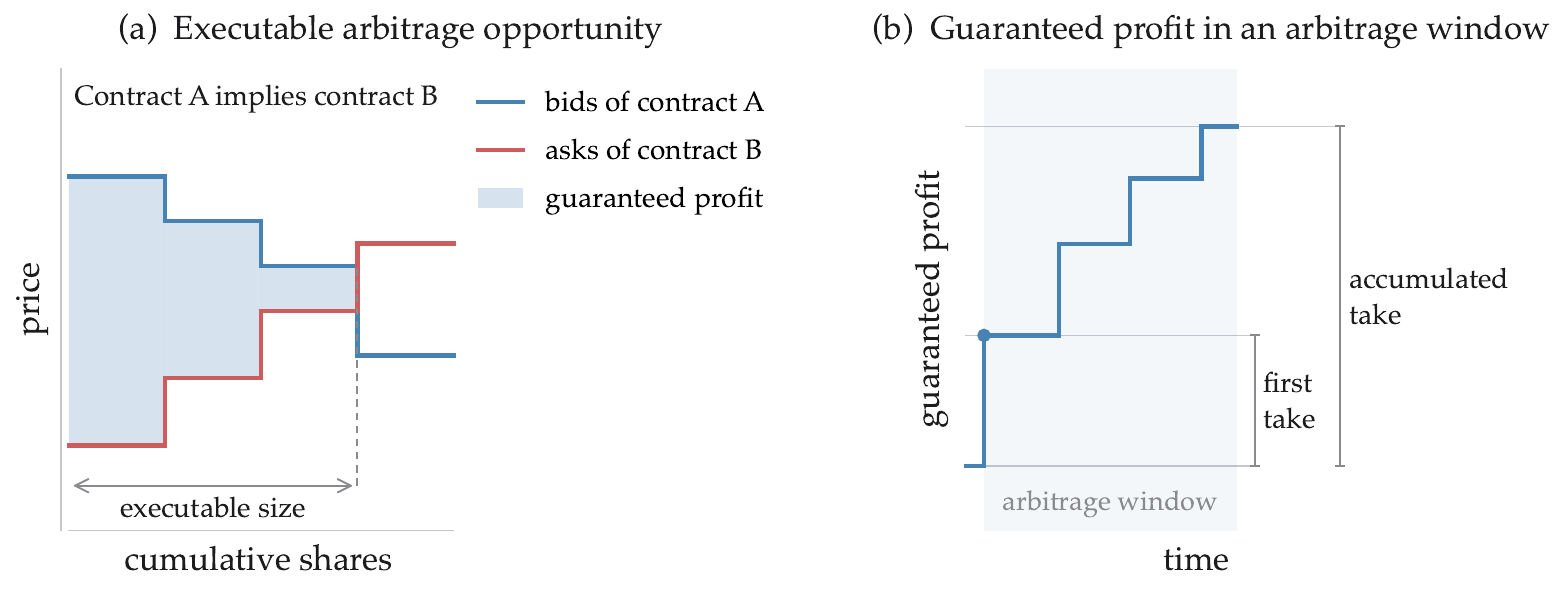}
\caption{Simulated arbitrage trades: panel (a) illustrates guaranteed profit across price levels,
and panel (b) distinguishes the first take from the accumulated take.
The illustration abstracts from fees.}
\label{fig:step2-take-bracket}
\vspace{-4.5pt}
\end{figure}

We report two measures of guaranteed profit within an executed window.
The \emph{first take} is the guaranteed profit of the first execution in each executed window.
The \emph{accumulated take} is the sum of guaranteed profit across all executions within the window.
The profit available after the first execution depends on how subsequent quotes respond to the
arbitrageur's trades.  We cannot infer these counterfactual responses from the observed data.
Conditional on execution at displayed quotes, we use the first take as a lower bound,
assuming no further execution within the window, and the accumulated take as an upper bound,
assuming subsequent quotes remain as observed.
Figure~\ref{fig:step2-take-bracket}(b) illustrates the two measures within one window.

\paragraph{Step 3: test arbitrage windows at resolution.}
We evaluate executed portfolios using their contracts' resolutions.
We classify an executed window as \emph{settled} if every filled leg's contract had resolved to zero or one
as of this writing.\footnote{Windows with an open or void leg are
classified as unsettled.  Void contracts pay half a dollar to each side.}
We define the portfolio's \emph{profit at resolution} as its~terminal~payoff less the cash outlay
and fees paid at execution, and its return as the ratio of its profit to capital committed. 
Capital is  purchase outlays plus one dollar of collateral per share sold on each leg.
The negative-risk mechanism specifically lowers the collateral of the sold legs in venue-asserted
portfolios: one dollar \linebreak mints a basket of one \YES{} token per contract in the group, so these legs
cost one dollar per basket.

\subsection{Guaranteed Profit versus Profit at Resolution}
\label{sec:incidence}

\begin{table}[!h]
\vspace{-7.5pt}
\centering
\caption{Arbitrage windows, guaranteed profit, and profit at resolution by relation layer}
\label{tab:sec5-funnel-v4}
\begin{threeparttable}
\footnotesize
\renewcommand{\arraystretch}{0.95}
\ifdefined\resolutioninset\else\newlength{\resolutioninset}\fi
\sbox0{Profit at resolution}\setlength{\resolutioninset}{0pt}\sbox2{Settled}\ifdim\wd2>\resolutioninset\setlength{\resolutioninset}{\wd2}\fi\sbox2{Capital}\ifdim\wd2>\resolutioninset\setlength{\resolutioninset}{\wd2}\fi\sbox2{112{,}376}\ifdim\wd2>\resolutioninset\setlength{\resolutioninset}{\wd2}\fi\setlength{\resolutioninset}{\dimexpr(\wd0-\resolutioninset)/2\relax}
\begin{tabular*}{\textwidth}{@{\extracolsep{\fill}}lrr<{\hspace{\resolutioninset}}rrr@{}}
\toprule
\multicolumn{6}{@{}l}{\textit{Panel A: All arbitrage windows}} \\[2pt]
 & \multicolumn{3}{c}{Windows} & \multicolumn{2}{c}{Percentage} \\
\cmidrule(lr){2-4}\cmidrule(l){5-6}
 & Detected & Executed & Settled & \sfrac
  {\scalebox{1.4}{\textrm{Executed}}}
  {\scalebox{1.4}{\textrm{Detected}}} & \sfrac
  {\scalebox{1.4}{\textrm{Settled}}}
  {\scalebox{1.4}{\textrm{Executed}}} \\
\midrule
Venue-asserted & 234{,}605 & 64{,}681 & 62{,}860 & 27.6\% & 97.2\% \\
Venue-derived & 99{,}159 & 31{,}317 & 31{,}243 & 31.6\% & 99.8\% \\
Semantic-inferred & 158{,}063 & 29{,}097 & 18{,}273 & 18.4\% & 62.8\% \\
All & 491{,}827 & 125{,}095 & 112{,}376 & 25.4\% & 89.8\% \\
\midrule[0.08em]
\multicolumn{6}{@{}l}{\textit{Panel B: Accumulated take, settled windows}} \\[2pt]
 & \multicolumn{3}{c}{Profit and capital (in thousands of dollars)} & \multicolumn{2}{c}{Return on capital} \\
\cmidrule(lr){2-4}\cmidrule(l){5-6}
 & Guaranteed profit & \multicolumn{1}{r}{Profit at resolution} & Capital & Guaranteed & At resolution \\
\midrule
Venue-asserted & 852 & 971 & 22{,}600 & 3.8\% & 4.3\% \\
Venue-derived & 1{,}467 & 2{,}291 & 79{,}027 & 1.9\% & 2.9\% \\
Semantic-inferred & 179 & 364 & 12{,}089 & 1.5\% & 3.0\% \\
All & 2{,}498 & 3{,}626 & 113{,}716 & 2.2\% & 3.2\% \\
\midrule[0.08em]
\multicolumn{6}{@{}l}{\textit{Panel C: First take, settled windows}} \\[2pt]
 & \multicolumn{3}{c}{Profit and capital (in thousands of dollars)} & \multicolumn{2}{c}{Return on capital} \\
\cmidrule(lr){2-4}\cmidrule(l){5-6}
 & Guaranteed profit & \multicolumn{1}{r}{Profit at resolution} & Capital & Guaranteed & At resolution \\
\midrule
Venue-asserted & 375 & 408 & 9{,}759 & 3.8\% & 4.2\% \\
Venue-derived & 695 & 815 & 29{,}564 & 2.4\% & 2.8\% \\
Semantic-inferred & 35 & 75 & 1{,}240 & 2.8\% & 6.0\% \\
All & 1{,}104 & 1{,}297 & 40{,}563 & 2.7\% & 3.2\% \\
\bottomrule
\end{tabular*}
\begin{tablenotes}[flushleft]
\footnotesize
\item\!\textit{Notes:} The sample runs from 23 June to 3 August 2026. Capital sums the amounts committed across executions
and can count the same dollar more than once.
\end{tablenotes}
\end{threeparttable}
\vspace{7.5pt}
\end{table}

Table~\ref{tab:sec5-funnel-v4} summarizes the scale of executable arbitrage. 
Overall, 25.4\% of the detected windows have an execution, and 89.8\% of executed windows have
settled on every filled leg.  The percentage settled is lower for semantic-inferred relations,
at 62.8\%, only because their contracts resolve later. \linebreak
Within the settled sample,
the accumulated take is \$2.50 million in profit on \$113.7 million of
capital committed---a guaranteed return of 2.2\%, and 3.5\% a week on the \$11.1 million outstanding at the peak.
The first take totals \$1.10 million of this profit on~\$40.6 million of capital;
the remaining \$1.39 million is guaranteed profit from later
executions within the same windows. At resolution, the executions in the accumulated take pay
\$3.63 million---a return of 3.2\% on the capital committed
and 4.8\% a week on the peak capital.\footnote{We retrieve resolution outcomes
only after completing arbitrage detection and trade simulation;
see Appendix~\ref{app:at-contract-resolutions}. \linebreak
 In our sample, no settled window in any layer
has a profit at resolution below its guaranteed profit.
This is consistent with the validity of identified logical relations.}

Across all layers,
much of the guaranteed profit comes from high-return windows.
The median window guarantees 0.5\% on capital, but half of the accumulated take comes from windows
with guaranteed returns above 6.5\% (Figure~\ref{fig:sec5-anatomy-v3}(a)).
The first take is also concentrated in high-return windows, with a profit-weighted
median guaranteed return of 10.4\% (Figure~\ref{fig:sec5-anatomy-v3}(b)).
%
By relation layer, venue-asserted and venue-derived relations contribute 34\%
and 59\% of the guaranteed profit of the accumulated take, respectively.
Semantic-inferred relations account for the remaining 7\%. \linebreak
The contribution of later executions also differs across layers.
They are particularly important in venue-derived relations, where the accumulated take is about 2
times the first take:
repeated quote arrivals allow the arbitrageur to execute again before the
window closes.

\begin{figure}[!t]
\centering
\includegraphics[width=\textwidth]{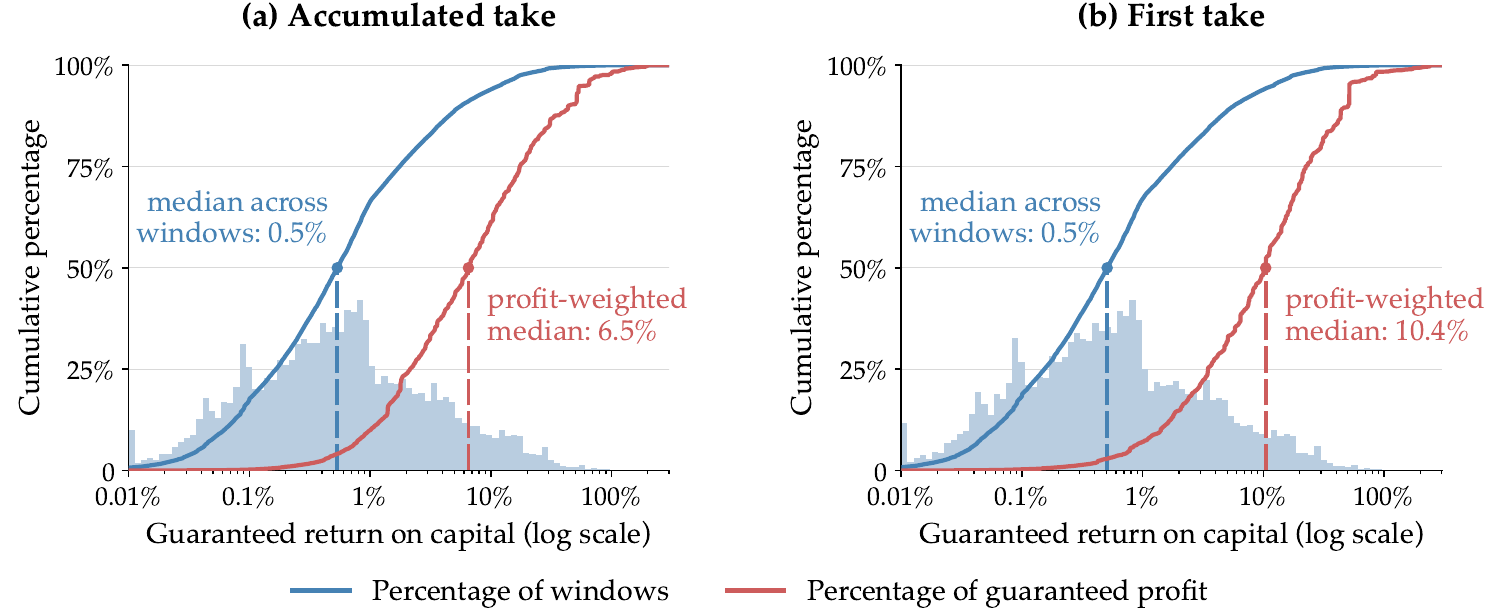}%
\caption{The distribution of executable arbitrage in settled windows.
Panels (a) and (b) show the guaranteed return on capital of the accumulated take and of the first
take on log scales.  The blue curve gives the cumulative percentage of windows,
the red curve their cumulative percentage of guaranteed profit, and the histogram counts windows in
log bins.}
\label{fig:sec5-anatomy-v3}
\end{figure}

The economic value of these returns also depends on how long capital is committed.
For venue-asserted (Layer 1) portfolios that sell the group,
the sold legs come from minted baskets rather than purchases, and the
sale proceeds exceed the minting cost and fees,
so capital need not remain committed until resolution once the legs have been sold.
Venue-asserted portfolios that buy the group, and venue-derived (Layer 2) and semantic-inferred
(Layer 3) portfolios have no equivalent conversion.
If held to resolution, they tie up capital until the latest leg resolves.
We measure the holding period from the execution to the resolution of the latest leg, using the venue's
recorded close time, and count sold venue-asserted baskets at 1~day because the conversion returns their
collateral at once.
In the settled sample, the window-level median is about 3~hours for venue-derived
portfolios and 7~days for semantic-inferred portfolios. 
Venue-derived contracts generally resolve at the end of a match, whereas semantic-inferred relations
also link elections, company valuations, and month-end rankings whose resolutions can be weeks or months away.
The capital in Table~\ref{tab:sec5-funnel-v4} sums capital \linebreak committed across executions,
so the same dollar may be counted multiple times as it is redeployed.
Under this holding-period convention, the capital-weighted average holding period is 1.6 days for
the accumulated take and 0.8 days for the first take.
Estimated capital outstanding peaks at \$11.1 million and \$3.5 million, respectively.

In summary, consider a fund that adopts the strategy: it starts with the \$11.1 million, executes all trades of the replay and takes the full displayed depth of every opportunity, holds each position to resolution, and reuses its capital (each dollar returning to the fund at resolution and being committed \linebreak again).
The fund then earns the \$2.50 million of guaranteed profit over the six weeks of the sample, while the weekly return reaches 3.5\% and the profit at resolution reaches 4.8\% per week. 

Of course, such an ideal arbitrageur does not exist, 
and a real arbitrageur is subject to frictions in liquidity and execution.
We test the measured profits under the following three checks: 
cancelled and re-posted quotes, partial participation, and order delay.
Table~\ref{tab:app-participation-v3b} reports the results.

\paragraph{Robustness check 1: cancel and re-post.}
The first concern is whether replenished depth represents new liquidity.
A maker may cancel and re-post an earlier quote \citep{hasbrouck2013lowlatency}.
If the arbitrageur has already consumed that liquidity, counting its return as new can overstate
profits.
We therefore retain the peak displayed depth at each price level as a reference after its last observation, 
with retention periods of 10 seconds, 1 minute, and 2 minutes.
Depth that disappears and returns within this period does not count as new liquidity;
only the amount above the reference does.  The first block of Table~\ref{tab:app-participation-v3b}
shows that the accumulated take remains within 1\% of the baseline.
The take can exceed the baseline because all windows draw on one liquidity ledger, so the rule
changes which window takes a quote.  In the baseline, a bid that a maker cancels and re-posts
every few seconds is taken at every re-post, and each take also consumes the best quotes on the
portfolio's other legs; under the rule, a re-post counts as new liquidity only after the retention
period, so those quotes remain for later windows.
To illustrate, for the high-profile Argentina--Cape Verde match on 3 July 2026, \linebreak
consider the contract ``\href{https://polymarket.com/sports/world-cup/fifwc-arg-cvi-2026-07-03}{Argentina wins}''.
The baseline sold 230{,}629 shares of it in 234 executions before 23:56:15~UTC,
and the window opening at that time filled 355 shares, earning \$0.99.
Under the 10-second rule, the earlier windows sold 145{,}916 shares in 96 executions, preserving liquidity for
the later window, which sold 72{,}703 shares in one of its 14 executions and earned \$3{,}529 in total.

\paragraph{Robustness check 2: partial participation.}
Another concern is that competing traders may consume part of the displayed depth before the
arbitrageur can trade.  We therefore let the arbitrageur fill only the best-priced 90\%, 70\%, or 50\% of the depth
available on each leg at each execution and book the remainder as taken by other traders.
The second block of Table~\ref{tab:app-participation-v3b} shows that, at 50\% participation,
the accumulated take retains 57\% of its baseline value, while the number of settled windows is
almost unchanged.
The retained share exceeds the participation share because the arbitrageur keeps the best-priced
depth, so the guaranteed return on capital rises from 2.2\% to 2.5\%. \looseness=-1

\paragraph{Robustness check 3: order delay.}
Sometimes quotes disappear while an arbitrageur identifies an opportunity and submits orders.
As a robustness check, we resubmit every order of the baseline after a delay of 100 milliseconds
to 1 second and fill it against the order book at that time, with the liquidity ledger as in the
baseline.
A portfolio whose legs fill unequally no longer has a guaranteed profit, so Panel~B of
Table~\ref{tab:app-participation-v3b} reports what the delayed orders fill and what they pay at
resolution.  With a 1-second delay, 54\% of the orders fill in full and the orders fill
84\% of the submitted quantity; the settled orders pay 79\% of the profit at resolution that
the same orders pay in the baseline, and 22\% of them lose money at resolution.
In the first two checks, returns on capital change little,
and no settled window has a profit at resolution below its guaranteed profit.

\begin{table}[!t]
\centering
\caption{Robustness checks for arbitrage opportunities}
\label{tab:app-participation-v3b}
\begin{threeparttable}
\footnotesize
\setlength{\tabcolsep}{2pt}
\ifdefined\profitinset\else\newlength{\profitinset}\fi
\sbox0{Percentage of baseline}\sbox2{Guaranteed profit}\setlength{\profitinset}{\dimexpr(\wd0-\wd2)/2\relax}
\ifdefined\settledheadwidth\else\newlength{\settledheadwidth}\fi
\settowidth{\settledheadwidth}{Windows settled}
\ifdefined\settledheadoverhang\else\newlength{\settledheadoverhang}\fi
\settowidth{\settledheadoverhang}{All settled windows}\addtolength{\settledheadoverhang}{-\settledheadwidth}
\begin{tabular*}{\textwidth}{@{\extracolsep{\fill}}lrrr<{\hspace{\profitinset}}rr@{}}
\toprule
\multicolumn{6}{@{}l}{\textit{Panel A: liquidity settings, settled windows}} \\[2pt]
 & \multicolumn{2}{c}{Number of} & \multicolumn{1}{c}{Percentage of baseline} & \multicolumn{2}{c}{Return on capital} \\
\cmidrule(l{-\settledheadoverhang}){2-3}\cmidrule{4-4}\cmidrule{5-6}
Setting & \makebox[\settledheadwidth][r]{All settled windows} & Executions & \multicolumn{1}{c}{Guaranteed profit} & Guaranteed & At resolution \\
\midrule
Baseline & 112{,}376 & 391{,}383 & 100.0\% & 2.2\% & 3.2\% \\
\midrule
\multicolumn{6}{@{}l}{\textit{Robustness check 1: cancel and re-post}} \\
\quad retention period of 10 s & 105{,}628 & 366{,}458 & 100.0\% & 2.3\% & 3.2\% \\
\quad retention period of 60 s & 103{,}010 & 344{,}745 & 100.5\% & 2.4\% & 3.3\% \\
\quad retention period of 120 s & 102{,}013 & 334{,}317 & 100.7\% & 2.4\% & 3.4\% \\
\midrule
\multicolumn{6}{@{}l}{\textit{Robustness check 2: partial participation}} \\
\quad 90\% & 112{,}226 & 391{,}067 & 92.4\% & 2.3\% & 3.2\% \\
\quad 70\% & 111{,}822 & 390{,}236 & 75.3\% & 2.4\% & 3.4\% \\
\quad 50\% & 111{,}102 & 388{,}737 & 56.5\% & 2.5\% & 3.5\% \\
\midrule[0.08em]
\end{tabular*}
\begin{tabular*}{\textwidth}{@{\extracolsep{\fill}}lrrrrrr@{}}
\multicolumn{7}{@{}l}{\textit{Panel B: order delay, settled orders}} \\[2pt]
 & & \multicolumn{2}{c}{Percentage filled} & \multicolumn{2}{c}{Profit at resolution} & \\
\cmidrule{3-4}\cmidrule{5-6}
Delay & Settled orders & Orders in full & Quantity & Thousand \$ & Percentage of baseline & Orders with a loss \\
\midrule
\multicolumn{7}{@{}l}{\textit{Robustness check 3: order delay}} \\
\quad 100 ms & 391{,}936 & 71.6\% & 92.4\% & 3{,}315 & 91.4\% & 14.2\% \\
\quad 250 ms & 392{,}047 & 65.1\% & 89.5\% & 3{,}270 & 90.2\% & 17.3\% \\
\quad 500 ms & 392{,}160 & 59.9\% & 86.9\% & 3{,}058 & 84.4\% & 19.6\% \\
\quad 750 ms & 392{,}221 & 57.0\% & 85.3\% & 2{,}990 & 82.5\% & 20.9\% \\
\quad 1{,}000 ms & 392{,}312 & 54.3\% & 84.0\% & 2{,}874 & 79.3\% & 22.1\% \\
\bottomrule
\end{tabular*}
\begin{tablenotes}[flushleft]
\footnotesize
\item\!\textit{Notes:} The baseline is Table~\ref{tab:sec5-funnel-v4}. Panel A reports the accumulated take.
\end{tablenotes}
\end{threeparttable}
\vspace{-2.5pt}
\end{table}

\FloatBarrier
\subsection{Persistence of Arbitrage Opportunities}
\label{sec:limits}
This section examines executability from another angle, by asking how long
arbitrage opportunities persist before they close.
Arbitrage opportunities should persist longer when their underlying payoff relations are harder to
identify.  Since venue-derived and semantic-inferred relations generally require more interpretation
than venue-asserted relations, we next compare persistence across the three layers to examine this
prediction.  We measure persistence by the elapsed time from an executed window's opening to its
close, as defined in Section~\ref{sec:implementation}.
Windows that last less than 0.3 seconds account for 47\% of executed windows but only 8\% of the
accumulated take.  Table~\ref{tab:sec5-persistence-all-v4} reports persistence for all executed
windows and for executed windows that last at least 0.3 seconds.

\enlargethispage{.25\baselineskip}

We first focus on executed windows that last at least 0.3 seconds.
The right half of \hyperref[tab:sec5-persistence-all-v4-A]{Panel~A}
shows that median duration rises from 2.9~seconds for
venue-asserted relations to 6.3~seconds for venue-derived relations and about 31~seconds for
semantic-inferred relations.  The percentage of windows that last at least 10 seconds
rises from 32\% to 44\% and 62\%.
The right half of \hyperref[tab:sec5-persistence-all-v4-B]{Panel~B}
finds that, in the \linebreak same sample,
the differences remain after controlling for portfolio characteristics 
(log first-take profit, \linebreak log first capital committed,
and portfolio direction: sales only, purchases only, or mixed), 
log number of windows the component containing the portfolio opened that day,
and day and hour fixed effects. \linebreak
Column~\ref{tab:sec5-persistence-all-v4-c6} indicates that,
relative to venue-asserted windows, venue-derived windows are
8 per-\textls[-5]{centage points more likely to last at least 10 seconds, 
parser windows 28 percentage points more likely,} \linebreak
and LLM windows 37 percentage points more likely. Exponentiating the log-duration coefficients in
Column~\ref{tab:sec5-persistence-all-v4-c7} implies that venue-derived
and semantic-inferred windows last about 2 and 10 times as long as
venue-asserted windows, respectively, conditional on the included controls.
In Column~\ref{tab:sec5-persistence-all-v4-c8}, \linebreak the corresponding
duration ratios for venue-derived, parser, and LLM windows are respectively \linebreak
about 2, 7, and 15.
The last coefficient row suggests that,
each additional leg is associated with a 0.5-percentage-point higher probability of lasting at
least 10 seconds and a 2\% increase in duration.
These patterns are consistent with the view that discovery and execution costs can limit arbitrage.
The results need to be interpreted as conditional associations,
since the layers also differ in their contracts, liquidity, and conversion mechanisms. 

\FloatBarrier
\begin{table}[!t]
\centering
\caption{Persistence of arbitrage opportunities by layer}
\label{tab:sec5-persistence-all-v4}
\begin{threeparttable}
\footnotesize
\setlength{\tabcolsep}{1pt}
\renewcommand{\arraystretch}{0.95}
\newcommand{\persistencelasts}{Lasts$\,\geq\mkern2mu$10\kern1pt{}s}
\ifdefined\persistencecolwidth\else\newlength{\persistencecolwidth}\fi
\setlength{\persistencecolwidth}{0pt}
\newcommand{\persistencewiden}[1]{\sbox0{#1}\ifdim\wd0>\persistencecolwidth\setlength{\persistencecolwidth}{\wd0}\fi}
\newcommand{\persistencewidenspan}[1]{\sbox0{#1}\dimen0=\dimexpr(\wd0-6\tabcolsep)/4\relax\ifdim\dimen0>\persistencecolwidth\setlength{\persistencecolwidth}{\dimen0}\fi}
\persistencewiden{Windows}\persistencewiden{Median (s)}\persistencewiden{p75 (s)}\persistencewiden{\persistencelasts}\persistencewiden{64{,}681}\persistencewiden{0.2}\persistencewiden{15\%}\persistencewiden{10.7}\persistencewiden{223.8}\persistencewiden{9{,}489}\persistencewiden{125{,}095}\persistencewiden{0.108***}\persistencewiden{(0.012)}\persistencewiden{Yes}\persistencewiden{0.114}\persistencewidenspan{\hspace{0.5em}All executed windows}\persistencewidenspan{\hspace{0.5em}Executed windows of at least 0.3 seconds}
\makeatletter
\newcommand{\persistencecol}[1]{\leavevmode\phantomsection\def\@currentlabel{(#1)}\label{tab:sec5-persistence-all-v4-c#1}(#1)}
\makeatother
\begin{tabular*}{\textwidth}{@{\extracolsep{\fill}}l*{8}{r}@{}}
\toprule
\multicolumn{9}{@{}l}{\phantomsection\label{tab:sec5-persistence-all-v4-A}\textit{Panel A: survival by layer}} \\
\addlinespace[2pt]
 & \multicolumn{4}{c}{\hspace{0.5em}All executed windows} & \multicolumn{4}{c}{\hspace{0.5em}Executed windows of at least 0.3 seconds} \\
\cmidrule(l){2-5}\cmidrule(l){6-9}
Layer & \makebox[\persistencecolwidth][r]{Windows} & \makebox[\persistencecolwidth][r]{Median (s)} & \makebox[\persistencecolwidth][r]{p75 (s)} & \makebox[\persistencecolwidth][r]{\persistencelasts} & \makebox[\persistencecolwidth][r]{Windows} & \makebox[\persistencecolwidth][r]{Median (s)} & \makebox[\persistencecolwidth][r]{p75 (s)} & \makebox[\persistencecolwidth][r]{\persistencelasts} \\
\midrule
Venue-asserted & 64{,}681 & 0.2 & 2.0 & 15\% & 29{,}070 & 2.9 & 10.7 & 32\% \\
Venue-derived & 31{,}317 & 1.5 & 20.6 & 31\% & 22{,}407 & 6.3 & 45.7 & 44\% \\
Semantic-inferred & 29{,}097 & 0.2 & 28.2 & 30\% & 14{,}287 & 30.8 & 223.8 & 62\% \\
\quad Parsers & 21{,}744 & 0.1 & 17.0 & 27\% & 9{,}489 & 32.0 & 214.6 & 62\% \\
\quad LLMs & 7{,}353 & 2.9 & 89.9 & 40\% & 4{,}798 & 29.0 & 270.7 & 61\% \\
\midrule
All & 125{,}095 & 0.4 & 6.8 & 22\% & 65{,}764 & 5.7 & 44.2 & 43\% \\
\midrule[0.08em]
\multicolumn{9}{@{}l}{\phantomsection\label{tab:sec5-persistence-all-v4-B}\textit{Panel B: layer and the speed of correction}} \\
\addlinespace[2pt]
 & \multicolumn{4}{c}{\hspace{0.5em}All executed windows} & \multicolumn{4}{c}{\hspace{0.5em}Executed windows of at least 0.3 seconds} \\
\cmidrule(l){2-5}\cmidrule(l){6-9}
 & \multicolumn{2}{c}{\hspace{0.5em}\persistencelasts} & \multicolumn{2}{c}{\hspace{0.5em}log duration} & \multicolumn{2}{c}{\hspace{0.5em}\persistencelasts} & \multicolumn{2}{c}{\hspace{0.5em}log duration} \\
\cmidrule(l){2-3}\cmidrule(l){4-5}\cmidrule(l){6-7}\cmidrule(l){8-9}
 & \persistencecol{1} & \persistencecol{2} & \persistencecol{3} & \persistencecol{4} & \persistencecol{5} & \persistencecol{6} & \persistencecol{7} & \persistencecol{8} \\
\midrule
Venue-derived & 0.108*** & 0.085*** & 1.496*** & 1.233*** & 0.088*** & 0.077*** & 0.668*** & 0.583*** \\
 & (0.012) & (0.011) & (0.107) & (0.108) & (0.013) & (0.012) & (0.072) & (0.069) \\\addlinespace[3pt]
Semantic-inferred & 0.277*** &  & 2.451*** &  & 0.324*** &  & 2.333*** &  \\
 & (0.023) &  & (0.215) &  & (0.022) &  & (0.111) &  \\\addlinespace[3pt]
Parsers &  & 0.199*** &  & 1.539*** &  & 0.275*** &  & 1.959*** \\
 &  & (0.035) &  & (0.301) &  & (0.034) &  & (0.174) \\\addlinespace[3pt]
LLMs &  & 0.352*** &  & 3.333*** &  & 0.372*** &  & 2.692*** \\
 &  & (0.023) &  & (0.289) &  & (0.021) &  & (0.113) \\\addlinespace[3pt]
Number of legs & 0.006*** & 0.006*** & 0.046*** & 0.038*** & 0.005*** & 0.005*** & 0.022*** & 0.018*** \\
 & (0.001) & (0.001) & (0.007) & (0.007) & (0.001) & (0.001) & (0.006) & (0.006) \\\addlinespace[3pt]
\midrule
Controls and FE & Yes & Yes & Yes & Yes & Yes & Yes & Yes & Yes \\
Observations & 125{,}095 & 125{,}095 & 125{,}095 & 125{,}095 & 65{,}764 & 65{,}764 & 65{,}764 & 65{,}764 \\
Within $R^2$ & 0.114 & 0.119 & 0.168 & 0.178 & 0.108 & 0.109 & 0.179 & 0.183 \\
\bottomrule
\end{tabular*}
\begin{tablenotes}[flushleft]
\footnotesize
\item\!\textit{Notes:} Duration is the exchange time from a window's opening to its close. Panel B
uses a linear probability model for \persistencelasts; the venue-asserted layer is the omitted
category. Controls are the log first take, log first capital committed, the portfolio's direction
(purchases only, or purchases and sales; sales only is omitted), and log number of windows the same
component opened that day, with day and hour fixed effects (FE). Standard errors in parentheses are
clustered by leg set; \linebreak
$^{*}$, $^{**}$, and $^{***}$ denote significance at the 10\%, 5\%, and 1\% levels.
\end{tablenotes}
\end{threeparttable}
\vspace{-12.5pt}
\end{table}

\enlargethispage*{.75\baselineskip}

We next consider all executed windows.
The left half of \hyperref[tab:sec5-persistence-all-v4-A]{Panel~A} shows an increase from
venue-asserted to venue-derived and then to semantic-inferred windows in the 75th percentile of
duration, and the percentage of windows that last at least 10 seconds doubles from venue-asserted
windows to the other two layers.
Median duration also increases from venue-asserted to venue-derived and then to LLM windows.
The left half of \hyperref[tab:sec5-persistence-all-v4-B]{Panel~B} reports the same regressions for the full sample; \linebreak
the estimated differences across layers retain their signs, ordering, and statistical significance. 
\clearpage

\FloatBarrier
\subsection{Who Takes the Arbitrage}
\label{sec:takers}

\begin{table}[!h]
\vspace{-7.5pt}
\centering
\caption{Matched trades in arbitrage windows}
\label{tab:sec5-fill-match-v4}
\begin{threeparttable}
\footnotesize
\setlength{\tabcolsep}{2pt}
\renewcommand{\arraystretch}{0.975}
\sisetup{group-separator={,}, group-minimum-digits=4}
\ifdefined\pairsame\else\newlength{\pairsame}\fi\setlength{\pairsame}{0pt}
\sbox0{Same direction}\ifdim\wd0>\pairsame\setlength{\pairsame}{\wd0}\fi
\sbox0{(matched trades)}\ifdim\wd0>\pairsame\setlength{\pairsame}{\wd0}\fi
\sbox0{~(44.1\%)}\addtolength{\pairsame}{-\wd0}
\sbox0{35,635}\ifdim\wd0>\pairsame\setlength{\pairsame}{\wd0}\fi
\ifdefined\pairopposite\else\newlength{\pairopposite}\fi\setlength{\pairopposite}{0pt}
\sbox0{Opposite}\ifdim\wd0>\pairopposite\setlength{\pairopposite}{\wd0}\fi
\sbox0{direction}\ifdim\wd0>\pairopposite\setlength{\pairopposite}{\wd0}\fi
\sbox0{~(6.8\%)}\addtolength{\pairopposite}{-\wd0}
\sbox0{6,230}\ifdim\wd0>\pairopposite\setlength{\pairopposite}{\wd0}\fi
\ifdefined\pairboth\else\newlength{\pairboth}\fi\setlength{\pairboth}{0pt}
\sbox0{Both}\ifdim\wd0>\pairboth\setlength{\pairboth}{\wd0}\fi
\sbox0{directions}\ifdim\wd0>\pairboth\setlength{\pairboth}{\wd0}\fi
\sbox0{~(5.6\%)}\addtolength{\pairboth}{-\wd0}
\sbox0{5,073}\ifdim\wd0>\pairboth\setlength{\pairboth}{\wd0}\fi
\ifdefined\pairtop\else\newlength{\pairtop}\fi\setlength{\pairtop}{0pt}
\sbox0{By top 10}\ifdim\wd0>\pairtop\setlength{\pairtop}{\wd0}\fi
\sbox0{~(93.7\%)}\addtolength{\pairtop}{-\wd0}
\sbox0{27,620}\ifdim\wd0>\pairtop\setlength{\pairtop}{\wd0}\fi
\ifdefined\anyaddress\else\newlength{\anyaddress}\fi
\sbox0{Number of executed windows with a matched trade}\setlength{\anyaddress}{\wd0}\addtolength{\anyaddress}{-\pairtop}\sbox0{~(93.7\%)}\addtolength{\anyaddress}{-\wd0}\addtolength{\anyaddress}{-\tabcolsep}
\sbox0{By any address}\ifdim\wd0>\anyaddress\setlength{\anyaddress}{\wd0}\fi
\sbox0{35,635}\ifdim\wd0>\anyaddress\setlength{\anyaddress}{\wd0}\fi
\ifdefined\lagcol\else\newlength{\lagcol}\fi\settowidth{\lagcol}{Lag until every leg trades}\addtolength{\lagcol}{2\tabcolsep}\setlength{\lagcol}{0.5\lagcol}\sbox0{Median (s)}\ifdim\wd0>\lagcol\setlength{\lagcol}{\wd0}\fi
\ifdefined\pwidth\else\newlength{\pwidth}\fi\settowidth{\pwidth}{p}
\begin{tabular*}{\textwidth}{@{\extracolsep{\fill}}lS[table-format=6.0]S[table-format=5.0, table-fixed-width=true, table-column-width=\pairsame, table-number-alignment=right]@{\extracolsep{0pt}~}l@{\extracolsep{\fill}\hspace{2\tabcolsep}}S[table-format=4.0, table-fixed-width=true, table-column-width=\pairopposite, table-number-alignment=right]@{\extracolsep{0pt}~}l@{\extracolsep{\fill}\hspace{2\tabcolsep}}S[table-format=4.0, table-fixed-width=true, table-column-width=\pairboth, table-number-alignment=right]@{\extracolsep{0pt}~}l@{\extracolsep{\fill}\hspace{2\tabcolsep}}>{\centering\arraybackslash\hspace*{2\pwidth}}p{\lagcol}@{\extracolsep{0pt}}>{\centering\arraybackslash\hspace*{2\pwidth}}p{\lagcol}@{}}
\toprule
\multicolumn{9}{@{}l}{\textit{Panel A: matched trades by layer}} \\
\addlinespace[2pt]
 & \multicolumn{7}{c}{Number of executed windows} & & \\
\cmidrule{2-8}
 & & \multicolumn{6}{c}{With any address trading every leg} & \multicolumn{2}{c}{Lag until every leg trades} \\
\cmidrule{3-8}\cmidrule{9-10}
Layer & {All executed} & \multicolumn{2}{c}{\shortstack{\strut Same direction\\\strut (matched trades)}} & \multicolumn{2}{c}{\shortstack{\strut Opposite\\\strut direction}} & \multicolumn{2}{c}{\shortstack{\strut Both\\\strut directions}} & {Median (s)} & {p75 (s)} \\
\midrule
Venue-asserted & 64681 & 28507 & (44.1\%) & 4366 & (6.8\%) & 3631 & (5.6\%) & 2.7 & \phantom{0}4.8 \\
Venue-derived & 31317 & 4554 & (14.5\%) & 1776 & (5.7\%) & 1371 & (4.4\%) & 1.0 & \phantom{0}3.5 \\
Semantic-inferred & 29097 & 2574 & (8.8\%) & 88 & (0.3\%) & 71 & (0.2\%) & 3.1 & 17.9 \\
\quad Parsers & 21744 & 2265 & (10.4\%) & 78 & (0.4\%) & 62 & (0.3\%) & 3.1 & 16.6 \\
\quad LLMs & 7353 & 309 & (4.2\%) & 10 & (0.1\%) & 9 & (0.1\%) & 3.1 & 74.4 \\
\midrule
All & 125095 & 35635 & (28.5\%) & 6230 & (5.0\%) & 5073 & (4.1\%) & 2.6 & \phantom{0}4.8 \\
\midrule[0.08em]
\end{tabular*}
\begin{tabular*}{\textwidth}{@{\extracolsep{\fill}}lS[table-format=3.0]S[table-format=5.0, table-fixed-width=true, table-column-width=\anyaddress]S[table-format=5.0, table-fixed-width=true, table-column-width=\pairtop, table-number-alignment=right]@{\extracolsep{0pt}~}l@{}}
\multicolumn{5}{@{}l}{\textit{Panel B: wallet addresses by layer}} \\
\addlinespace[2pt]
 & & \multicolumn{3}{c}{Number of executed windows with a matched trade} \\
\cmidrule{3-5}
Layer & {Number of distinct addresses} & {By any address} & \multicolumn{2}{c}{By top 10} \\
\midrule
Venue-asserted & 666 & 28507 & 24344 & (85.4\%) \\
Venue-derived & 390 & 4554 & 3399 & (74.6\%) \\
Semantic-inferred & 151 & 2574 & 2272 & (88.3\%) \\
\quad Parsers & 95 & 2265 & 2122 & (93.7\%) \\
\quad LLMs & 83 & 309 & 176 & (57.0\%) \\
\midrule
All & 854 & 35635 & 27620 & (77.5\%) \\
\bottomrule
\end{tabular*}
\begin{tablenotes}[flushleft]
\footnotesize
\item\!\textit{Notes:} A matched trade requires one taker address to trade every leg in the portfolio's direction. Panel A
percentages use all executed windows; Panel B percentages use windows matched by any address. Lag
runs from opening until every leg has a same-direction fill. Top 10 addresses are ranked by matched
windows within each layer and, for the two All rows, across layers.
\end{tablenotes}
\end{threeparttable}
\vspace{2.5pt}
\end{table}

The preceding analysis in Sections~\ref{sec:implementation}--\ref{sec:limits}
measures what a hypothetical arbitrageur could earn. \linebreak
We now use on-chain fills to examine whether actual trades follow the detected portfolios.
For each executed window, we collect fills on its contracts from 5 seconds before it opens to 30
seconds after it closes.  A taker trades in the portfolio's direction when it buys \YES{} or sells
\NO{} on a purchased leg, or sells \YES{} or buys \NO{} on a sold leg.
A window has a \emph{matched trade} when the same taker address has trades on every leg of the
portfolio, with each trade in the portfolio's direction on that leg.
Table~\ref{tab:sec5-fill-match-v4} (Panel~A) reports the percentage of executed windows with a
matched trade (same direction).  We provide a placebo benchmark for two-sided trading by repeating
the match with every leg's direction reversed.
Windows matched in both directions enter both counts; the two matches may involve different taker
addresses.  These matches establish directional trading activity, but do not by themselves establish
execution of the portfolio at the prices and quantities used in our calculations. \looseness=-1

The excess over the placebo is largest in venue-asserted groups.
In 44\% of these windows there is a matched trade, compared with about 7\% in the placebo.
Venue-derived relations show a much smaller difference: about 15\% of windows have a matched trade,
compared with~6\% in the placebo.  About 30\% of these same-direction matched venue-derived
windows also have an opposite-direction match.
The smaller excess over the placebo is consistent with ordinary two-sided betting flow substantially
contributing to the matches.  For semantic-inferred relations, the matching rate is lower but still
exceeds the placebo.  Only 9\% have a matched trade, compared with less than 1\% in the placebo.
Within this layer, 10.4\% of parser windows have a matched trade, but only 4.2\% of LLM windows. \looseness=-1

Participation is also concentrated (Panel~B of Table~\ref{tab:sec5-fill-match-v4}):
across all layers, about 850 addresses make matched trades, and the 10 most active trade every leg
in about 78\% of the matched windows.  Concentration is high within each layer as well:
the 10 most active addresses in each layer account for 85\% of matched windows in venue-asserted
relations, 75\% in venue-derived relations, and 88\% in semantic-inferred relations.

These trading patterns help interpret the persistence results.
Where the venue publishes the relation and, for sales, offers basket conversion, trades follow the
arbitrage direction more often and opportunities disappear quickly.
Among opportunities lasting at least 0.3 seconds, those based on semantic-inferred relations persist
longest, yet few have one address trading every leg.
The duration regressions also associate additional legs with longer windows, consistent with the
difficulty of assembling trades across separate books.
The evidence is thus consistent with both discovery and execution costs limiting arbitrage,
even when displayed prices offer a profit after spreads and fees. 

\section{Conclusion}
\label{sec:conclusion}

We develop a framework for combinatorial arbitrage among logically related prediction-market contracts,
linking constraints on their joint payoffs to executable portfolios.
We recover relations from venue mechanisms, typed contract fields, and natural-language resolution
criteria, and compile them into a common no-arbitrage program.
Full-depth order books then determine how much can be traded and the profit guaranteed after spreads
and fees, conditional on valid relations and the modeled payoff states.

The Polymarket data reveal substantial executable opportunities.
Among settled windows, an arbitrageur accumulates \$2.50 million in guaranteed profit,
assuming execution at displayed quotes and that subsequent quotes remain as observed.
Restricting the arbitrageur to the first execution in each window yields \$1.10 million.
Much of the accumulated take comes from a small proportion of high-return windows.
Persistence also differs across relation layers.
Among windows lasting at least 0.3 seconds, median duration rises from 2.9 seconds for
venue-asserted relations to about 31 seconds for semantic-inferred relations.
These differences remain after controlling for portfolio characteristics and time effects.

The results suggest that finding the payoff relation is an important part of arbitrage.
Opportunities based on groups the venue publishes disappear quickly and more often have one
address trading every leg.  Relations recovered from contract text are associated with longer
opportunities and less observed participation.
Because the layers also differ in liquidity and execution mechanisms, these comparisons do not
isolate discovery costs from execution costs.
They show how both can matter in markets where related claims trade in separate order books.
\clearpage

\begingroup
\onehalfspacing
\bibliographystyle{jf-refs}
\bibliography{references_upshape}
\endgroup
\clearpage

\appendix
\numberwithin{equation}{section}
\numberwithin{figure}{section}
\numberwithin{table}{section}
\makeatletter
\renewcommand{\@seccntformat}[1]{%
  \csname appendix@prefix@#1\endcsname \csname the#1\endcsname \csname
  appendix@suffix@#1\endcsname\quad}
\newcommand{\appendix@prefix@section}{Appendix~}
\newcommand{\appendix@suffix@section}{}
\makeatother

\section{Arbitrage Theory and Detection}
\label{app:theory}

\subsection{Proofs and Remark}
\label{app:proofs}

\begin{proof}[Proof of Footnote~\ref{fn:pairwise}]
The claim is that two sets $P$ and $Q$ are equal for every $S\subseteq I$, where
\begin{equation}
\label{eq:pairwise-exclusivity}
P:=\{\bm{z}\in\{0,1\}^{I}:\ z_i+z_j\le 1\ \text{for all } i<j \text{ in } S\}, \quad
Q:=\{\bm{z}\in\{0,1\}^{I}:\ \textstyle\sum_{i\in S} z_i\le 1\}.
\end{equation}
The inclusion $Q\subseteq P$ is immediate: if $\sum_{i\in S}z_i\le 1$, then
$z_i+z_j\le\sum_{k\in S}z_k\le 1$ for every pair $i<j$, since every $z_k\ge 0$.  For the reverse
inclusion, suppose, for contradiction, that some $\bm{z}\in P$ has $\sum_{i\in S}z_i\ge 2$.
Since each $z_k$ is $0$ or $1$, at least two coordinates $i\ne j$ in $S$ equal one, so
$z_i+z_j=2>1$, contradicting $\bm{z}\in P$.  This step requires integrality:
on $[0,1]^{I}$ the two sets $P$ and $Q$ differ, 
as $(\tfrac12,\tfrac12,\tfrac12)$ is a counterexample for $|S|=3$.
\end{proof}

\begin{proof}[Proof of Proposition~\ref{prop:coherence}]
If $\bm{p} \in M$ then $\bm{p} = \mathbb{E}_q[\bm{\phi}]$ for some belief $q$.
Hence the initial cost of every portfolio equals its expected terminal payoff,
$\bm{p}^\top\bm{y} = \mathbb{E}_q[\bm{y}^\top\bm{\phi}]$, so its expected profit under $q$ is zero.
Its profit therefore cannot be strictly positive in every state, and \eqref{eq:arb} has no solution.
If $\bm{p} \notin M$, then since $M$ is compact and convex, the separating-hyperplane theorem gives a
direction $\bm{y}$ with
$\bm{p}^\top\bm{y} < \min_{\bm{\mu} \in M} \bm{y}^\top\bm{\mu} = \min_{\omega} \bm{y}^\top
\bm{\phi}(\omega)$.  The portfolio's worst-case terminal payoff therefore exceeds its initial cost,
so its guaranteed profit is strictly positive: an arbitrage exists.
\end{proof}

\begin{proof}[Proof of Proposition~\ref{prop:l1}]
A linear function attains the same minimum over a finite point set and over its convex hull,
so the worst-case terminal payoff satisfies
$\min_{\omega}\bm{y}^\top\bm{\phi}(\omega)=\min_{\bm{\mu}\in M}\bm{y}^\top\bm{\mu}$.  Hence
\begin{equation}
\label{eq:proof-l1-duality}
\begin{aligned}
\max_{\|\bm{y}\|_\infty\le1}\;\min_{\omega\in\Omega}\big\{\bm{y}^\top\bm{\phi}(\omega)-\bm{p}^\top\bm{y}\big\}
&=\max_{\|\bm{y}\|_\infty\le1}\;\min_{\bm{\mu}\in M}\;\bm{y}^\top(\bm{\mu}-\bm{p})\\
&=\min_{\bm{\mu}\in M}\;\max_{\|\bm{y}\|_\infty\le1}\;(\bm{\mu}-\bm{p})^\top\bm{y}
=\min_{\bm{\mu}\in M}\|\bm{\mu}-\bm{p}\|_1,
\end{aligned}
\end{equation}
where the exchange of maximum and minimum follows from linear-programming duality on compact convex
sets, and coordinatewise maximization gives the $\ell_1$ norm.
Let $\bm{y}^{\medstar}$ be a maximin optimizer and $\bm{\mu}^{\medstar}$ a minimizing coherent price vector.
The minimax equality implies that they form a saddle point, so $y_i^{\medstar}=+1$ where
$p_i<\mu_i^{\medstar}$ and $y_i^{\medstar}=-1$ where $p_i>\mu_i^{\medstar}$.
When $p_i=\mu_i^{\medstar}$, the corresponding position is chosen jointly with the other positions so
that $\bm{y}^{\medstar\top}(\bm{\mu}-\bm{\mu}^{\medstar})\ge0$ for every $\bm{\mu}\in M$;
it need not be zero.
\end{proof}

\begin{proof}[Proof of Proposition~\ref{prop:executable}]
Proposition~\ref{prop:executable} is a special case of the following proposition:
\begingroup
\newtheoremstyle{tightprop}
  {\topsep}
  {\topsep}
  {\itshape}{}{\bfseries}{.}{0.5em}{}
\theoremstyle{tightprop}
\newtheorem{manualprop}{Proposition}
\renewcommand{\themanualprop}{3$^+$}
\begin{manualprop}
\label{prop:executable-general}
Let $K:\mathbb{R}^{I}\to\mathbb{R}\cup\{+\infty\}$ be a proper closed convex function, with $K(\bm{0})=0$
and a compact domain $D=\operatorname{dom}K$.  Let
$K^{\asterisk}(\bm{\mu}):=\sup_{\bm{y}\in\mathbb{R}^{I}}\{\bm{\mu}^\top\bm{y}-K(\bm{y})\}$ denote its convex
conjugate, and let \linebreak
\(\partial K(\bm{0}):=\{\bm{\mu}\in\mathbb{R}^{I}: \bm{\mu}^\top\bm{y}\le K(\bm{y})-K(\bm{0})
\text{ for every }\bm{y}\in\mathbb{R}^{I}\}\) denote its subdifferential at zero.  Then
\[
  \Pi^{\medstar}:=\max_{\bm{y}\in D}\min_{\bm{z}\in\Phi}\big\{\bm{y}^\top\bm{z}-K(\bm{y})\big\}
\]
is nonnegative, is attained, and is equal to $\min_{\bm{\mu}\in M}K^{\asterisk}(\bm{\mu})$.  Moreover,
$\Pi^{\medstar}>0$ if and only if $M\cap\partial K(\bm{0})=\emptyset$.
\end{manualprop}
\endgroup
We first use Proposition~\ref{prop:executable-general} to prove Proposition~\ref{prop:executable},
and prove Proposition~\ref{prop:executable-general} itself later.
The order-book cost function is $K(\bm{y})=\sum_{i\in I}K_i(y_i)$, 
where $K_i$ can be formally written as
\begin{equation}
\label{eq:contract-cost}
K_i(y_i):=
\left\{
\begin{array}{@{}l@{\quad}l@{}}
\int_0^{y_i} c_i(s)\,ds
  & -Q^{\bid}_i\le y_i\le Q^{\ask}_i,\\[0.5ex]
+\infty
  & \text{otherwise}.
\end{array}
\right.
\end{equation}
Here $c_i(s)$ denotes the price of the $s$-th share traded in contract~$i$,
and is a step function:
\begin{equation}
\label{eq:marginal-price}
c_i(s):=
\left\{
\begin{array}{@{}l@{\qquad}l@{}l@{}l@{}l@{\quad}l@{}}
\tilde a_{i\ell}
  & \text{if }  & Q^{\ask}_{i,\ell-1} & <   s\le & Q^{\ask}_{i\ell}
  & \text{(buying)},\\[0.5ex]
\tilde b_{i\ell}
  & \text{if } -& Q^{\bid}_{i\ell}    & \le s<  -& Q^{\bid}_{i,\ell-1}
  & \text{(selling)},
\end{array}
\right.
\end{equation}
where \(Q^{\bid}_{i\ell}:=\medsum\nolimits_{k=1}^{\ell} q^{\bid}_{ik}\)
and \(Q^{\ask}_{i\ell}:=\medsum\nolimits_{k=1}^{\ell} q^{\ask}_{ik}\)
are the depths up to the $\ell$-th level of bids and asks, respectively,
with the convention $Q^{\bid}_{i0}=Q^{\ask}_{i0}=0$.

By~\eqref{eq:contract-cost} and~\eqref{eq:marginal-price}, $K$ is proper, closed and convex, vanishes at zero, and has the compact
domain~\eqref{eq:executable-set}, and thus Proposition~\ref{prop:executable-general} applies.  Because $K$ is a sum of functions of
separate coordinates, \linebreak $\partial K(\bm{0})$ is the product of the sets $\partial K_i(0)$.  Each
$K_i$ is piecewise linear, and at zero its right-hand and left-hand derivatives are $(K_i)^{\prime}_{+}(0)=\tilde a_{i1}$
and $(K_i)^{\prime}_{-}(0)=\tilde b_{i1}$,
so $\partial K_i(0)=[\tilde b_{i1},\tilde a_{i1}]$~\citep{rockafellar1970convex,hiriarturruty1993convex}.
The condition $M\cap\partial K(\bm{0})=\emptyset$ is therefore condition~\eqref{eq:exec-no-arb}.
\end{proof}

\begin{proof}[Proof of Proposition~\ref{prop:executable-general}]
Fix a portfolio $\bm{y}\in D$.  Since $K(\bm{y})$ is its initial cost and does not depend on the
realized state,
\begin{equation}
\label{eq:proof-executable-state}
  \min_{\bm{z}\in\Phi}\big\{\bm{y}^\top\bm{z}-K(\bm{y})\big\}
  =\min_{\bm{z}\in\Phi}\bm{y}^\top\bm{z}-K(\bm{y})
  =\min_{\bm{\mu}\in M}\{\bm{y}^\top\bm{\mu}-K(\bm{y})\}.
\end{equation}
The second equality uses $M=\operatorname{conv}(\Phi)$: a linear functional has the same minimum
over the payoff vectors and over their convex hull.  Hence
\begin{equation}
\label{eq:proof-executable-maxmin}
  \Pi^{\medstar}=\max_{\bm{y}\in D}\min_{\bm{\mu}\in M}\{\bm{y}^\top\bm{\mu}-K(\bm{y})\}.
\end{equation}
Because $K$ is closed, this objective is upper semicontinuous in $\bm{y}$ on the compact set $D$, so
its maximum is attained.  The portfolio $\bm{y}=\bm{0}$ lies in $D$ and has value
$-K(\bm{0})=0$, so $\Pi^{\medstar}\ge0$.

Now $M$ and $D$ are compact and convex, and on $M\times D$ the function
\begin{equation}
\label{eq:proof-executable-objective}
  h(\bm{\mu},\bm{y}):=\bm{\mu}^\top\bm{y}-K(\bm{y})
\end{equation}
is continuous and affine in $\bm{\mu}$, and upper semicontinuous and concave in $\bm{y}$ because $K$
is closed and convex.  The minimax theorem therefore allows the order of optimization to be
exchanged:
\begin{equation}
\label{eq:proof-executable-minimax}
  \Pi^{\medstar}=\min_{\bm{\mu}\in M}\max_{\bm{y}\in D}\{\bm{\mu}^\top\bm{y}-K(\bm{y})\}
  =\min_{\bm{\mu}\in M}K^{\asterisk}(\bm{\mu}),
\end{equation}
where the last equality holds because $K(\bm{y})=+\infty$ outside $D$.

\textls[-5]{It remains to characterize when $\Pi^{\medstar}\mkern-1mu \mkern-1mu = \mkern-1mu 0$.
Since $K(\bm{0}) \mkern-1mu = \mkern-1mu 0$, every $\bm{\mu}$ satisfies
$K^{\asterisk}\mkern-1mu(\bm{\mu})
\mkern-1mu\ge\mkern-1mu \bm{\mu}^\top\bm{0} \mkern-1mu-\mkern-1mu K(\bm{0})
\mkern-1mu = \mkern-1mu 0$.} \linebreak
Equality holds exactly when
$\sup_{\bm{y}}\{\bm{\mu}^\top\bm{y}-K(\bm{y})\}\le 0$, which is equivalent to 
$\bm{\mu}^\top\bm{y}\le K(\bm{y})$ for every $\bm{y}$, \linebreak
or equivalently, when $\bm{\mu}\in\partial K(\bm{0})$.
Because $D$ is compact, $K^{\asterisk}$ is finite everywhere; as a conjugate, $K^{\asterisk}$ is always convex. 
Hence $K^{\asterisk}$ is continuous,  so its minimum over the compact set $M$ is attained.
Therefore $\Pi^{\medstar}=0$ if and only if some $\bm{\mu}\in M$ has
$K^{\asterisk}(\bm{\mu})=0$, that is, if and only if $M\cap\partial K(\bm{0})\neq\emptyset$.
\end{proof}

\begin{remark}[The quoted-price benchmark as a special case]
For a coherent valuation $\bm{\mu}\in M$, $K^{\asterisk}(\bm{\mu})$ is the largest gain from establishing a
portfolio at the displayed order book and valuing it at $\bm{\mu}$, and
Proposition~\ref{prop:executable-general} states that $\Pi^{\medstar}$ is this gain under the coherent
valuation that makes it smallest.  The linear quoted-price model of Section~\ref{sec:coherent} is
the special case with cost
\begin{equation}
\label{eq:linear-cost}
K^{\textrm{quoted}}(\bm{y}):=\bm{p}^\top\bm{y}+\iota_{\{\bm{u}: \|\bm{u}\|_\infty\le1\}}(\bm{y}),
\end{equation}
where $\iota_C$ is $0$ on the set $C$ and $+\infty$ outside it.  This $K^{\textrm{quoted}}$ is proper, closed and convex,
vanishes at zero, and has the compact domain $\{\bm{u}: \|\bm{u}\|_\infty\le1\}$, so
Proposition~\ref{prop:executable-general} applies.

The subdifferential of $K^{\textrm{quoted}}$ at zero is
\begin{equation}
\label{eq:linear-subdifferential}
\partial K^{\textrm{quoted}}(\bm{0})=\{\bm{\mu}:(\bm{\mu}-\bm{p})^\top\bm{y}\le0\text{ for all }
\|\bm{y}\|_\infty\le1\}=\{\bm{p}\},
\end{equation}
so $\Pi^{\medstar}>0$ if and only if $\bm{p}\notin M$, which is Proposition~\ref{prop:coherence}.
The conjugate of $K^{\textrm{quoted}}$ is
\begin{equation}
\label{eq:linear-conjugate}
(K^{\textrm{quoted}})^{\asterisk}(\bm{\mu})=\sup_{\|\bm{y}\|_\infty\le1}(\bm{\mu}-\bm{p})^\top\bm{y}=\|\bm{\mu}-\bm{p}\|_1,
\end{equation}
so $\Pi^{\medstar}=\min_{\bm{\mu}\in M}\|\bm{\mu}-\bm{p}\|_1$, which is Proposition~\ref{prop:l1}.

In Proposition~\ref{prop:l1}, the normalization $\|\bm{y}\|_\infty \le 1$ fixes only the scale:
in the linear quoted-price model, an arbitrageur could otherwise increase a profitable position
without bound.  When the distance is zero, $\bm{p}\in M$, so no arbitrage exists and some belief $q$
rationalizes the quoted prices.  When the distance is positive, an arbitrageur buys contracts priced
below the minimizing coherent price vector and sells those priced above it.
This characterization is the quoted-price analogue of Proposition~2.4 in \citet{kroer2016arbitrage},
with the $\ell_1$ distance in place of their Bregman divergence.
\end{remark}

\begin{proof}[Proof of Proposition~\ref{prop:underidentification}]
Since $\Phi\subseteq\Phi^{-}$, minimizing the same net-payoff function over $\Phi^{-}$ can only give
a weakly smaller value than minimizing it over $\Phi$.
Maximizing both sides over the common executable set $D$ preserves the inequality.
Strict positivity for a particular portfolio under $\Phi^{-}$ therefore implies strict positivity
under $\Phi$.
\end{proof}

\subsection{The Cutting-Plane Detector}
\label{app:portfolio}

It is stated in Proposition~\ref{prop:coherence} that an arbitrage exists exactly when $\bm{p}$ lies
outside the convex hull~$M$, 
and in Proposition~\ref{prop:l1} that,
$V(\bm{y})=\min_{\bm{\mu}\in M}\bm{y}^\top\bm{\mu}$ is the worst-case terminal payoff,
and the normalized portfolio with the largest guaranteed profit attains
\begin{align}\label{eq:whereM}
  \max_{\|\bm{y}\|_\infty \le 1}\; \big[\, V(\bm{y}) - \bm{p}^\top\bm{y} \,\big]
  \;=\; \min_{\bm{\mu} \in M}\; \|\bm{\mu} - \bm{p}\|_1 .
\end{align}
Computing either side of \eqref{eq:whereM} by enumeration can be infeasible, though. In particular,
the right-hand side is compact in form but computationally demanding:
$M$ is the convex hull of up to $2^{|I|}$ admissible payoff vectors, and can have exponentially many extreme points,
so evaluating the minimized $\ell_1$ distance by enumeration is 
not necessarily easier than solving the maximin on the left. \looseness=-1

\paragraph{The \textsc{Master} problem.}
To avoid enumerating all admissible payoff vectors, we introduce an auxiliary variable \(\eta\) \citep{boyd2004convex} as a lower bound on the portfolio's terminal payoff. \linebreak For a working set \(\Phi_{\mathrm{work}}\subseteq\Phi\), the \textsc{Master} problem is
\[
\max_{\bm y,\eta}\ {\eta-\bm p^\top\bm y}
\quad\text{subject to}\quad
\bm z^\top\bm y\geq\eta
\quad\text{for all }\bm z\in\Phi_{\mathrm{work}},
\qquad \lVert\bm y\rVert_\infty\leq1.
\]
The restricted \textsc{Master} problem provides an upper bound on~\eqref{eq:whereM}.
Its solution is globally optimal once the \textsc{Oracle} certifies that no payoff constraint is violated.
Replacing the linear cost and unit box \linebreak with the order-book cost $\sum_{i\in I}K_i(y_i)$ preserves linearity:
since each \(K_i\) is convex and piecewise linear, \linebreak
its execution cost can be represented by quantity variables at individual price levels,
each bounded by the displayed depth and assigned the corresponding price and fee.
The optimum fills these levels in ladder order.
Notably, the exception is that the 5-leg cap of Section~\ref{sec:implementation} introduces binary indicators \linebreak
\(x_i\in\{0,1\}\), with \(-Q_i^b x_i\leq y_i\leq Q_i^a x_i\) and \(\sum_{i\in I}x_i\leq5\),
making the \textsc{Master} problem a mixed-integer linear program (MILP).
These indicators are needed only for the detection in \textbf{Step~1} of Section~\ref{sec:implementation};
the reoptimization in \textbf{Step~2} fixes the selected legs and their directions,
and therefore remains a linear program over the full-depth ladders. \looseness=-1

\paragraph{The \textsc{Oracle} problem.}\!\!
Given a \textsc{Master} solution \((\bm y,\eta)\), the \textsc{Oracle} solves the integer program~\eqref{eq:worst-case-payoff} \linebreak to find the minimum terminal payoff over \(\Phi\). If this minimum is below \(\eta\), the minimizing payoff vector violates a \textsc{Master} constraint and is added to \(\Phi_{\mathrm{work}}\). Otherwise, the candidate is feasible for the full problem and hence globally optimal. This constraint-generation procedure \citep{kelley1960cutting,grotschel1988geometric} terminates finitely, since each violated constraint contributes a previously unseen payoff vector. Moreover, a zero \textsc{Master} optimum already rules out arbitrage, without requiring an \textsc{Oracle} call.

For exceptionally large active components, one could trade detection power for computational efficiency by replacing the integer \textsc{Oracle} with its LP relaxation \citep{nemhauser1988integer,wolsey2020integer}---$\min_{\bm{z}\in\Phi_{\mathrm{frac}}}\bm{y}^\top\bm{z}$, where $\Phi_{\mathrm{frac}}=\{\bm{z}\in[0,1]^{I}:F\bm{z}\leq\bm{g}\}$ is the LP-relaxation polytope---an alternative approach not pursued in this study.
Since \(\Phi\subseteq\Phi_{\mathrm{frac}}\), the relaxed minimum provides a lower bound on the true worst-case payoff. As in Proposition~\ref{prop:underidentification}, this relaxation may cause genuine arbitrage opportunities to be missed, but any opportunities detected remain valid. 

\enlargethispage{.25\baselineskip}

\subsection{Early Liquidation}
\label{app:early-liquidation}

The quoted-price benchmark also permits the arbitrageur to exit before resolution:
for the optimal portfolio of Section~\ref{sec:coherent}, any future coherent price vector
$\bm{p}_\tau\in M$ yields
\begin{equation}
\label{eq:early-liquidation}
\bm{y}^{\medstar\top}(\bm{p}_\tau-\bm{p})
\ge
\min_{\bm{\mu}\in M}\bm{y}^{\medstar\top}(\bm{\mu}-\bm{p}),
\end{equation}
so an arbitrageur who liquidates at these prices earns at least the profit guaranteed at resolution.

Extending this conclusion to a finite-depth order book requires the future liquidation order book. \linebreak
If $K_0(\bm{y})$ is the entry cost and $K_\tau$ the execution-cost function of the time-$\tau$ order
book, closing at $\tau$ produces liquidation profit $-K_0(\bm{y})-K_\tau(-\bm{y})$,
which depends on future spreads, depth, and fees that are unknown at entry.
Resolution avoids this uncertainty because the contracts pay directly according to $\bm{z}\in\Phi$,
requiring no further market transaction.
We therefore treat resolution as the guaranteed exit; earlier liquidation realizes the arbitrage
sooner only when the future order book supports it. \looseness=-1

\subsection{Additional Logical Relations}
\label{app:logical-examples}

In addition to Figure~\ref{fig:sell-buy-asymmetry}, Figure~\ref{fig:additional-relations}
illustrates partition and conjunction relations in three dimensions.

\begin{figure}[!h]
\centering
\includegraphics[width=\textwidth]{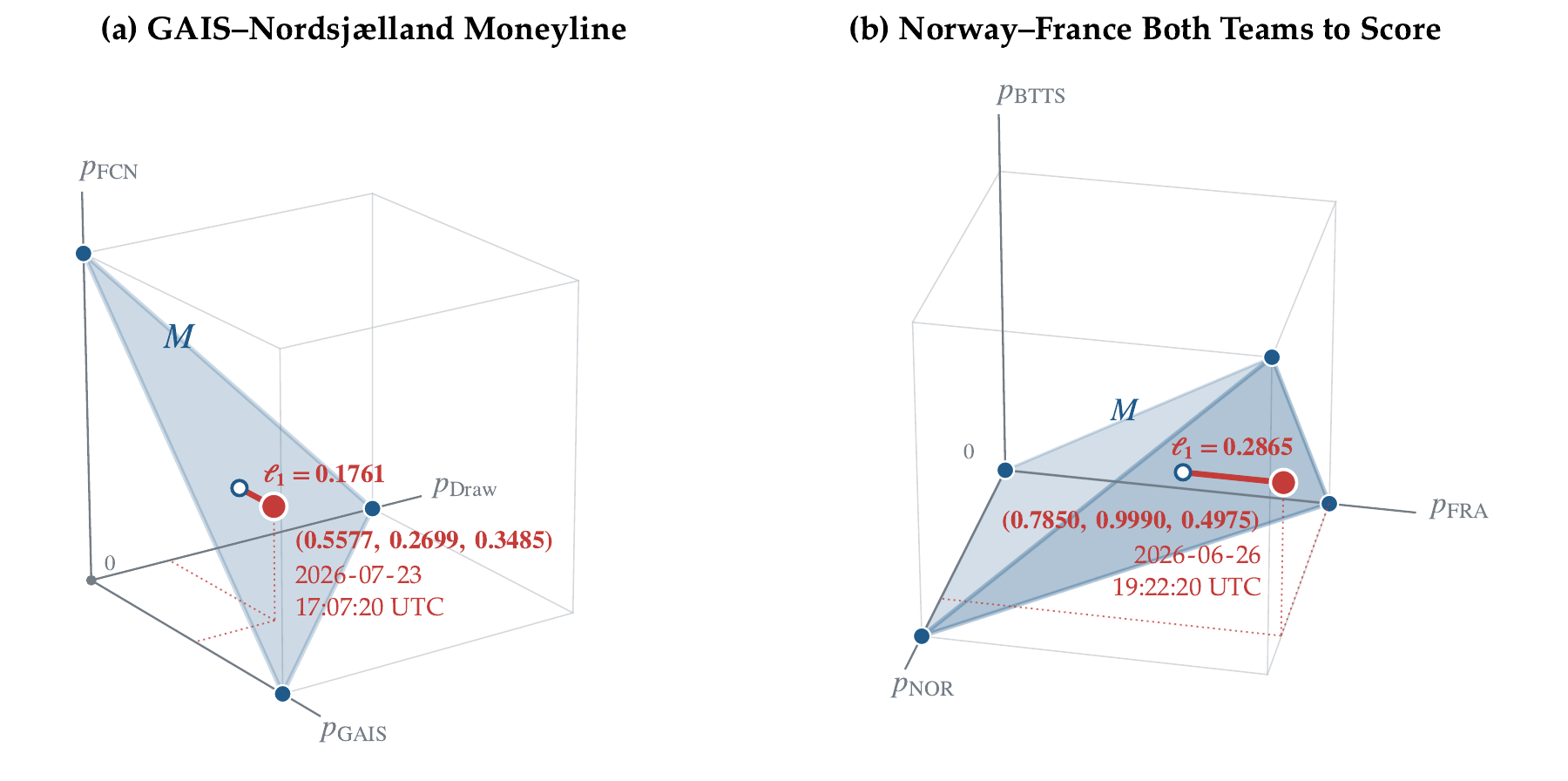}
\caption{Price coherence for partition and conjunction relations.
Panel (a) shows the GAIS--Nordsj\ae{}lland win and draw contracts;
panel (b) shows the Norway--France both-teams-to-score relation.
Red segments mark the $\ell_1$ distance to coherence.}
\label{fig:additional-relations}
\end{figure}

Panel (a) of Figure~\ref{fig:additional-relations} shows three \YES{} contracts
``\href{https://polymarket.com/event/col-gai-nor-2026-07-23}{GAIS win},''
``\href{https://polymarket.com/event/col-gai-nor-2026-07-23}{Draw},''
and
``\href{https://polymarket.com/event/col-gai-nor-2026-07-23}{FC Nordsj\ae{}lland win},'' 
with Polymarket market IDs 
\href{https://gamma-api.polymarket.com/markets/2974928}{2974928}, 
\href{https://gamma-api.polymarket.com/markets/2974929}{2974929}, and 
\href{https://gamma-api.polymarket.com/markets/2974930}{2974930}.
The contracts form an exhaustive partition, so their joint payoff vector can take the values
$(1,0,0)$, $(0,1,0)$, and $(0,0,1)$, corresponding to a GAIS win, a draw, and a Nordsj\ae{}lland win, respectively.
Their convex hull is the blue triangular coherent set~$M$. 
At 17:07:20 UTC on 23 July 2026, the fee-adjusted bid vector $(0.5577,0.2699,0.3485)$ lies outside
$M$ at $\ell_1$ distance $0.1761$, the red segment along the GAIS axis:
an arbitrageur who sells one share of each contract receives \$1.1761 initially and pays at most one
dollar at resolution, guaranteeing \$0.1761 after fees.

Panel (b) considers the contracts
``\href{https://polymarket.com/sports/world-cup/fifwc-nor-fra-2026-06-26}{Norway Over 0.5 goals},''
``\href{https://polymarket.com/sports/world-cup/fifwc-nor-fra-2026-06-26}{France Over 0.5 goals},''
and
``\href{https://polymarket.com/sports/world-cup/fifwc-nor-fra-2026-06-26}{Both Teams to Score},''
with Polymarket market IDs 
\href{https://gamma-api.polymarket.com/markets/2481818}{2481818}, 
\href{https://gamma-api.polymarket.com/markets/2481821}{2481821}, and 
\href{https://gamma-api.polymarket.com/markets/2323839}{2323839}.
The third contract is the conjunction of the first two, so their joint payoff vector can take the
values $(0,0,0)$, $(1,0,0)$, $(0,1,0)$, and $(1,1,1)$, whose convex hull is the blue tetrahedral
coherent set $M$.  At 19:22:20 UTC on 26 June 2026, the fee-adjusted bids for the two team totals
and the fee-adjusted ask for the both-teams-to-score contract form the execution quote vector
$(0.7850,0.9990,0.4975)$, which violates the Fr\'echet lower bound $0.7850+0.9990-1=0.7840$ on the
third coordinate; the red segment marks the $\ell_1$ distance $0.2865$ to $M$.
At the best quotes, an arbitrageur can guarantee \$0.2865 by selling one share of each of the first
two contracts and buying one share of the both-teams-to-score contract.
Trading against the 32 shares displayed yields total guaranteed profit of \$9.17 after fees.

\subsection{Two Sides of Negative-Risk Partitions}
\label{app:partition-sides}

\enlargethispage*{.5\baselineskip}

Polymarket describes negative risk as 
``\href{https://docs.polymarket.com/concepts/negative-risk}%
{a mechanism for multi-outcome events where only one outcome can win},''
and its adapter converts a \NO{} token of any contract into one \YES{} token of every other
contract in the same negative-risk group of $n$ contracts;
it neither converts the resulting bundle of \linebreak \YES{} tokens back into the original \NO{} token,
nor does it convert a \YES{} token---together with $(n-2)$ dollars in collateral---into one \NO{} token
of every other contract.\footnote{In general, Polymarket's negative-risk adapter
allows a trader to convert $m$ \NO{} tokens with $0\le m \le n$
(by convention, $m=0$ mints one \YES{} token per contract;
see also Footnote~\ref{fn:one-exception}) into $(m-1)$ dollars in
collateral and one \YES{} token for each of the
remaining contracts, denoted by \NOtoYES{},
but not the reverse, that is, it does not convert $m$ \YES{} tokens \linebreak
\textls[-5]{($m\ge0$) together with $(n-m-1)$ dollars in collateral
into the remaining $(n-m)$ \NO{} tokens, denoted \YEStoNO{}} \linebreak
(\href{https://github.com/Polymarket/neg-risk-ctf-adapter/blob/%
main/docs/NegRiskAdapter.md\#convertpositions}{Polymarket (GitHub), \texttt{NegRiskAdapter.md}}).
By the split--merge identity $\YES{}_j+\NO{}_j=\$1$,
the two unsupported $m$-\YES-\linebreak{}to-$(n\!-\!m)$-\NO{}
conversions ($m$ takes the values $n-1$ and $1$ here, respectively)
would yield the same payoff surplus:
\vspace{-5pt}\begin{equation}\vspace{-5pt}
\label{eq:why-only-NotoYES}
\NO{}_i-\sum\nolimits_{j\ne i}\YES{}_j
=\sum\nolimits_{j\ne i}\NO{}_j-\big[\YES{}_i+\$(n-2)\big]
=\$1-\sum\nolimits_{j=1}^{n}\YES{}_j.
\end{equation}
This surplus arises precisely when no contract wins,
reflecting the same failure of exhaustiveness.}

The asymmetry follows from what the adapter guarantees: at most one contract resolves \YES{}, \linebreak
possibly none.
A \NO{} token therefore pays  
no less than the other contracts' \YES{} tokens combined
in every admissible state,
with the surplus when no contract wins,
so it can be exchanged for them; \linebreak
\textls[-5]{the reverse exchange would pay out more than it takes in when no contract wins,
so it is not offered.}\looseness=-1\footnote{The forward exchange is safe for the adapter, 
because in no admissible state does it pay out more than
it takes in, and the reverse exchange is not.
Both claims follow directly from~\eqref{eq:why-only-NotoYES}.
Verbally, if the \NO{} token pays zero, its \YES{} counterpart pays one, 
forcing all the other \YES{} tokens to pay zero;
if it pays one, the other \YES{} tokens sum to at most one, 
and to zero when no contract wins, 
which implies a one-dollar payoff loss for the trader.
The reverse exchange is therefore unsafe for the adapter.}

Trades against resting orders were subject to the applicable taker fees during our sample period,
whereas split, merge, and negative-risk \NOtoYES{} conversion
incurred no protocol fees.\looseness=-1\footnote{Polymarket charges taker fees:
a taker who trades $C$ shares at price $p$ pays $C\,r\,p(1-p)$,
where $r$ is the contract's fee rate. \linebreak
By contrast, 
none of the three operations---split, merge, or
\NOtoYES{} conversion (for negative-risk groups only)---carried a fee 
during the period, as three pieces of evidence show.
First, Polymarket's fee schedule charges a fee only to the taker of a trade, and it lists no fee
for split, merge or conversion
(\href{https://docs.polymarket.com/trading/fees}{Polymarket, ``Fees''}).
Second, the contracts that carry out these operations contain no fee for split or merge,
and the contract that performs the conversion allows only a conversion fee, which is fixed when a
group is created
(\href{https://github.com/Polymarket/neg-risk-ctf-adapter/blob/%
main/src/NegRiskAdapter.sol\#L330-L331}{Polymarket (GitHub),
\texttt{NegRiskAdapter.sol}}); Polygonscan verifies the source code of all three smart contracts
against those deployed on Polygon
(\href{https://polygonscan.com/}{Polygonscan}:
\href{https://polygonscan.com/address/0xAdA100Db00Ca00073811820692005400218FcE1f\#code}{\texttt{CtfCollateralAdapter}},
\href{https://polygonscan.com/address/0xadA2005600Dec949baf300f4C6120000bDB6eAab\#code}{\texttt{NegRiskCtfCollateralAdapter}},
\href{https://polygonscan.com/address/0xd91E80cF2E7be2e162c6513ceD06f1dD0dA35296\#code}{\texttt{NegRiskAdapter}}).
We read this conversion fee from the smart contract
\href{https://polygonscan.com/address/0xd91E80cF2E7be2e162c6513ceD06f1dD0dA35296\#code}{\texttt{NegRiskAdapter}}
on Polygon for each of the 20{,}866 negative-risk groups in
our sample,  and it is zero for every one of them.
Third, in on-chain transactions sampled across the period, every split, merge and conversion
paid out in full, with nothing deducted. \looseness=-1}

A seller of a \YES{} basket can therefore proceed as follows:
split one dollar into one \YES{} and one \NO{} token of any contract,
and then convert the \NO{} token into one \YES{} token of every other contract \linebreak
to mint the basket of one \YES{} token of every contract in the group.
The seller can sell these tokens and collect the proceeds at once,
whereas a buyer of the basket collects its payoff only at resolution.
In practice, the seller needs only a few liquid members whose bids carry the overpricing
to secure (part of) a guaranteed profit.
Alternatively, the seller can buy and hold the corresponding \NO{} tokens until resolution,
earning more than the guaranteed profit if no contract in the non-exhaustive group wins. \looseness=-1

A buyer of a \YES{} basket, by contrast, is covered only after buying every member, one leg each, \linebreak
and only in a group whose rules make exactly one contract win; selling needs only that at most one wins, 
so the seller's profit is guaranteed in every negative-risk group, not only in the partitions.
\clearpage

In the Algeria--Austria group (three contracts:
``\href{https://polymarket.com/sports/world-cup/fifwc-alg-aut-2026-06-27}{Algeria wins}'',
``\href{https://polymarket.com/sports/world-cup/fifwc-alg-aut-2026-06-27}{Draw}'',
and
``\href{https://polymarket.com/sports/world-cup/fifwc-alg-aut-2026-06-27}{Austria wins}'', \linebreak
of which exactly one is \YES{})
of Section~\ref{sec:identify}, 
the two moments below---both in the first half of the match on 28 June 2026 (kick-off at 02:00 UTC)---make the asymmetry concrete.
Purely to keep the arithmetic simple, we use bids and asks before taker fees instead of the
fee-adjusted ones used elsewhere in the paper
(the corresponding Figure~\ref{fig:partition-sides-v2} nonetheless includes these fees).

\emph{The sell side.}  At 02:44:45 UTC the bids of the three \YES{} contracts were $0.16$, $0.31$ and
$0.57$, which sum to $1.04$, more than one dollar.
\begin{itemize}[itemsep=-0.1em, topsep=0em]
\item With the mechanism: split one dollar on
  ``\href{https://polymarket.com/sports/world-cup/fifwc-alg-aut-2026-06-27}{Algeria wins}''
  into one \YES{} and one \NO{} of it, and convert that \NO{} into one \YES{} each of
  ``\href{https://polymarket.com/sports/world-cup/fifwc-alg-aut-2026-06-27}{draw}'' and
  ``\href{https://polymarket.com/sports/world-cup/fifwc-alg-aut-2026-06-27}{Austria wins}'';
  this mints a basket of three \YES{} tokens for one dollar.
  Sell all three \YES{} at the bids for $1.04$ and net $0.04$;
  the dollar comes back on the spot.
\item The same result can be reached without touching the \YES{} tokens: buy the three \NO{}
  directly at their asks of $0.84$, $0.69$ and $0.43$, paying $1.96$ in all, and convert the three
  \NO{} into two dollars in one call, on the spot.  The net gain is also $0.04$, and again nothing
  is held.
\item Without the mechanism: buy the three \NO{} for $1.96$ and hold them to
  resolution, where they pay two dollars.  The gain is the same $0.04$, but the $1.96$ dollars stay
  locked until the match ends. \linebreak
  What the negative-risk mechanism does is to return those $1.96$ dollars early. Splitting one
  dollar on each contract and selling the three \YES{} at the bids comes to exactly the same thing:
  it leaves the same three \NO{} in hand and the same $1.96$ dollars locked.
\end{itemize}

\emph{The buy side.}  At 02:28:24 UTC the asks of the three \YES{} contracts were $0.20$, $0.36$
and $0.35$, which sum to $0.91$, less than one dollar; the bids at the same moment were $0.08$,
$0.32$ and $0.33$.
\begin{itemize}[itemsep=-0.1em, topsep=0em]
\item Buy the three \YES{} for $0.91$ and hold the basket.  At resolution exactly
  one of them pays one dollar, so the gain is $0.09$, but the $0.91$ dollars stay locked until the
  match ends.
\item Selling a basket of \NO{} instead leads to the same position: split one
  dollar on each of the three contracts, paying three dollars for three \YES{} and three \NO{}, and
  sell the three \NO{} at their bids of $0.80$, $0.64$ and $0.65$ for $2.09$ in total.  The net cash outlay is
  still $0.91$, and the holding is still the basket of \YES{}: it is the same trade under another
  name.
\item The only way to cash out early is to sell the three \YES{} back at their bids, which brings
  only $0.73$, a loss of $0.18$; conversion cannot help, because it consumes \NO{} and there is no
  \NO{} in hand.  So the buy side's $0.09$ waits for resolution.
\end{itemize}

\enlargethispage{.25\baselineskip}

Figure~\ref{fig:partition-quotes-v2} draws the two moments (fee excluded): 
the bid and ask prices of each \YES{} and \NO{} token---the 
\NO{} prices mirroring the \YES{} prices---and 
the sums of the \YES{} prices and of the \NO{} prices
against the coherent values one and two, respectively.
Figure~\ref{fig:partition-sides-v2} shows the same two moments with
fees included---the sports fee rate was still $0.03$ on 28 June 2026; Polymarket raised it to $0.05$ on 10 July 2026
(\href{https://docs.polymarket.com/changelog/predictions\#jul-10-2026}{Polymarket, ``Changelog''}).
At 02:44:45 UTC, as Algeria's equalizer reached the books, \linebreak
the fee-adjusted bids summed to $1.0222$; 
the replay minted 68 baskets, sold them for \$69.51 and earned \$1.51 at once.
At 02:28:24 UTC, as Austria scored, the fee-adjusted asks summed to $0.9285$; \linebreak
the replay bought 136 shares of each contract for \$126.28 and earned \$9.72 at resolution. 
\clearpage

\begin{figure}[!p]
\centering
\includegraphics[width=\textwidth]{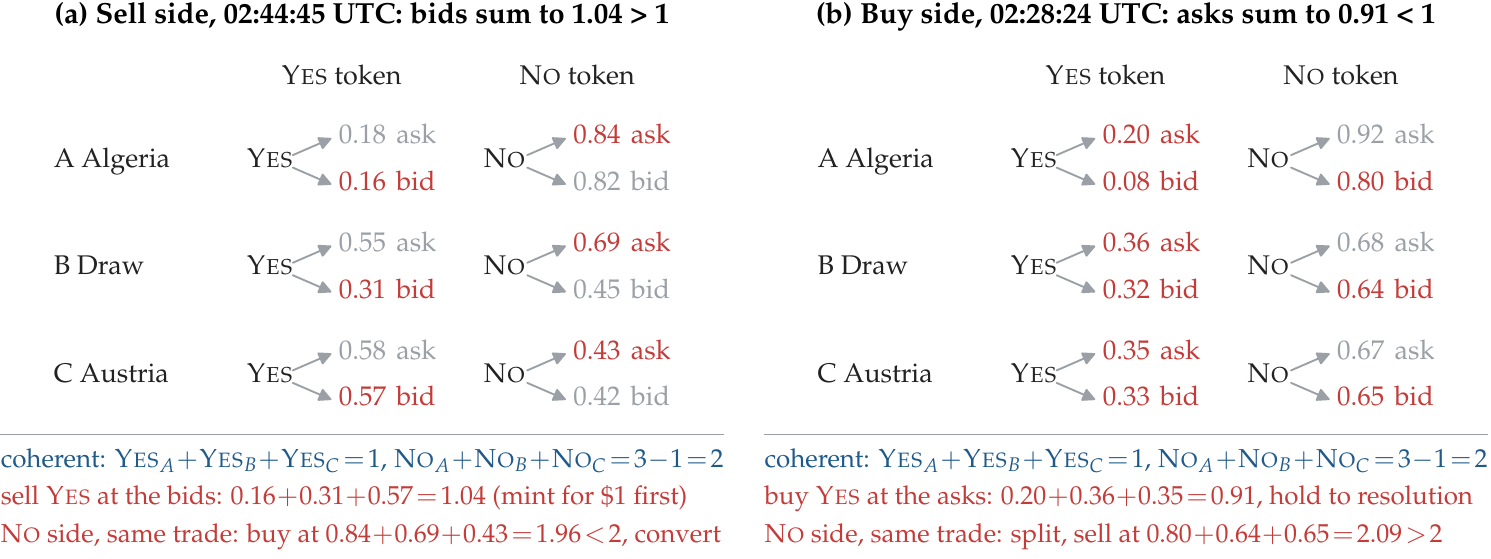}
\caption{The quotes (fees excluded) of three contracts at two moments of the Algeria--Austria 
World Cup match on 28 June 2026: panel (a) at 02:44:45 UTC and panel (b) at 02:28:24 UTC.
Each contract's \YES{} and \NO{} tokens branch into a bid (sell) and an ask (buy); a \NO{} bid is one minus the \YES{} ask
and a \NO{} ask is one minus the \YES{} bid.  The prices the exposition on the previous page uses are in red, the others in
gray.  Coherent prices sum to one over the \YES{} tokens and to two over the \NO{} tokens. \linebreak (a) The
bids sum to $1.04$: sell the three \YES{} after minting, or buy the three \NO{} for $1.96$ and
convert.  (b) The asks sum to $0.91$: buy the three \YES{} and hold, or split and sell the three
\NO{} for $2.09$, the same trade.}
\label{fig:partition-quotes-v2}
\vspace{-10pt}
\end{figure}

\begin{figure}[!p]
\centering
\includegraphics[width=\textwidth]{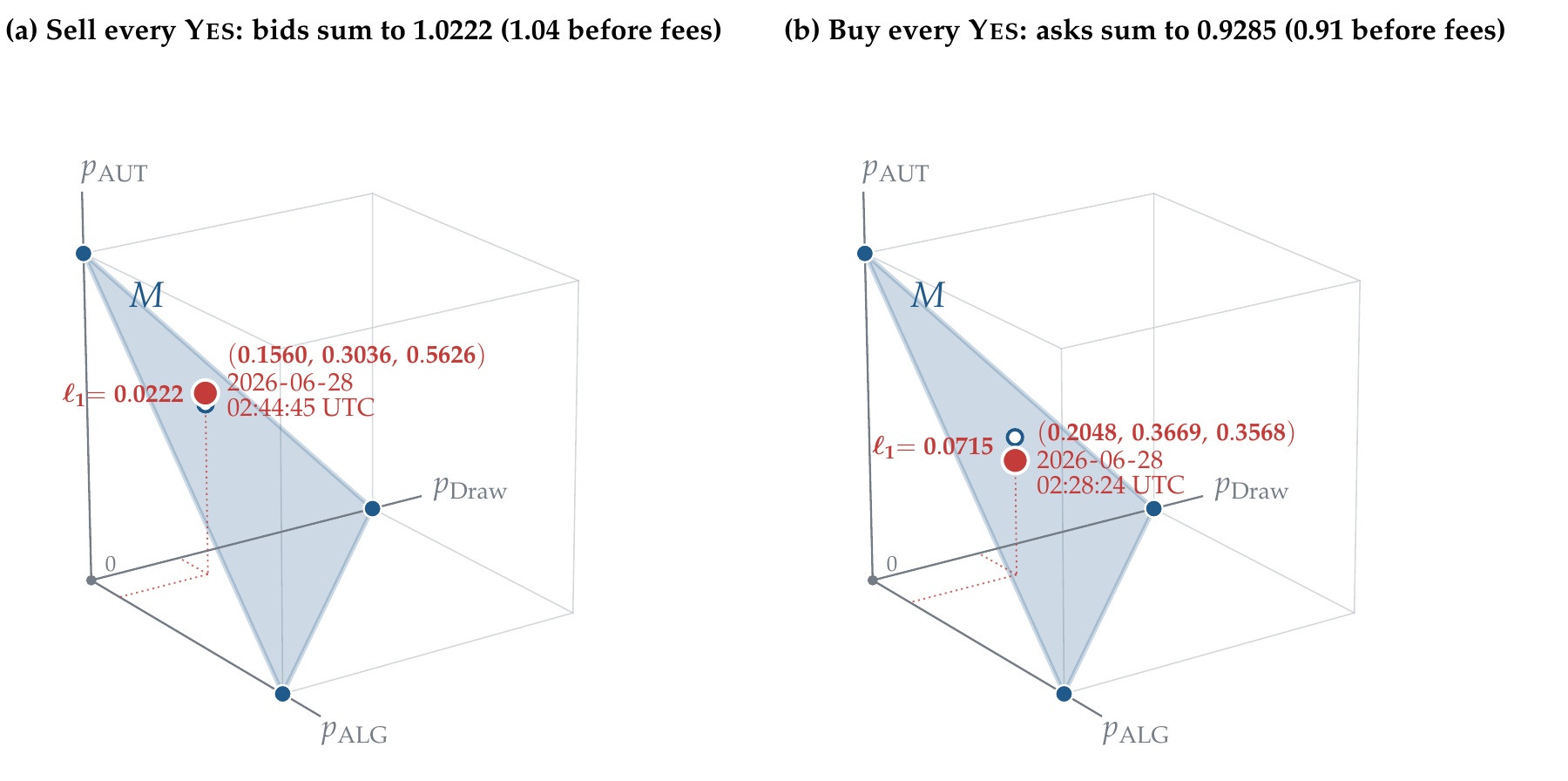}
\caption{Two sides of a partition: the Algeria--Austria match-result
contracts on 28 June 2026, with the coherent set $M$ in blue and the red segment marking the
$\ell_1$ distance.  (a) At 02:44:45 UTC the fee-adjusted bids sum to $1.0222$ ($1.04$ before fees), above $M$: minting
a basket for one dollar and selling every \YES{} earns $0.0222$ per basket at once.  (b) At
02:28:24 UTC the fee-adjusted asks sum to $0.9285$ ($0.91$ before fees), below $M$: buying every \YES{} earns $0.0715$
per basket at resolution.}
\label{fig:partition-sides-v2}
\end{figure}

\FloatBarrier
\section{Data, Relation Extraction, and Layered Detection}
\label{app:data}

\subsection{Data Construction}
\label{app:dataconstruction}

\subsubsection{Metadata and Resolution Criteria}
\label{app:metadata}

We retrieve the market catalog from Polymarket's Gamma API, following its pagination to
collect the event, condition, and outcome-token records.
We retain the raw responses before constructing \linebreak the dataset.
The catalog supplies the event hierarchy, condition and outcome-token identifiers, trading status,
fee parameters, negative-risk identifiers, and the displayed resolution text.
Because displayed catalog text can change after listing, it is not by itself a listing-time record
of the rules under which a contract was introduced.

Polymarket's UMA adapters, both regular and negative-risk, emit \texttt{QuestionInitialized} events whose
ancillary data record the resolution question when a condition is initialized,
so the listing-time rule is recoverable on chain through the question and condition identifiers.
The relation audit reads the catalog text preserved at retrieval.
We do not compare that text with the listing-time record, so the checks assess relations under the
displayed rules at retrieval.

\subsubsection{Live Full-Depth Order-Book Collection}
\label{app:collector}

Polymarket does not archive historical full-depth order books, so we collect its live order-book
messages.  Contracts can begin and finish trading between refreshes of the contract list.
To capture these contracts, we first retrieve the list of contracts accepting orders and subscribe to the
messages of every outcome token, spreading the tokens over as many WebSocket connections as needed,
with at most 600 on each to avoid exceeding the token limits per connection that Polymarket's
server accepted.
We then monitor \texttt{new\_market} messages and add newly listed contracts within about 30
seconds.  Every three hours, we refresh the contract list, remove contracts no longer accepting orders,
and rebuild the connections.  A \texttt{book} message reports bid and ask quantities at every
displayed price level; a subsequent \texttt{price\_change} message replaces the quantity at a
specified level.  Reconnecting supplies a new snapshot.
After verifying that updates for the \YES{} and \NO{} tokens are complementary,
we retain the canonical \YES{}-side book described in Section~\ref{sec:institution}.

\subsubsection{Offline Order-Book Reconstruction}
\label{app:reconstruction}

For detection, we initialize each order book from its latest valid snapshot and apply subsequent
price-level updates in exchange-timestamp order.
We apply all messages sharing a timestamp before evaluating the books and retain both synchronized
best quotes and the full-depth messages.  Changes in quotes or contract eligibility determine when
we evaluate each connected component.  Appendix~\ref{app:measurement-details} explains the ordering
and checks for gaps.

For each detected arbitrage window, we reconstruct the full-depth books from its opening to closing, 
starting from the latest valid snapshots within the preceding six hours.
The arbitrageur buys at the asks and sells at the bids, pays the contract-specific taker fee at each
price level, and reoptimizes quantities over the depth remaining after earlier executions.
The arbitrageur cannot reuse consumed depth, even if an unchanged snapshot still displays it.
These books supply the execution-cost function in Section~\ref{sec:executable}.

\subsubsection{On-Chain Fills}
\label{app:onchain-fills}

We decode every \texttt{OrderFilled} event of the CTF Exchange V2 and Neg Risk CTF Exchange V2
smart contracts over the sample period into outcome token, side, price, quantity, fee,
maker and taker identifiers, transaction hash and log index, and block timestamp.
The token identifiers link each fill to the catalog and the reconstructed order books,
and shared transaction hashes identify fills settled in the same transaction.
Recorded maker and taker fields identify addresses rather than ultimate economic owners;
Section~\ref{sec:takers} matches fills to the detected portfolios by contract, direction, and time. \looseness=-1

\subsubsection{Contract Resolutions}
\label{app:at-contract-resolutions}

We retrieve contract resolution outcomes from Polymarket's Gamma API and query its CLOB API
when a confirmed outcome is unavailable. We identify payoffs from explicit resolution or
winner records. We save the outcomes and retrieval timestamps in a frozen snapshot for the
settlement analysis in Section~\ref{sec:implementation} and Appendix~\ref{app:measurement-details}.
This collection takes place after the sample period ends and arbitrage detection is complete.

\subsection{LLM Relation Extraction}
\label{sec:prompt}

\subsubsection{Candidate Retrieval}
\label{app:candidate-retrieval}

We represent each contract by its event title, question, and group title.
Words and adjacent word pairs receive inverse-document-frequency weights,
and we rank contract pairs by cosine similarity or weighted containment.
The search runs separately within four broad topic groups.
Contracts with the venue's typed sports fields are handled by venue-derived relations;
other contracts concerning sports remain eligible for semantic inference.
For each contract, we keep up to eight candidates from other events and three from the same event
above a fixed similarity threshold.  We exclude pairs within the same negative-risk group and pairs
with nonoverlapping trading intervals.  Candidate groups contain between three and twelve contracts,
drawn either from all contracts listed under an event or from contracts linked by the pairwise
search.

\subsubsection{Prompt Templates}
\label{app:prompt-templates}

The two prompts beginning on the next page are reproduced verbatim.
The first asks the model to identify pairwise relations,
and the second to identify multivariate relations;
both treat the model's output as a proposal requiring further review.
The audits described in Section~\ref{sec:identify} and detailed in Section~\ref{app:llm-examples}
use the catalog text preserved at retrieval. 

\newpage
\begin{Verbatim}[
  fontsize=\scriptsize,
  frame=single,
  framesep=4pt,
  vspace=0pt,
  numbers=left
]
SYSTEM_PROMPT = """You extract candidate logical relations between prediction-market YES outcomes.
Judge from the supplied as-of-retrieval question and resolution text only. A relation must hold in every
possible world allowed by the rules; matching topics or one realized outcome is not enough.

Labels:
- equivalent: A and B resolve Yes in exactly the same worlds.
- a_implies_b: whenever A resolves Yes, B must resolve Yes, but not necessarily conversely.
- b_implies_a: whenever B resolves Yes, A must resolve Yes, but not necessarily conversely.
- disjoint: A and B cannot both resolve Yes.
- overlap_uncertain: related outcomes, but no listed relation is certified by the supplied rules.
- unrelated: not the same underlying proposition or no meaningful logical connection.

Be conservative about deadline, measurement source, jurisdiction, office/control versus election win,
temporary versus permanent events, and market-creation lower bounds. Evidence strings must be exact
substrings copied from the supplied A or B fields. Model output is only a candidate for manual audit."""
\end{Verbatim}

\begin{Verbatim}[
  fontsize=\scriptsize,
  frame=single,
  framesep=4pt,
  vspace=0pt,
  numbers=left
]
SYSTEM_PROMPT = """You extract candidate multivariate logical relations among prediction-market YES outcomes.
Judge only from the supplied as-of-retrieval question and resolution text. A relation must hold in every
possible world allowed by those rules. The result is a candidate for listing-time manual audit, not authority.

Return zero or more irreducibly group-level relations:
- at_most_one: no two listed members can both resolve Yes; sum(members) <= 1.
- at_least_one: at least one listed member must resolve Yes; sum(members) >= 1.
- exactly_one: exactly one listed member resolves Yes; sum(members) = 1.
- complete_partition_of_parent: mutually exclusive children completely partition parent;
  sum(members) = parent.
- incomplete_partition_of_parent: mutually exclusive children are only a subset of parent;
  sum(members) <= parent. Do not use merely because each child implies parent unless mutual
  exclusivity is also certified.
- union_equivalent: parent resolves Yes iff at least one member resolves Yes; members may overlap.
- intersection_equivalent: parent resolves Yes iff every member resolves Yes.

Do not emit ordinary pairwise implication/equivalence/disjoint; the pairwise branch handles those.
Every condition id must come from the supplied group. Parent must be empty for at_most_one,
at_least_one, and exactly_one; otherwise it must be a supplied condition not repeated among members.
Use the smallest sufficient member set. Missing markets matter: when group_complete_view is false,
do not infer exhaustiveness merely because the displayed list looks complete. An exhaustive relation
may still be emitted if the supplied rule text itself proves it. Evidence quotes must be exact
substrings from each involved market's supplied fields. Be conservative about deadlines, measurement
sources, fallback outcomes, other/none categories, cancellation rules, and changed resolution rules."""
\end{Verbatim}

\subsubsection{Illustrative Model Outputs and Audit Responses}
\label{app:llm-examples}

The following three examples show how the proposal models classify relations and how the audit models
assess them.  We reproduce here the substantive fields of their pairwise proposal outputs and audit verdicts.
The examples cover sports, politics, and finance and include both retained and rejected proposals.
In each case, the two audit models independently read the resolution criteria and assess the
proposed relation.  The proposal model's label, confidence, and explanation do not serve as
evidence.
All three examples apply the same rule: we retain a proposal only when both audit models
independently support the same relation from the contract criteria.
A shared topic does not establish a logical relation, whereas an admissible counterexample suffices to reject it.

\newpage
\paragraph{Sports: Qwen3-235B-A22B (retained).}
The input pairs contract A, ``Will the announcers say `Goal' 20+ times during the Czechia vs Mexico
FIFA World Cup Match?,'' with contract B, the otherwise identical 40+ threshold.\vphantom{g}
\begin{Verbatim}[
  fontsize=\scriptsize,
  frame=single,
  framesep=4pt,
  vspace=-3pt,
  breaklines,
  numbers=left
]
relation_type: b_implies_a
confidence: high
a_evidence: ["Will the announcers say 'Goal' 20+ times during the Czechia vs Mexico FIFA World Cup Match?"]
b_evidence: ["Will the announcers say 'Goal' 40+ times during the Czechia vs Mexico FIFA World Cup Match?"]
rationale: "Both markets refer to the same event (Czechia vs Mexico World Cup match on June 24, 2026) and the same term ('Goal'). If the announcers say 'Goal' 40+ times (B), then they must have said it 20+ times (A). So B implies A. However, A does not imply B (could be 25 times). Therefore, the correct relation is b_implies_a."
\end{Verbatim}
\vspace{-5pt}
The final audit outcome is as follows.
\auditbox{
\textbf{Claude Fable 5.1 audit:}
\texttt{b\_implies\_a}, high confidence.\par
\textbf{GPT-5.6 Sol audit:}
\texttt{b\_implies\_a}, high confidence.\par
\textbf{Decision:}
Both verify the same event, term, broadcast, and measurement window;
the 40+ threshold is nested within the 20+ threshold.  Retained.
}

\vspace{.75em}

\paragraph{Politics: Llama-4-Maverick (rejected).}
The input pairs contract A, ``Will Donald Trump visit Virginia in 2026?,'' with contract B,
``Will Donald Trump visit Michigan in 2026?''\vphantom{g}
\begin{Verbatim}[
  fontsize=\scriptsize,
  frame=single,
  framesep=4pt,
  vspace=-3pt,
  breaklines,
  numbers=left
]
relation_type: disjoint
confidence: high
a_evidence: ["Virginia", "2026"]
b_evidence: ["Michigan", "2026"]
rationale: "The markets are about visiting different states."
\end{Verbatim}
\vspace{-5pt}
The final audit outcome is as follows.
\auditbox{
\textbf{Claude Fable 5.1 audit:}
No relation, high confidence.\par
\textbf{GPT-5.6 Sol audit:}
No relation, high confidence.\par
\textbf{Decision:} Visiting one state does not preclude visiting
the other; $A=B=1$ is feasible.  Rejected.
}

\vspace{.75em}

\paragraph{Finance: GLM-5.2 (retained).}
The input pairs contract A, ``Will Databricks' valuation hit (HIGH) \$160B by December 31?,''
with contract B, ``Will Databricks' valuation hit (HIGH) \$190B by July 31?'' The output shown is
the third independent GLM-5.2 evaluation of this pair.\vphantom{g}
\begin{Verbatim}[
  fontsize=\scriptsize,
  frame=single,
  framesep=4pt,
  vspace=-3pt,
  breaklines,
  numbers=left
]
relation_type: b_implies_a
confidence: high
a_evidence: ["Will Databricks' valuation hit (HIGH) $160B by December 31?"]
b_evidence: ["Will Databricks' valuation hit (HIGH) $190B by July 31?"]
rationale: "If Databricks' valuation reaches $190B by July 31, it necessarily reaches or exceeds $160B by December 31 (since July 31 is before December 31). The converse does not hold."
\end{Verbatim}
\vspace{-5pt}
The final audit outcome is as follows.
\auditbox{
\textbf{Claude Fable 5.1 audit:}
\texttt{b\_implies\_a}, high confidence.\par
\textbf{GPT-6 Astra audit:}
\texttt{b\_implies\_a}, high confidence.\par
\textbf{Decision:}
Both verify the same Nasdaq Private Market series, a higher threshold for B,
and a narrower window for B.  Retained.
}
\clearpage

\subsection{Detecting in Three Layers}
\label{app:three-layers}

\begin{table}[!h]
\vspace{-12.5pt}
\centering
\caption{Arbitrage detection by relation layer}
\label{tab:layer-detection-v3}
\begin{threeparttable}
\footnotesize
\begin{tabular*}{\textwidth}{@{\extracolsep{\fill}}
                             >{\raggedright\arraybackslash}p{2.7cm}
                             >{\raggedright\arraybackslash}p{4.5cm}
                             >{\raggedright\arraybackslash}p{4.8cm}
                             >{\raggedright\arraybackslash}p{2.55cm}@{}}
\toprule
Layer & Admissible payoff set & What detection solves & Size of the system \\
\midrule
Venue-asserted
& Group constraints $\sum_{i\in G}z_i\leq1$, with equality for the 20{,}118 groups treated as exhaustive
& Sell when $\sum_{i\in G}\tilde b_i>1$ over the quoted members; buy when $\sum_{i\in G}\tilde a_i<1$ only for exhaustive groups with every member quoted
& 20{,}866 groups \\
\addlinespace[7.5pt]
Venue-derived
& Binary payoff vectors satisfying constraints derived from the common score state of a match
& Enumerate feasible payoff vectors within active blocks and maximize guaranteed profit at best quotes, subject to the leg cap
& 3{,}291 matches \\
\addlinespace[7.5pt]
Semantic-inferred
& $\Phi=\{\bm{z}\in\{0,1\}^{I}:F\bm{z}\leq\bm{g}\}$, with rows from parsed and model-inferred relations
& Enumerate feasible payoff vectors within active blocks and maximize guaranteed profit at best quotes, subject to the leg cap
& 118{,}898 rows, 61{,}349 contracts, 3{,}591 components \\
\bottomrule
\end{tabular*}
\begin{tablenotes}[flushleft]\footnotesize
\item\!\textit{Notes:} $\tilde b_i$ and $\tilde a_i$ are fee-adjusted best bids and asks, as in
Table~\ref{tab:relation-templates}.
\end{tablenotes}
\end{threeparttable}
\vspace{-5pt}
\end{table}

\paragraph{What each detector solves.}
We detect arbitrage separately within each relation layer.
Each layer supplies its own payoff constraints, which determine the worst-case payoff of a proposed
portfolio.  Table~\ref{tab:layer-detection-v3} summarizes these constraints and how we use them in
detection.

\paragraph{Venue-asserted and venue-derived relations.}
A negative-risk group of $n$ contracts admits $n+1$ payoff vectors:
$n$ where exactly one contract resolves \YES{}, and one where all contracts resolve \NO{}. \linebreak
For the 20{,}118 groups compiled as exhaustive partitions,
the all-zero vector is no longer admissible.
These constraints yield the bid and ask sum tests in Table~\ref{tab:layer-detection-v3}.
Selling uses whichever members are quoted,
whereas buying the basket requires a quote on every member of an exhaustive group. \linebreak
For venue-derived relations, we first map contracts on the same match to an $11 \times 11$ score grid
of $121$ states.
We use these states to derive exclusivity, partition, implication, 
and conjunction constraints. \linebreak Detection then minimizes over the binary payoff vectors admitted by
the compiled constraints, using the same constraints in the full-depth calculation.
This set can include payoff vectors that no score state produces, so its minimum payoff can be lower
than the minimum over the score grid.

\enlargethispage{.25\baselineskip}

\paragraph{Compiling the semantic-inferred layer.}
Semantic-inferred relations link contracts across events without a common observed variable,
such as a final score, that determines all their payoffs.
The compiled rows describe a payoff set in $\{0,1\}^{61{,}349}$ without requiring us to list every
vector.  Algorithm~\ref{alg:detector} provides a general solution by adding payoff vectors as
needed.  Our implementation instead enumerates the feasible vectors within each active block and
reuses them across quote states while the constraints remain unchanged.
For a given block and constraint system, this gives the same worst-case payoff as the \textsc{Oracle} step.
We do not evaluate active blocks containing more than 16 contracts in the semantic-inferred layer;
the corresponding caps are 32 for venue-asserted and 48 for venue-derived relations.
Doubling the semantic-inferred cap to 32 contracts changes detection little: \linebreak
the detector finds 160{,}891 arbitrage windows instead of 158{,}063, and 157{,}750 windows are
identical under both caps.  The 5-leg cap separately restricts venue-derived and semantic-inferred
portfolios; sales in venue-asserted groups are exempt, while purchases are not, as described in Section~\ref{sec:implementation}.

\paragraph{Three separate systems.}
We keep the three systems separate so that a portfolio's guaranteed profit remains attributable to
the relations that support it.  Combining them would allow portfolios to rely on constraints from
several sources.  As Section~\ref{sec:classification-errors} shows, an invalid constraint can
create a false guarantee, whereas an omitted valid constraint can only reduce the calculated
guaranteed profit.  We therefore assess each layer using its own evidence:
the venue's collateral mechanism and group rules, typed metadata, or contract text.

Some semantic-inferred relations repeat or overlap with venue-asserted (Layer~1) relations.
Among the 118{,}898 semantic-inferred rows, 133 involve only contracts within one
negative-risk group.  Of these 133 rows,
26 encode 13 partitions that coincide with partitions also imposed in Layer~1,
106 exclusivity constraints are strictly weaker than the corresponding constraints in Layer~1,
and the remaining row imposes an implication between members of the
``\href{https://polymarket.com/event/%
will-mrbeast-hit-billion-views-by-july-31-20260622172639870}%
{MrBeast YouTube Views}'', \linebreak whose cumulative-threshold wording
conflicts with its negative-risk designation.
On 23 June 2026, Polymarket announced immediate resolution to ``134 billion,''
ahead of the July~31 deadline, and refunds for losing shares held
at the time of the announcement.

Another 27{,}276 rows, or 22.9\% of the semantic-inferred system, include at least one negative-risk group member:
11{,}453 connect two or more groups, and 15{,}823 connect a group member with an ungrouped contract.
Keeping the systems separate excludes portfolios that require combining these relations with a negative-risk
group's mutual-exclusivity constraint.  No row includes a contract in Layer~2,
since those contracts are handled by score mappings and excluded from semantic inference. 

Following Section~\ref{sec:empirical}, we simulate one arbitrageur trading across all three relation layers.
The arbitrageur fills displayed depth only once and records consumed quantities in a shared
liquidity ledger.  This depth is unavailable for later trades, including those in other windows,
layers, or days.  If the same portfolio becomes profitable again, the arbitrageur can trade only
against depth not already consumed.  In the baseline, the arbitrageur can trade a quote re-posted at
a previously consumed price as new depth; Section~\ref{sec:incidence} examines this assumption.
The 491{,}827 detected windows contain 39{,}254 distinct portfolios, of which 26 appear in more than
one layer.  These portfolios account for 1{,}461 windows, including 350 executed windows with a total
first take of \$340.

\paragraph{Corrective excision.}
An audit on September 22, 2026 identified six parser-generated relations that incorrectly imposed
exhaustiveness.  The contracts permitted cases where all members resolve~\NO{}, 
including no offering, no prison time, and a rate increase between two contracts' cutoff dates. \linebreak
As an ex-post correction to the completed replay, we relax the affected equality constraints to
impose only mutual exclusivity and recalculate the guaranteed profit of the affected windows. 
We excise 112 detected windows and 97 executed windows
having negative guaranteed profits under the corrected constraints.
We also remove the corresponding first take and accumulated take, \linebreak
which reduces the first take by \$2{,}201 and the accumulated take by \$17{,}404,
all within the semantic-inferred layer.  Of these 97 executed windows, 82 had settled at a profit
(but these 82 settled windows were excised anyway),
illustrating why profitable settlement does not validate an ex-ante guarantee. 
We do not rerun the ledger or make depth consumed by excised trades available to other
trades. 
\clearpage

\subsection{From Raw Messages to Settlement: Measurement Details}
\label{app:measurement-details}

\noindent This section expands Section~\ref{sec:implementation}.
Figure~\ref{fig:measurement-pipeline} and Table~\ref{tab:measurement-validation}
show inputs, outputs, and key statistics. 

\begin{figure}[!h]
\vspace{5pt}
\centering
\includegraphics[width=\textwidth]{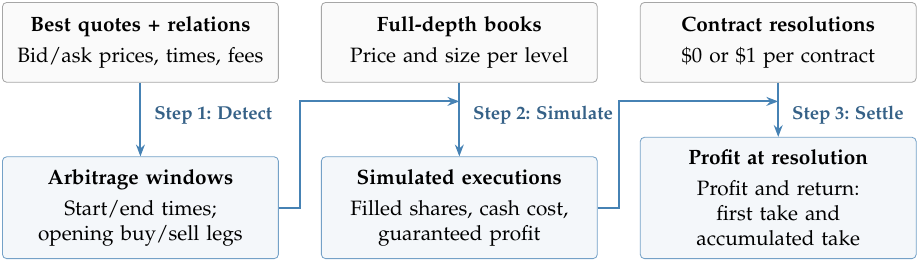}
\caption{Inputs (top) and outputs (bottom) of the three measurement steps.
Within each window, the arbitrageur trades against the displayed depth;
contract resolutions then determine profit.  Best quotes and full-depth books use the same raw
messages.}
\label{fig:measurement-pipeline}
\vspace{-3pt}
\end{figure}

\paragraph{Preparation: order-book reconstruction and quote checks.}
We reconstruct best quotes from raw snap\-shots and price-level updates.
We store the messages in a common format, preserving exchange and receipt timestamps, source order,
token identifiers, and quantities without sampling.  We also retain the raw archive.
We order messages by exchange timestamp, including delayed messages received in later files,
to align quotes across connections.  Messages with the same timestamp retain their source-file, row,
and within-message order.  We remove exact duplicates and apply all messages sharing a timestamp
before evaluating the books, so partial updates do not create artificial windows. \looseness=-1

A snapshot replaces both sides of the book.  A price-level update reports the new resting quantity,
not a change in quantity; zero removes the level.
We evaluate a book only after receiving a snapshot.
After applying each timestamp batch, we compare the reconstructed best quotes with those reported in
the final price-level update.  If they disagree, we exclude the book until agreement resumes or a
new snapshot restores it.  We record these gaps and separately count missing reported quotes.
This check verifies best quotes, but cannot establish the completeness of depth at other price
levels. \linebreak Both detection and full-depth reconstruction use the same exchange timestamps.

\paragraph{Step 1: detect arbitrage windows from best quotes.}
Detection uses quoted contracts and valid constraints for which all required contracts are
available.  Active blocks exceeding the respective layer caps of 32, 48,
and 16 contracts are recorded and excluded from evaluation.
Portfolios are capped at 5 legs, except for sales in venue-asserted groups
(see also Section~\ref{sec:implementation} and Appendix~\ref{app:partition-sides}).
We evaluate every relevant update, including changes in fees,
trading eligibility, and the validity of relations.
A previous solution is reused only when the objective, constraints, variable domains, bounds,
and solver settings are unchanged; quantities or fees can change the problem even when prices do not
move.  We continue checking book and window states when reusing a solution and retry failed
optimizations.  For each window, we save the opening portfolio, supporting relations,
profit and capital measures, and the half-open interval $[t_{\mathrm{open}},t_{\mathrm{close}})$.

We retain previously discovered or enumerated payoff vectors as \textsc{Master} constraints
across quote states as long as the constraint system remains unchanged,
since \(\Phi\) depends only on \((F,\bm g)\).
This caching can reduce the number of \textsc{Oracle} calls in Algorithm~\ref{alg:detector} or avoid repeated enumeration.
It should be possible to adapt our detector for use with a live quote feed rather than the replay: 
while the logical relations can be recovered in advance,
to keep pace with the feed would require handling delayed and out-of-order messages 
as well as solving the arbitrage problem fast enough.
We do not assess real-time performance.

\paragraph{Step 2: simulate arbitrage trades against full-depth order books.}
Within each detected window, the arbitrageur trades the opening portfolio's legs in their original
directions and reoptimizes quantities at each eligible update.
We retain the window's original opening and closing times, including when the arbitrageur has
exhausted its displayed depth.  We initialize each required book from the latest valid snapshot
within the preceding six hours.  Without such a snapshot, or during a recorded gap,
the arbitrageur cannot trade that contract.  The arbitrageur buys at the asks and sells at the bids,
paying the applicable taker fee at each price level.
If a leg is unavailable or depleted, its quantity is set to zero;
the arbitrageur may still trade a profitable feasible subset.

The arbitrageur records consumed depth by token, book side, and price in a shared liquidity ledger
across days, windows, and relation layers.  Let $O_{\ell,t}$ be the quantity displayed at price
level~$\ell$ after the update at time $t$, $O_{\ell,t^-}$ the quantity displayed just before that update
(zero for a level not seen before), $R_{\ell,t^-}$ the quantity remaining after previous executions,
and $q_{\ell,t}$ the quantity the arbitrageur trades at time $t$.
After an order-book update, the arbitrageur can trade the available quantity
$R^{\mathrm{available}}_{\ell,t}$; after execution, the remaining quantity is $R_{\ell,t}$:
\vspace{-2.5pt}\[\vspace{-2.5pt}
 R^{\mathrm{available}}_{\ell,t} =\max\{0,R_{\ell,t^-}+O_{\ell,t}-O_{\ell,t^-}\},
 \qquad R_{\ell,t}=R^{\mathrm{available}}_{\ell,t}-q_{\ell,t}.
\]
For example, an arbitrageur who has filled 100 displayed shares cannot fill them again while the
displayed quantity at that level stays at 100, since the ledger compares quantities, not orders.
If displayed depth rises to 130, the arbitrageur can fill 30 additional shares.
When opportunities overlap, the arbitrageur trades first in the window with the
highest guaranteed profit at displayed depth, breaking ties by window identifier.
Each trade reduces the depth available for subsequent trades.
This sequential allocation need not maximize profit jointly across windows.
We process windows sharing liquidity together in exchange-time order, preserving the allocation when
calculations run in parallel.

\enlargethispage*{.75\baselineskip}

The first take is the guaranteed profit from the arbitrageur's first execution,
which may occur after the window opens; the accumulated take sums guaranteed profit from all
executions in the window.  Windows without an execution remain in the count of detected windows.
Price-level data do not identify individual orders, so cancellation followed by re-posting can
supply new depth in the baseline.  The robustness checks vary the treatment of re-posted depth,
the participation rate, and order delay separately.
With order delay, the arbitrageur waits the specified delay after the relevant quote
information becomes available and then trades against the order book at execution time.
For every execution, we save the quantities filled at each price level, cash outlay, fees,
guaranteed profit, and committed capital.
\clearpage

\paragraph{Step 3: test arbitrage windows at resolution.}
We match the arbitrageur's simulated executions to contract resolutions.
Before calculating profit, we check that all outputs are present without duplication,
use the same inputs and settings, and preserve every detected window.
We also verify that execution-level amounts sum to window totals.
We assign resolution payoffs only to legs with positive filled quantities.
The saved resolution data link each contract to the token listed first in the catalog,
which we use as its canonical \YES{} side even when its displayed name differs.
We use explicit resolution or winner records to determine payoffs, not a closed market's last price.
Missing, unresolved, and void outcomes remain distinct and are never replaced by zero.
A window is settled only when every filled leg of every execution has a binary payoff.
We also record whether the first execution's legs have settled.

For each settled execution, profit at resolution is
$\bm y^{\top}\bm z^{\mathrm{resolved}}-K(\bm y)$, using the same filled shares, \linebreak cash flows,
and fees that produced its guaranteed profit.
Comparisons of the first take and the accumulated take use the same settled windows and the capital
committed to their respective executions.  For venue-asserted negative-risk portfolios,
the arbitrageur needs one dollar of collateral for each minted basket;
the number of baskets equals the largest quantity sold on any leg.
For other portfolios, the arbitrageur needs one dollar for each share sold on each leg.
Both measures also include cash paid for purchases.
We compare profit at resolution with guaranteed profit and record any shortfall.
Section~\ref{sec:takers} separately compares the simulated portfolios with actual trades.

\begin{table}[!t]
\centering
\caption{Data reconstruction, execution, and settlement}
\label{tab:measurement-validation}
\begin{threeparttable}
\footnotesize
\renewcommand{\arraystretch}{0.975}
\begin{tabular*}{\textwidth}{@{\extracolsep{\fill}}lr@{}}
\toprule
Statistic & Result \\
\midrule
\multicolumn{2}{@{}l}{\textit{Step 1: reconstruction and detection}} \\
Order-book messages (billions) & 73.22 \\
Reconstructed best-quote comparisons (billions) & 13.3 \\
Agreement with venue-reported best quotes & 97.2\% \\
Arbitrage windows detected & 491{,}827 \\
\midrule
\multicolumn{2}{@{}l}{\textit{Step 2: simulated arbitrage trades}} \\
Percentage of detected windows with an execution & 25.4\% \\
Executions & 462{,}151 \\
Percentage of executed windows first filled after opening & 12.2\% \\
Opening-to-first-execution delay, 99th percentile (seconds) & 231.9 \\
\midrule
\multicolumn{2}{@{}l}{\textit{Step 3: settlement}} \\
Percentage of executed windows settled & 89.8\% \\
Executions in settled windows & 391{,}383 \\
Settled windows earning less than guaranteed profit & 0 \\
\bottomrule
\end{tabular*}
\begin{tablenotes}[flushleft]
\footnotesize
\item\!\textit{Notes:} First-execution delays use exchange timestamps and cover executed windows. The final row checks both
the first take and the accumulated take.
\end{tablenotes}
\end{threeparttable}
\vspace{-7.5pt}
\end{table}

\enlargethispage*{.5\baselineskip}

\paragraph{Reproducibility checks.}
For each calculation, we record the code, software and solver versions, input checksums, settings,
and output counts.  The saved resolution data allow us to repeat the analysis at the same cutoff
date.  We check that raw and reformatted messages produce the same books,
that reusing unchanged optimization problems preserves the detected windows,
and that the arbitrageur cannot reuse depth across days or equivalent token representations.
We further check that delayed first executions are recorded correctly and
ensure that windows with unresolved or void filled legs are excluded from settlement comparisons.

\end{document}